\documentclass[a4paper,UKenglish]{lipics-v2016}
\usepackage{microtype}%if unwanted, comment out or use option "draft"
\usepackage{stmaryrd} %needed for \llbracket and \rrbracket

\newcommand{\equals}{\stackrel{\mathrm{def}}{=}}

\title{Mixed-criticality Scheduling with Dynamic Redistribution of Shared Cache\footnote{This work was partially supported by National Funds through FCT/MEC (Portuguese Foundation for Science and Technology) and co-financed by ERDF (European Regional Development Fund) under the PT2020 Partnership, within the CISTER Research Unit (CEC/04234); also by by FCT/MEC and the EU ARTEMIS JU within project ARTEMIS/0001/2013- JU grant nr. 621429 (EMC2).}}
\titlerunning{Mixed-criticality Scheduling with Dynamic Redistribution of Shared Cache} %optional, in case that the title is too long; the running title should fit into the top page column

\author[1]{Muhammad Ali Awan}
\author[2]{Konstantinos Bletsas}
\author[3]{Pedro F. Souto}
\author[4]{Benny \r{A}kesson}
\author[5]{Eduardo Tovar}
\affil[1]{CISTER Research Centre and ISEP, Porto Portugal\\
  \texttt{muaan@isep.ipp.pt}}
\affil[2]{CISTER Research Centre and ISEP, Porto Portugal\\
  \texttt{ksbs@isep.ipp.pt}}
\affil[3]{University of Porto, Faculty of Engineering and CISTER Research Centre and ISEP, Porto Portugal\\
  \texttt{pfs@fe.up.pt}}
\affil[4]{CISTER Research Centre and ISEP, Porto Portugal\\
  \texttt{kbake@isep.ipp.pt}}
\affil[5]{CISTER Research Centre and ISEP, Porto Portugal\\
  \texttt{emt@isep.ipp.pt}}
  
\authorrunning{M.\,A. Awan, K. Bletsas, P.\,F. Souto, B. \r{A}kesson, and E. Tovar} %mandatory. First: Use abbreviated first/middle names. Second (only in severe cases): Use first author plus 'et. al.'

\Copyright{Muhammad Ali Awan, Konstantinos Bletsas, Pedro F. Souto, Benny \r{A}kesson, and Eduardo Tovar}%mandatory, please use full first names. LIPIcs license is "CC-BY";  http://creativecommons.org/licenses/by/3.0/

\subjclass{C.3 Real-time and embedded systems}%  -- please refer to \url{http://www.acm.org/about/class/ccs98-html} mandatory: Please choose ACM 1998 classifications from http://www.acm.org/about/class/ccs98-html . E.g., cite as "F.1.1 Models of Computation". 
\keywords{Mixed Criticality Scheduling, Vestal Model, Dynamic Redistribution of Shared Cache, Shared Last-level Cache Analysis, Cache-aware Scheduling}% mandatory: Please provide 1-5 keywords
  \EventEditors{Marko Bertogna}
  \EventNoEds{1}
  \EventLongTitle{29th Euromicro Conference on Real-Time Systems (ECRTS 2017)}
  \EventShortTitle{ECRTS 2017}
  \EventAcronym{ECRTS}
  \EventYear{2017}
  \EventDate{June 28--30, 2017}
  \EventLocation{Dubrovnik, Croatia}
  \EventLogo{}
  \SeriesVolume{76}
  \ArticleNo{18}
\begin{document}

\maketitle

\begin{abstract}
The design of mixed-criticality systems often involves painful tradeoffs between safety guarantees and performance.
However, the use of more detailed architectural models in the design and analysis of scheduling arrangements for mixed-criticality 
systems can provide greater confidence in the analysis, but also opportunities for better performance. Motivated by this 
view, we propose an extension of Vestal's model for mixed-criticality multicore systems that (i) accounts for the 
per-task partitioning of the last-level cache and (ii) supports the dynamic reassignment,
for better schedulability, of cache portions initially reserved for lower-criticality tasks to the higher-criticality tasks, 
when the system switches to high-criticality mode. To this model, we apply partitioned
EDF scheduling with Ekberg and Yi's deadline-scaling technique. Our schedulability analysis and scalefactor calculation 
is cognisant of the cache resources assigned to each task, by using WCET estimates that take into account these resources. 
It is hence able to leverage the dynamic reconfiguration of the cache partitioning, at mode change, for better performance, 
in terms of provable schedulability. 
We also propose heuristics for partitioning the cache in low- and high-criticality mode, that promote schedulability.
Our experiments with synthetic task sets, indicate tangible improvements in schedulability compared
to a baseline cache-aware arrangement where there is no redistribution of cache resources from low- to high-criticality tasks 
in the event of a mode change.
\end{abstract}

\section{Introduction}
Many real-time embedded systems (automotive, avionics, aerospace) host functions of different criticalities. 
A deadline miss by a high-criticality function can be disastrous, but losing a low-criticality function only moderately affects the quality 
of service. 
Scalability and cost concerns favour mixed-criticality (MC) systems, whereby tasks of different criticalities are scheduled on 
the same core(s). However, this brings challenges. Lower-criticality tasks interfering unpredictably with higher-criticality tasks can be 
catastrophic. Conversely, rigid prioritisation by criticality leads to inefficient processor usage. Therefore, researchers have 
been working on scheduling models and techniques for (i)~efficient use of processing capacity and (ii)~schedulability guarantees for all 
tasks under typical conditions subject to (iii)~ensured schedulability of high-criticality tasks in all cases. 
Most works~\cite{Burns_Davis_13} are based on Vestal's model~\cite{Vestal_07,baruah:burns:ada:europe:2011}, which views the system operation as different modes,
whereby only tasks of a certain criticality or above execute; additionally, different worst-case task execution times (WCETs) are assumed for the same task in each mode that it can be a part of, with corresponding degrees of confidence. This is because the cost of provably safe WCET estimation (and the associated pessimism) is justified only for high-criticality tasks. Other tasks have less rigorous WCET 
estimates, which might be exceeded, very rarely.

Many variants of the Vestal task model have been explored in recent years, with ever more sophisticated scheduling approaches 
and corresponding schedulability analysis techniques being devised for those. Yet, more progress is needed in terms of making the platform 
model more realistic, by incorporating more details about the architecture. The potential benefits could be (i)~more accurate, hence 
safer, schedulability analysis, but also (ii)~improved performance, from scheduling arrangements that acknowledge and leverage those 
architectural details. In particular, one could look for inspiration at efforts from the general-purpose (i.e., non-mixed-criticality) 
real-time systems domain, towards more \textit{cache-aware} scheduling and analysis. Notably, Mancuso et al.~\cite{Mancuso_DBPMM_13}, 
in the context of the Single-Core Equivalence (SCE) framework~\cite{Sha_CMKY_14}, consider 
(i)~a cache-partitioned multicore architecture and (ii)~task WCET estimates that are cognisant of the cache-partitioning. 

Our work is inspired from the SCE framework and specifically seeks to integrate the effects of one particular shared resource,
the last-level cache, into a dual-criticality Vestal model. We assume a last-level cache shared by all cores and partitioned among 
the different tasks via the Coloured Lockdown approach, to mitigate intra- and inter-core interference.
For better resource usage and schedulability, instead of a static cache partitioning, we reclaim the cache pages allocated to 
low-criticality tasks (L-tasks) and redistribute those to high-criticality tasks (H-tasks), 
upon a switch to high-criticality mode (H-mode). In turn, the additional resources afforded to those tasks drive down their
(cache-cognisant) H-mode WCETs. We propose a new mixed-criticality schedulability analysis that takes into account these
effects, allowing for improvements in the guaranteed schedulability of the system. In a summary, these are the main
contributions of our work:
 
\begin{enumerate}
\item 
We integrate the shared platform resources into a mixed-criticality model and dynamically redistribute those resources 
as a part of mixed-criticality scheduling. We demonstrate this principle by applying to the shared last-level cache. 	
\item 
We formulate schedulability analysis for the proposed model, assuming EDF scheduling using Ekberg and Yi's deadline scaling.
Our analysis leverages the fact that cache resources are reclaimed from low-criticality tasks, in the event of a mode change, 
and redistributed to high-criticality tasks. This allows for improved schedulability. 
\item 
We propose a two-staged allocation heuristic for allocating cache resources to the tasks, in the two modes of operation,
and implement it by Integer Linear Programming (ILP).
\end{enumerate}

Our experiments with synthetic task sets indicate appreciable schedulability improvements
over approaches that perform no reclamation of cache-resources at mode change.

This paper is organised as follows. Section~\ref{s:related_work} presents the related work. The system model and the 
assumptions are discussed in Section~\ref{s:system_model}. 
The schedulability analysis for that model is presented in Section~\ref{s:schedulability_analysis},
followed by some proposed heuristics for cache allocation to the tasks in the two modes, in Section~\ref{s:heuristics}.
Section~\ref{s:evaluation} presents and discusses the experiments used to evaluate the performance of the proposed approach. 
Conclusions are drawn in Section~\ref{s:conclusions}. 

\section{Related Work\label{s:related_work}}

Several feasibility tests are known for Vestal-model systems scheduled under, e.g., EDF or Fixed Priorities. 
 One drawback, when using EDF, is that an H-task too close to its deadline, at the moment of a mode change, may be unable to accommodate its outstanding 
execution time (associated with its H-WCET) until its deadline,  leading to a deadline miss. Therefore, the \textit{deadline-scaling} technique 
was conceived~\cite{Baruah_BALMSS_12, Ekberg_Yi_12, Masrur_MW_15, Gu_Easwaran_16}, to avert such scenarios if possible. 
It originated with EDF-VD~\cite{Baruah_BALMSS_12}, 
which uses standard EDF scheduling rules but, instead of reporting the real deadlines to the EDF scheduler for scheduling decisions, it reports 
shorter deadlines (if needed) for H-tasks during L-mode operation. This helps with the schedulability of H-tasks in the case of 
a switch to H-mode, because it prioritises H-tasks more than conventional EDF would, over parts of the schedule. This allows them 
to be sufficiently ``ahead of schedule'' and catch up with their true deadlines if any task overruns its L-WCET. In H-mode, the true 
H-task deadlines are used for scheduling and L-tasks are ``dropped'' (i.e., idled). EDF-VD proportionately shortens the H-task deadlines 
according to a single common scalefactor and its schedulability test considers the task utilisations in both modes. 
Ekberg and Yi~\cite{Ekberg_Yi_12} 
improved upon EDF-VD by enabling and calculating distinct scale factors for different H-tasks and using a more precise demand bound function 
(dbf) based schedulability test~\cite{Baruah_MR_90}, for better performance. The scalefactor calculation is an iterative 
task-by-task process (for details, see~\cite{Ekberg_Yi_12,Ekberg_Yi_14}). 

However, the aforementioned scheduling solutions typically only 
consider the task execution on the processor cores and do not consider other platform resources, such as interconnects,
caches and main memory.
Some other works do consider interference on shared resources and propose mechanisms for
its mitigation, albeit for single-criticality systems. For instance, several software-based approaches 
are proposed for mitigating cache and memory interference in multi-core 
platforms~\cite{Mancuso_DBPMM_13, Yun_YRMS_13, Nowotsch_PBTWS_14, Flodin_LY_14, Behnam_INS_13,Valsan_YF_16}. 
Some of these works integrate the interference on shared resources to the schedulability analysis of the system. 
Mancuso et al.~\cite{Mancuso_PCLH_15} integrate the effect of multiple shared resources (cache, memory bus, DRAM memory) 
on a multicore platform under partitioned fixed-priority preemptive scheduling. Pellizzoni and 
Yun~\cite{Pellizzoni_Yun_16} generalise the arrangement (and the analysis from~\cite{Mancuso_PCLH_15}) 
to uneven memory budgets per core and propose a new analysis for different memory scheduling schemes.
Behnam et al.~\cite{Behnam_INS_13} incorporated the effect of interference on shared resources under server-based hierarchical 
scheduling, that provides isolation between independent applications. 

A software-based memory throttling mechanism for explicitly controlling the memory interference 
under fixed-priority preemptive scheduling is proposed in~\cite{Yun_YPCS_12}, although it only considers 
the timing requirements of tasks on a single critical core, whereupon all critical tasks are scheduled. 
The rest of the cores (interfering cores) are assumed to have non-critical tasks. 
Nevertheless, the analyses in existing works that consider the shared resources in the context of scheduling, assume 
that resources are statically allocated. Our proposed mixed-criticality algorithm considers the dynamic redistribution of shared 
resources, in order to efficiently exploit their availability and improve the schedulability of the system. 
In this work, we demonstrate this principle with one particular resource: the last-level cache.

\section{System model and assumptions \label{s:system_model}}
\subsection{Platform}
We assume a multicore platform composed of $m$ identical cores accessing main memory via a shared memory controller. 
A core can have multiple outstanding (i.e., not served yet) memory requests . Prefetchers and speculative units are disabled. Our assumptions 
about the memory subsystem are inspired by those of the SCE~\cite{Sha_CMKY_14} framework. We assume that all cores share 
a big last-level cache, but have dedicated upper-level caches (closer to the cores). Colored Lockdown~\cite{Mancuso_DBPMM_13} is used, 
to mitigate the intra-/inter-core interference. It allows a task to lock its $\sigma$ most frequently used pages (hot pages) 
in the last-level cache,  which facilitates upper-bounding the number of residual memory accesses (i.e., last-level cache misses)
and, by extension, the WCET as a function of $\sigma$.
In this work, we only analyse the integration and dynamic redistribution of one particular resource (the shared last-level cache) 
into a mixed criticality scheduling theory, as proof of concept, and we genuinely believe that a similar approach can be used 
to integrate other shared resources. The SCE framework also deploys the OS-level memory bandwidth regulation mechanism Memguard~\cite{Yun_YPCS_16} and
 the DRAM-bank partitioning mechanism PALLOC~\cite{Yun_MZP_14} to mitigate the interference on those shared resources. 
In the future, we intend to also  exploit these SCE mechanisms and dynamically redistribute memory access budgets at the mode switch.

\subsection{Task model}

Consider Vestal's base model with two criticality levels, high and low, as a starting point. 
Each task has an associated criticality, low or high. High-criticality tasks (H-tasks) have two WCET estimates: The L-WCET, which 
is \textit{de facto} deemed safe, and the H-WCET, which is provably safe and possibly much greater. For low-criticality tasks (L-tasks), 
only the L-WCET is defined. There are two modes of operation. The system boots and remains in low-criticality mode (L-mode) as long as no job 
(instance of a task) executes for longer than its L-WCET. However, if any job exceeds its L-WCET, then the system immediately switches into 
high-criticality mode (H-mode) and permanently dispenses with the execution of all L-tasks. It is pessimistically assumed that 
in H-mode all jobs by H-tasks (including any existing jobs at the time of the mode switch) may execute for up to their H-WCET. 
Under these assumptions, it must be provable by an offline schedulability test that (i)~no task misses a deadline in L-mode and 
(ii)~no H-task misses a deadline in H-mode. We extend this basic model and assume that both the L-WCET and the H-WCET are 
functions of the number of the task's hottest pages locked in the last-level cache. In detail:   

Our task set consists of $n$ independent sporadic tasks ($\tau \equals \lbrace \tau_1, \tau_2,  \ldots, \tau_n \rbrace$).  
Each task $\tau_i\in \tau$ has a minimum inter-arrival time $T_i$, a relative deadline $D_i$ and a criticality level 
$\kappa_i \in \lbrace L, H \rbrace$ (low or high, respectively). The subsets of low-criticality and high-criticality 
tasks are defined as $\tau(L) \equals \lbrace \tau_i \in \tau | \kappa_i = L \rbrace$ 
and $\tau(H) \equals \lbrace \tau_i \in \tau | \kappa_i = H \rbrace$. We assume constrained deadlines, i.e., $D_i \leq T_i$. 
The original Vestal model is extended based on the following assumptions: 

\begin{itemize}
\item 
The (actual) WCET of a task depends on the number of its pages (selected in order of access frequency) locked in place in the last-level cache. 
\item 
Different \textit{estimates} of that WCET (derived via different techniques), are to be used for the L-mode and H-mode.
\end{itemize}

For each task $\tau_i$, the L-WCET $C_i^L(\cdot)$ and the H-WCET $C_i^H(\cdot)$ are not single values, but rather functions of the pages locked 
in the last-level cache. For example $C^{L}_i(6)$ denotes the L-WCET of $\tau_i$ when this task is configured with its 
$6$ ``hottest'' pages locked in the cache. How the ordered list of hot pages per task is obtained (and its accuracy) is beyond
the scope of this paper and orthogonal to both the WCET estimation techniques and the safety of our analysis, as long as the same
$\sigma$ pages were assumed locked in cache when deriving $C^{L}_i(\sigma)$ and $C^{H}_i(\sigma)$. In practice, the profiling 
framework in~\cite{Mancuso_DBPMM_13} can be used for ranking each tasks's pages by access frequency. 
Estimating the WCET
in isolation, for each task,
assuming that the top $\sigma$ pages in the list are locked in the cache, allows for the construction of a 
\textit{progressive lockdown curve} (WCET vs number of locked pages in last-level cache). 
More locked pages in the last-level cache means fewer last-level cache misses (i.e., fewer residual memory requests) 
and, consequently, also a smaller WCET. 

The technique in~\cite{Mancuso_DBPMM_13} for generating the progressive lockdown curve is measurement-based, 
so its output is not provably safe, but it can serve as the L-WCET progressive lockdown curve $C_i^L(\cdot)$. Moreover, 
some static analysis tools comprehensively cover all possible control flows (or even some infeasible paths) in a task, and these 
can be used to estimate the H-WCETs. By safely modelling accesses to the hot pages locked-in by Colored Lockdown as 
``always hit upon reuse'', 
the static analysis tool can derive tighter WCET estimates than it would without this knowledge -- and the improvement will be 
greater the more pages are locked in the cache. Hence, a progressive lockdown curve similarly exists for the H-WCET $C_i^H(\cdot)$. 

To demonstrate the concept, Fig.~\ref{f:example_pldc} shows (imaginary) H- and L-mode progressive lockdown curves of a task $\tau_i$. 
The $x$ and $y$ axes show the number of locked pages and WCET, respectively. Ideally, these two curves are non-increasing 
functions\footnote{In the general case, the progressive lockdown curves are not necessarily convex, and we make no such assumption 
nor does our approach depend on such a property (convexity).}. 
Let us assume that $\sigma_i^L$ and $\sigma_i^H$ denote the number of pages of a task $\tau_i$ locked in last-level cache 
in L- and H-mode, respectively. Then, the utilisation of a task in the L-mode (H-mode) is defined as 
$U_i^L(\sigma_i^L)\equals \frac{C_i^L(\sigma_i^L)}{T_i}$ 
(resp., $U_i^H(\sigma_i^H) \equals \frac{C_i^H(\sigma_i^H)}{T_i}$). 
In this paper we assume that the L- and H-mode progressive lockdown curves for each task are already provided to us
as input. We also assume fully partitioned scheduling, i.e., no task ever migrates.

In case the overheads of unlocking and locking pages in the cache at mode change would be excessive, one could use per-task cache partitions \textit{without any locking} (i.e., populated with lines dynamically). Techniques like APTA~\cite{apta} could derive the equivalent of a parametric WCET curve as a function of the partition size in the L-mode and the H-mode. However, for simplicity, in the rest of the paper we assume the use of page locking.

\subsection{Impact of mode change upon WCET}
Under our model, a job by an H-task $\tau_i$ released in L-mode has its $\sigma_i^L$ hottest pages in the cache
and a job by the same task released in H-mode has its $\sigma_i^H$ hottest pages in the cache. Both $\sigma_i^L$
and $\sigma_i^H$ are decided at design time (with $\sigma_i^H \geq \sigma_i^L,~\forall \tau_i \in\tau(H)$).
We assume that, as soon as a mode change occurs, the system can reclaim the cache pages hitherto allocated to L-tasks, for
redistribution to the H-tasks. However, it is conservatively assumed that only new jobs by H-tasks, released after the mode 
change, benefit from the additional cache pages (either because it is only opportune to distribute them at the next release, 
or because, in the worst-case, the improvement from additional pages afforded to a job already having started its execution 
might not be quantifiable).
For analysis purposes, we therefore conservatively assume that any H-job caught in the mode change 
may execute for up to $C_i^{H}(\sigma_i^L)$ time units, whereas any subsequent job by the same task only executes for up to 
$C_i^{H}(\sigma_i^H) \leq C_i^{H}(\sigma_i^L)$.

\begin{figure}
\begin{minipage}[b]{1\linewidth}
\centering
\includegraphics[width=0.9\linewidth]{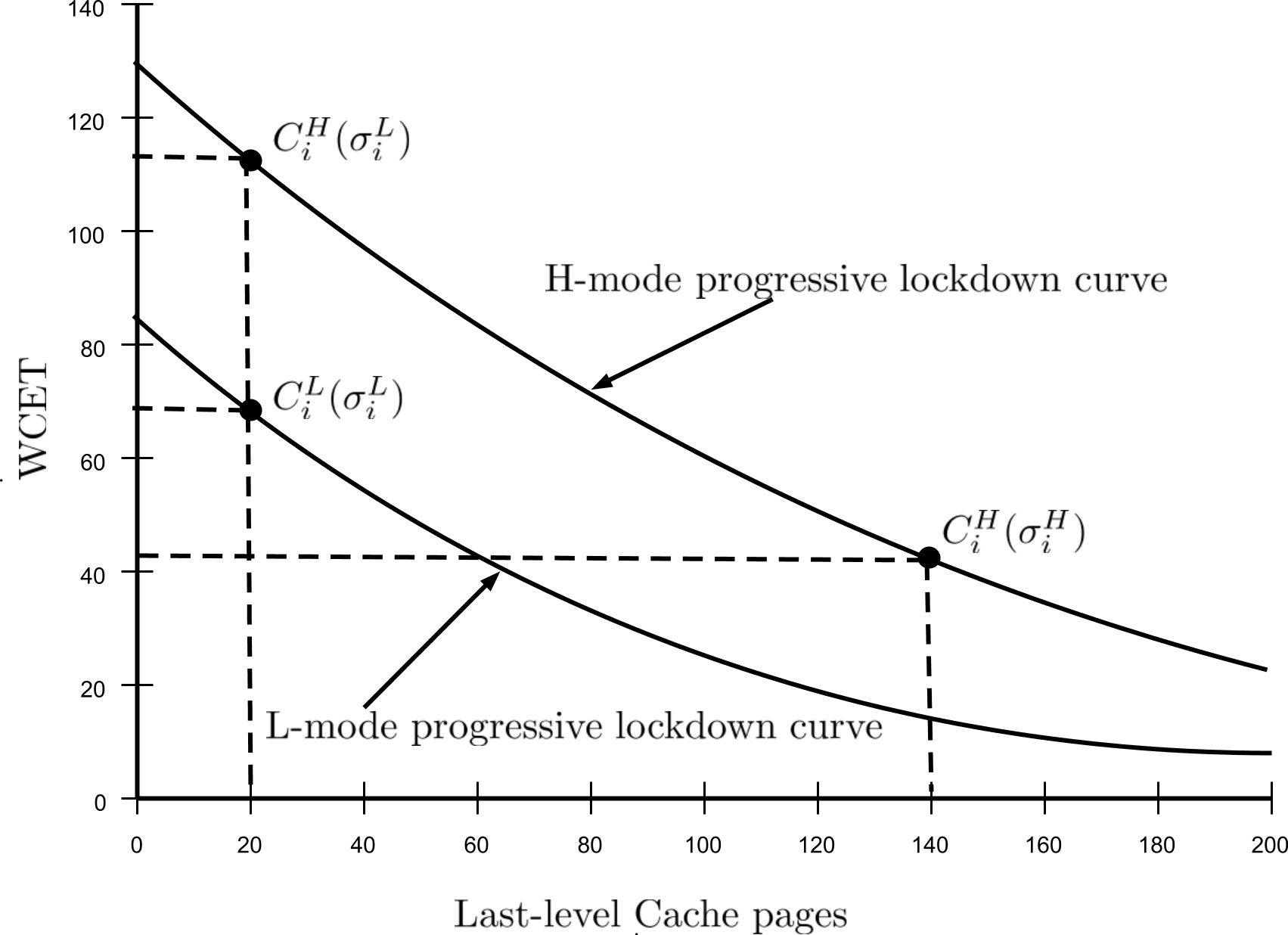}
 \caption{\label{f:example_pldc} H-mode and L-mode progressive lockdown curves.}
\end{minipage}
\end{figure} 

One interesting counter-intuitive property of our model is that there may be cases when 
$C_i^H (\sigma_i^H)\leq C_i^L(\sigma_i^L)$, unlike what holds for the classic Vestal model, where $C_i^H \geq C_i^L$ in all cases. 
This can happen if the reduction in last-level cache misses from the additional pages
allocated to the task in the H-mode offsets the pessimism from using a more conservative estimation technique 
for H-WCETs than for L-WCETs. Fig.~\ref{f:example_pldc} illustrates this possibility.
Leveraging such cases in the analysis can lead to improvements in provable schedulability, over
approaches that do not reallocate cache pages in the event of a mode switch.

\subsection{Aspects of deadline scaling}

As already mentioned, in L-mode, the H-tasks report to the EDF scheduler a shorter deadline $D_i^L \leq D_i$, for the purpose
of scheduling decisions. In H-mode, the true deadline is used (i.e., $D_i^H =D_i$). The designer has freedom over the selection 
of L-mode deadlines and the process that determines them is called deadline scaling. In~\cite{Ekberg_Yi_12}, Ekberg and Yi propose
a heuristic that, starting with $D_i^L=D_i$ for every task, iteratively tinkers with the task L-mode deadlines, using 
their schedulability test to guide the heuristic to identify opportunities to decrease a deadline by a notch. In our work, we
also use the same heuristic (details in~\cite{Ekberg_Yi_12}), with no changes except for the fact that our new schedulability 
analysis, cognisant of cache reclamation by H-tasks at mode change, is used, instead of the original analysis in~\cite{Ekberg_Yi_12}.

\section{Schedulability analysis}\label{s:schedulability_analysis}

In this section, we propose a schedulability analysis, drawing from that of Ekberg and Yi~\cite{Ekberg_Yi_12, Ekberg_Yi_14}, for the system model described earlier.
It assumes that the number of hot pages in the two modes ($\sigma_i^L$ and $\sigma_i^H$) for each task is given. Similarly, we also assume that  the scaled L-mode deadline $D_i^L$, with $D_i^H \equals D_i$, is given for each task. As explained, this analysis is to be coupled
with the heuristic of Ekberg and Yi to guide the derivation of the L-mode scaled deadlines.
How to assign values to $\sigma_i^L$ and $\sigma_i^H$, is discussed in the next section.

Ekberg and Yi's analysis is based on the \emph{demand bound function}, $dbf(\ell)$, which  upper-bounds
the execution demand over any time interval of length $\ell$ by all jobs whose scheduling windows  are fully contained in $\ell$. The \emph{scheduling window} of a job is the time interval between its release and its deadline.
The schedulability analysis for the L-mode can be done using standard dbf for EDF, in which the computation demand of a task is maximum when a job is released at 
the beginning of the time interval. In H-mode, if the time interval under consideration begins at the mode switch, in addition to the demand of jobs whose 
scheduling windows are fully contained in $\ell$, we need to consider the demand of \emph{carry-over} jobs of H-tasks, i.e.~jobs of H-tasks that were released, but not finished, 
at the time of the mode switch. Thus, for H-mode analysis, we consider that the scheduling window of a carry-over job always starts at the mode switch.

A key result of Ekberg and Yi's analysis is the following lemma, which allows to upper bound the demand in H-mode of a carry-over job:
\begin{lemma}\label{l:Ekberg_Yi_lema1} (Demand of carry-over jobs, Ekberg and Yi's \cite{Ekberg_Yi_14}'s Lemma 1) Assume that EDF uses relative deadlines $D_i^L$ and $D_i^H$, with $D_i^L \le D_i^H = D_i$ for high-criticality task $\tau_i$, and that we can guarantee that the demand is met in low-criticality mode (using $D_i^L$). If the switch to high-criticality mode happens when a job from $\tau_i$ has a remaining scheduling window of $x$ time units left until its true deadline, then the following hold:
\begin{enumerate}
\item If $x < D_i^H - D_i^L$, then the job has already finished before the switch.
\item If $x \ge D_i^H - D_i^L$, then the job may be a carry-over job, and no less than $\llbracket C_i^L - x + D_i ^H - D_i^L \rrbracket_0$ time units of the job's work were finished before the switch.
\end{enumerate}
\end{lemma}
where
\begin{equation}
\llbracket z \rrbracket_{\min}^{\max} \equals
\begin{cases}
	\min	& \text{if $ z < \min $}\\
	z       & \text{if $ \min \leq z \leq \max $}\\
	\max    & \text{if $ z > \max $}
\end{cases} \nonumber
\label{eq:operator}
\end{equation}
and the bound arguments $\min$ and $\max$ can be omitted, if they are $-\infty$ or $+\infty$, respectively.

In the classic Vestal model there is no resource reallocation upon a mode switch, except for CPU time. Therefore, the computation demand of a carry-over job never exceeds the demand of a full job, and the maximum demand in any time interval of length $\ell$ corresponds to executions that maximise the number of full jobs after the mode switch as shown in Fig.~\ref{f:ekberg_max_demand}(i).
In this scenario, the time interval of length $\ell$ under consideration ends with a deadline for each 
task $\tau_i$ present in H-mode. Accordingly, the subinterval of length $x =\ell \; \mathrm{mod} \; T_i$, 
which starts with the mode switch, is maximised, under the constraint that the number of full jobs is maximum, maximizing its demand,  because, by Lemma~\ref{l:Ekberg_Yi_lema1}, the maximum demand of a carry-over job is non-decreasing with the size of its scheduling window in H-mode.

\begin{figure}
\centering
 \includegraphics[scale=.9]{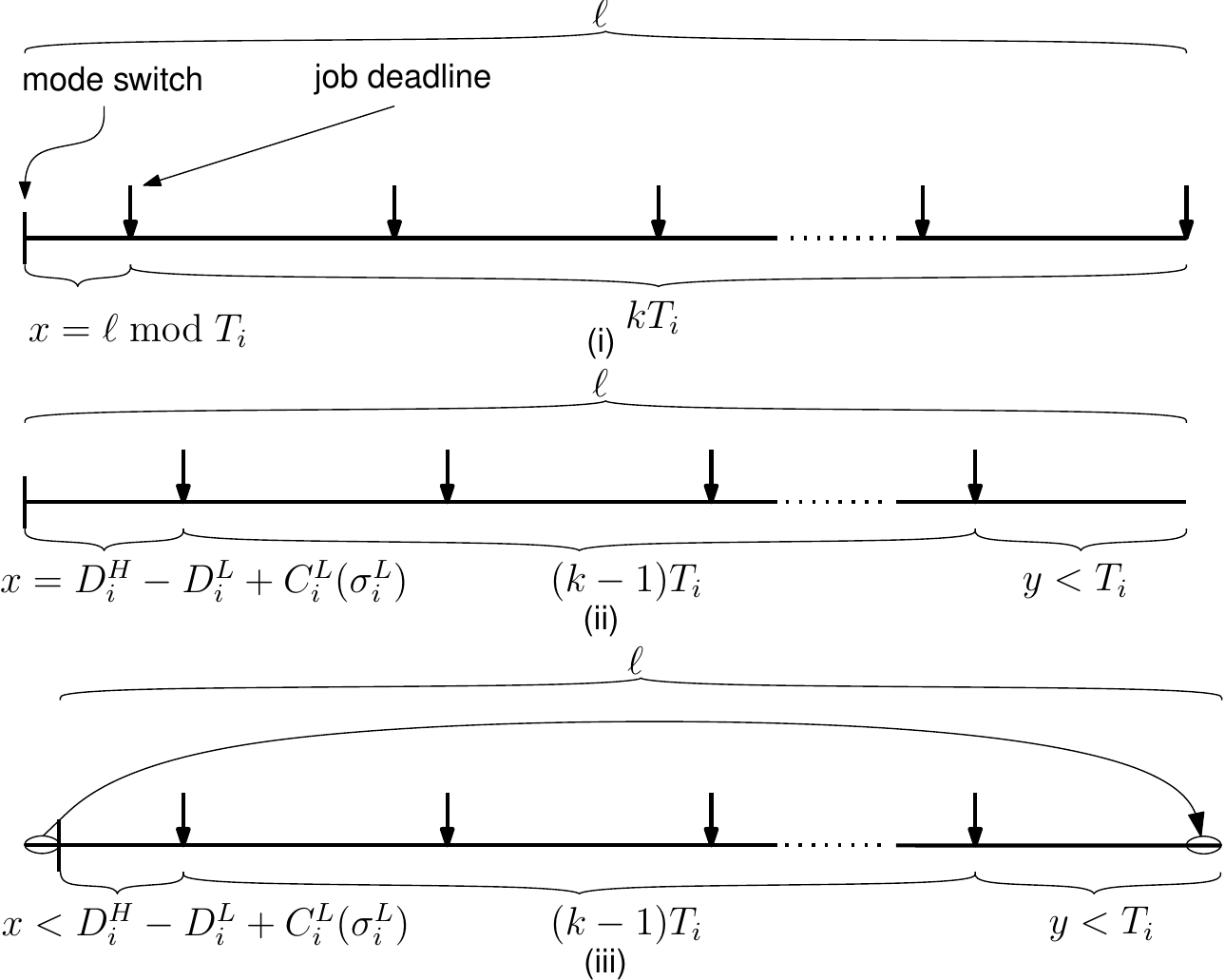}
\caption{Execution with maximum demand for $C_i^H(\sigma_i^H) \ge C_i^L(\sigma_i^L)$.}
\label{f:ekberg_max_demand}
\end{figure}

In our model, $\sigma_i^L < \sigma_i^H$, therefore the maximum demand does not necessarily occur in executions as shown in Fig.~\ref{f:ekberg_max_demand}(i). The reason is that a full job 
in H-mode executes with $\sigma_i^H$ pages locked in the cache, whereas a carry-over job executes with only $\sigma_i^L$ pages locked in 
the cache until the mode switch, and thus, for safety, we assume that it executes with only $\sigma_i^L$ pages locked in the cache 
throughout its execution. Therefore, the demand in H-mode of a full job is not necessarily larger than the (outstanding) demand in H-mode 
of a carry-over job. For example, if $ \ell = T_i$, then the execution shown in Fig.~\ref{f:ekberg_max_demand}(i) has no carry-over job, 
and the maximum demand is $C_i^H(\sigma_i^H)$. However, for such a value of $\ell$, we can have an execution in which there is a \emph{maximal 
carry-over job}, i.e.~a carry-over job with maximum demand in H-mode, $C_i^H(\sigma_i^L)$. If $\sigma_i^L $ and $\sigma_i^H$ are such that 
$C_i^H(\sigma_i^L) > C_i^H(\sigma_i^H)$, i.e. if the extra assigned cache lines are useful to the task and reduce its execution time, the latter execution scenario has a demand that is higher than the former.

Because, Ekberg and Yi's analysis, assumes that the demand of a carry-over job is never larger than the demand of a full-job, 
we need new analysis, built on the following lemma.

\begin{lemma}\label{l:max_demand_H-mode}
In H-mode, for any time interval of length $\ell$, the demand by the jobs of an H-task $\tau_i$ whose scheduling windows are fully contained in $\ell$ is maximum:
\begin{enumerate}
\item either in executions with the maximum number of full jobs after a carry-over job, if one fits, as illustrated in Fig.~\ref{f:ekberg_max_demand}(i)
\item or in executions with the maximal carry-over job with the earliest possible deadline 
followed by as many full jobs as can fit in the remaining time and that arrive as soon as possible, as illustrated in Fig.~\ref{f:ekberg_max_demand}(ii)
\end{enumerate}
\end{lemma}
\begin{proof}
The scheduling window of a carry-over job in H-mode always starts at the mode switch. Thus, in H-mode, a time interval of length $\ell$ can 
include the scheduling windows of at most one carry-over job of $\tau_i$, and of a number of full jobs (that is, jobs released 
at or after the mode switch).

If $\ell < D_i^H - D_i^L + C_i^L(\sigma_i^L)$, no full job contributes to the demand, because the shortest length of the scheduling window 
of a full job is $D_i$ and  $\ell < D_i^H$ (because $C_i^L(\sigma_i^L) \le D_i^L$). Thus, in this case, execution scenario 1) 
maximises the amount of time that can be used by a carry-over job and, by Lemma~\ref{l:Ekberg_Yi_lema1}, its demand is maximum.

Let $\ell \ge D_i^H - D_i^L + C_i^L(\sigma_i^L)$. Let $y$ be the length of the right-most subinterval of $\ell$, from the deadline
of the last job contained in $\ell$, if any, until the end of $\ell$.
If  $0 <  y < T_i$, execution scenario 2) maximises the demand at the beginning of the interval, because the earliest deadline of a 
maximal carry-over job is $ D_i^H - D_i^L + C_i^L(\sigma_i^L)$, by Lemma~\ref{l:Ekberg_Yi_lema1} (after substituting $C_i^L$ with 
$C_i^L(\sigma_i^L)$); but, the demand during $y$ does not increase, because the deadlines of two consecutive jobs of $\tau_i$ must be 
$T_i$ time units apart, and therefore subinterval $y$ cannot contain the scheduling window of a full job. 
If $x$ decreases by some amount and $y$ increases by the same amount, 
as illustrated in Fig.~\ref{f:ekberg_max_demand}(iii), the demand of the carry-over job decreases 
without increasing the demand at the end of the interval, unless $y$ becomes $T_i$. If this happens, we transform execution scenario 2) into 
execution scenario 1) and increase the total demand of full jobs, but decrease the demand of the carry-over job, possibly eliminating it. 
Thus, if the total demand in $\ell$ increases when $y$ becomes $T_i$, then execution scenario 1) has maximum demand, else execution 
scenario 2) has maximum demand. 
Decreasing $x$ by a larger amount than necessary for $y$ to become equal to $T_i$, does not increase the demand w.r.t execution scenario 1), 
since it increases neither the total demand of full jobs nor the demand of the carry-over job.
Finally, if $y=0$,  then execution scenarios 1) and 2) are identical and both have a maximal carry-over job of $\tau_i$ with the 
earliest deadline, and the maximum number of full-jobs of $\tau_i$ that can fit in $\ell$, therefore their demand is maximum.
\end{proof}

Thus, a tight demand bound function for any execution in H-mode is the maximum of the demands of execution scenarios 1) and 2), 
illustrated respectively in Fig.~\ref{f:ekberg_max_demand}(i) and (ii).
Next, we adapt Ekberg and Yi's demand bound function for execution scenario 1) to take into account a different number of pages 
locked in the cache per mode. After that, we develop the demand bound function for execution scenario 2), 
which was not relevant in previous work.

In \cite{Ekberg_Yi_14}, Ekberg and Yi provide a bound for the demand of execution scenario 1) in a time interval of length $\ell$, as follows:
\begin{equation} \label{e:dbf_ekberg_n}
\textit{full}_i^H(\ell) - done_i^H(\ell)
\end{equation}
where $\textit{full}_i^H(\ell)$, given by \eqref{eq:ekberg:full}, 
is the maximum demand by all jobs of $\tau_i$ whose scheduling window is fully contained in that interval (in H-mode, the scheduling window of a carry-over job begins at the mode switch and ends at its deadline), and $done_i^H(\ell)$, given by \eqref{eq:ekberg:done}, is the minimum demand of any carry-over job that must be satisfied before the mode switch. 

\begin{equation}
\textit{full}_i^H(\ell) = \left \llbracket \left( \left \lfloor \frac{\ell - \left(D_i ^H - D_i^L\right)}{T_i} + 1 \right \rfloor \right) C_i^H \right \rrbracket_0
\label{eq:ekberg:full}
\end{equation}
\begin{equation}
done_i^H(\ell) =
%\begin{cases}
% \llbracket C_i^L - (\ell \; mod \; T_i) + D_i ^H - D_i^L \rrbracket_0, \\
%& \hspace*{-34mm} \text{if }  (D_i ^H - D_i^L ) \le \ell \; mod \; T_i < D_i^H \\
%0, & \hspace*{-34mm} \text{otherwise }
%\end{cases}
\begin{cases}
\llbracket C_i^L - (\ell \; mod \; T_i) + D_i ^H - D_i^L \rrbracket_0, & \text{if }  (D_i ^H - D_i^L ) \le \ell \; mod \; T_i < D_i^H \\
0,                                                                     & \text{otherwise }
\end{cases}
\label{eq:ekberg:done}
\end{equation}

We now derive the new expressions for $\textit{full}_i^H(\ell)$ and $done_i^H(\ell)$ to take into account that the number of pages locked in the 
cache in the L-mode and in the H-mode may be different. In this derivation, like Ekberg and Yi in \cite{Ekberg_Yi_14}, we assume that 
there is a carry-over job, if one fits. At the end of this section, we show that, for any time interval $\ell$ after the mode switch, the demand is maximum when there is a carry-over job.

So, assuming that the first job is a carry-over job, if one fits, we modify \eqref{eq:ekberg:full}
(originally (2) in \cite{Ekberg_Yi_14}) as follows:
\begin{equation}
\begin{split}\label{e:fulll_i^H_n}
\textit{full}_i^H(\ell) = & \left \llbracket \left \lfloor \frac{\ell - (D_i ^H - D_i^L )}{T_i} \right \rfloor + 1 \right \rrbracket_0^1 C_i^H(\sigma_i^L)  \\
& + \left \llbracket \left \lfloor \frac{\ell - (D_i ^H - D_i^L )}{T_i} \right \rfloor  \right \rrbracket_0 C_i^H(\sigma_i^H)
\end{split}
\end{equation}
The first term bounds the demand in H-mode of the carry-over job. As shown by Lemma~\ref{l:Ekberg_Yi_lema1},
$D_i^H - D_i^L$ is the smallest scheduling window (in H-mode) of a carry-over job of $\tau_i$. To be safe, we assume that the number of locked 
pages of the carry-over job is $\sigma_i^L$, therefore the maximum demand of the carry-over job, 
ignoring any demand that may have been satisfied before the mode switch, is $C_i^H(\sigma_i^L)$.  
The second term bounds the demand of the jobs that are released after the mode switch and therefore we use their maximum execution time 
with the respective number of locked pages in H-mode, $C_i^H(\sigma_i^H)$.

Likewise, for $done_i^H(\ell)$, we modify ~\eqref{eq:ekberg:done} (originating as (3) in \cite{Ekberg_Yi_14})
by substituting $C_i^L$ with $C_i^L(\sigma_i^L)$. That is, we make explicit that any computation before the mode switch must have been 
performed with $\sigma_i^L$ pages locked in the cache.

Thus, by replacing \eqref{eq:ekberg:full} and \eqref{eq:ekberg:done}
(i.e., (2) and (3) in \cite{Ekberg_Yi_14}) with their versions aware of the number of pages locked in the cache, 
\eqref{e:dbf_ekberg_n} provides a bound for execution scenario 1) when the number of pages locked in the cache is changed 
from $\sigma_i^L$ to $\sigma_i^H$ upon a switch to H-mode. 

The demand under execution scenario 2) is a step function and is given by \eqref{e:dbf_max_carry-over_n}.
\begin{equation} \label{e:dbf_max_carry-over_n}
\begin{split}
step_i^H(\ell) = & \left \llbracket \left \lfloor \frac{\ell - \big(D_i ^H - D_i^L + C_i^L(\sigma_i^L)\big)}{T_i} \right \rfloor + 1 \right \rrbracket_0^1 C_i^H(\sigma_i^L) \\ 
& + \left \llbracket \left \lfloor \frac{\ell - \big(D_i ^H - D_i^L + C_i^L(\sigma_i^L)\big)}{T_i} \right \rfloor \right \rrbracket_0  C_i^H(\sigma_i^H)
\end{split}
\end{equation}
where the first term bounds the demand in H-mode of the carry-over job, which is maximum and has a deadline at the earliest time instant, 
and the second term bounds the demand of the maximum number of full jobs that fit after the carry-over job.

Thus, a demand bound function for any interval of length $\ell$ in H-mode is:
\begin{equation}
dbf_i^H(\ell) = max\big( step_i^H(\ell),\; \textit{full}_i^H(\ell) - done_i^H(\ell)\big)
\end{equation}

Finally, we show that executions with a carry-over job have a higher demand than executions without a carry-over job, an assumption we made above in the derivation of $dbf_i^H(\ell)$.

\begin{lemma}
For any sporadic task $\tau_i$, its maximum demand in H-mode in a time interval of length $\ell$ with only full jobs is not higher than its maximum demand in a time interval of the same length $\ell$ with a carry-over job.
\end{lemma}
\begin{proof}
The demand of an execution of  $\tau_i$ in H-mode with only full jobs can be bounded by the standard dbf for sporadic tasks with the appropriate parameters:
\begin{equation}
\label{e:std_dbf}
\begin{split}
 \left \llbracket \left( \left \lfloor \frac{\ell - D_i ^H}{T_i} \right \rfloor + 1 \right) C_i^H(\sigma_i^H) \right \rrbracket_0
\end{split}
\end{equation}

\begin{figure}
\centerline{\hspace*{0mm}\includegraphics[scale=.65]{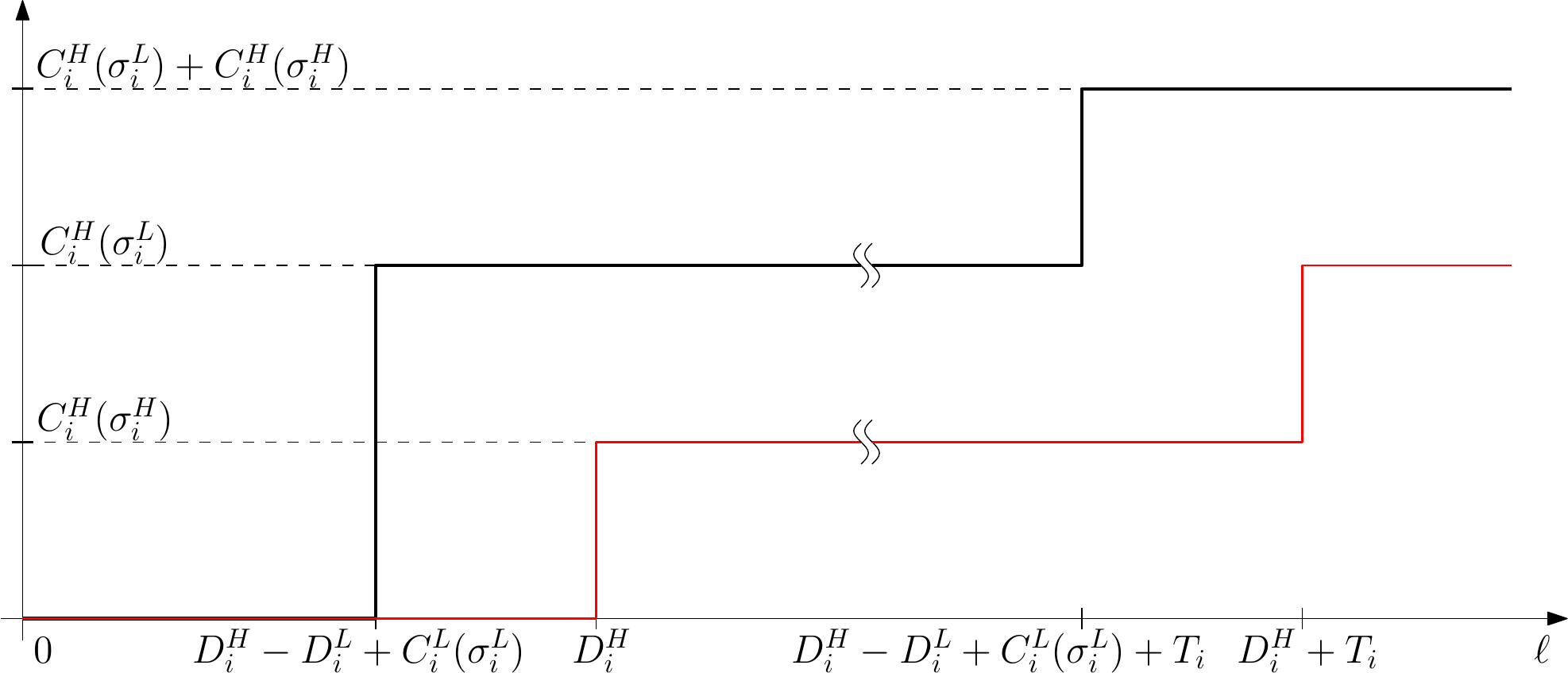}}
\caption{The demand over an interval of length $\ell$ starting at mode switch when the carry-over job is maximum at the earliest time (black line) dominates the maximum demand over an interval of the same length without a carry-over job.}
\label{f:carry-over_dominance}
\end{figure}

If $\ell < D_i^H$ the demand is zero. Let $\ell \ge D_i^H$. Consider a time interval of length $\ell$ starting at the mode switch. Consider an execution in which the carry-over job has maximum demand, $C_i^H(\sigma_i^L)$, at the earliest possible time, $D_i^H-D_i^L+C_i^L(\sigma_i^L)$, and the following jobs of $\tau_i$ arrive as soon as possible. For any time interval of length $\ell$, the demand of such an execution is never lower than the demand of an execution without a carry-over job (see Fig.~\ref{f:carry-over_dominance} for further intuition). Indeed:
\begin{enumerate}
\item $C_i^H(\sigma_i^L) \ge C_i^H(\sigma_i^H)$, because $\sigma_i^L \le \sigma_i^H$
\item $D_i^H - D_i^L + C_i^L(\sigma_i^L) \le D_i^H$, because $D_i^L \ge C_i^L(\sigma_i^L)$
\end{enumerate}
\end{proof}

\section{Last-level cache allocation \label{s:heuristics}}

The description of our analysis assumed that for each task the number of its hottest pages that are locked in the cache in  
each mode is already determined. We now propose a heuristic for this allocation. Our objective is to efficiently 
distribute the last-level cache among the tasks, for improved schedulability. 

A heuristic is needed because a brute-force
combinatorial exploration of all possible allocations is intractable and the arbitrary nature of progressive lockdown curves
means that there is no structure in the problem to employ for strict optimality, in the general case.
Additionally, since any possible allocation configuration of the cache for the L-mode can be re-configured in many ways 
for the H-mode, we opt for a two-staged heuristic. We first determine the L-mode allocation and then, subject to the constraints
stemming from that, we determine the H-mode allocation. Since the schedulability analysis is conceptually complex (and made 
even more so by the deadline scaling), our idea is to optimise, in each mode, for a metric that strongly correlates with
schedulability: the task set utilisation. So, we first (i)~assign values to the $\sigma_i^L$ variables (corresponding to the number of
locked pages in L-mode) for each task so that the L-mode utilisation ($\sum_{\tau_i \in \tau} \frac{C_i^L(\sigma_i^L)}{T_i}$
is minimised and subsequently (ii)~assign values to the $\sigma_i^H$ variables for the H-tasks (with $\sigma_i^H \geq \sigma_i^L$)
such that the (steady) H-mode utilisation ($\sum_{\tau_i \in \tau(H)} \frac{C_i^H(\sigma_i^H)}{T_i}$) is minimised. Next, we discuss
the ILP formulation implementing this heuristic.

\subsection{L-mode allocation}
Let $\sigma_i^L$ be the number of pages by $\tau_i$ in the last-level cache in the L-mode 
and $\sigma^T$ be the total number of pages that fit in that cache. Intuitively, lower utilisation correlates with 
better schedulability, hence, our objective is to set the $\sigma_i^L$ values such that the total task set utilisation in L-mode is minimised. 
To model this heuristic with ILP formulation, we define a binary decision variable variable $UL_{i,j}$ such that: 

\begin{align*}
UL_{i,j} = 
\begin{cases}
 1,  & \text{if $j$ pages are assigned to $\tau_i \in \tau$ in L-mode} \\
 0,  &  \text{otherwise}
\end{cases}
\end{align*}

Since our aim is to minimise the system utilisation in L-mode, the objective function and constraints take the form: 
\begin{align}
&\text{Minimise}  \sum_{\forall \tau_i \in \tau} \sum_{j=0}^{\sigma^T} UL_{i,j} \times U_i^L(j) \label{eq:Lmode_Objective}\\
%& \textbf{Subject to:} \nonumber\\
\textbf{s. t.}~& \sum_{j=0}^{\sigma^T} UL_{i,j} =1, \; \forall \tau_i \in \tau \label{eq:AllocateForEachTaskAtleaseOneSizeInLmode}\\
& \sum_{\forall \tau_i \in \tau} \sum_{j=0}^{\sigma^T}  j \times UL_{i,j}  \leq \sigma^T  \label{eq:CacheSizeConstraintInLMode}\\
& \sum_{j=0}^{\sigma^T} UL_{i,j} \times U_i^L(j) \leq 1, \; \forall \tau_i \in \tau   \label{eq:UtilisationPerTaskLE1InLmode}\\
& \sum_{\forall \tau_i \in \tau} \sum_{j=0}^{\sigma^T} UL_{i,j} \times U_i^L(j) \leq m, \; \forall \tau_i \in \tau   \label{eq:TotalUtilisationLE1InLmode}
\end{align}

The $U_i^L(j)$ constants are derivable from the tasks' progressive lockdown curves.
The set of constraints given by~\eqref{eq:AllocateForEachTaskAtleaseOneSizeInLmode} ensures that each task is considered for allocation in the last-level cache. 
A task can be allocated any number of pages from zero to all $\sigma^T$ pages in the cache. However, the sum of all pages allocated to tasks 
should not exceed the cache capacity (i.e., $\sum_{\forall \tau_i \in \tau}\sigma_i^L \leq \sigma^T$), which is ensured
by~\eqref{eq:CacheSizeConstraintInLMode}. 
Additionally, the utilisation of each task in the L-mode for the selected number of pages should not exceed one 
(i.e., $U_i^L(\sigma_i^L) \leq 1$), which is ensured by~\eqref{eq:UtilisationPerTaskLE1InLmode}; otherwise, 
the task will not be unschedulable. Finally, the set of constraints given by~\eqref{eq:TotalUtilisationLE1InLmode} 
ensures that the total utilisation of the task set in L-mode, under the particular allocation, should not exceed the number of cores 
in the platform ($\sum_{\forall \tau_i \in \tau} U_i^L(\sigma_i^L)\leq m$); otherwise the task 
set would be unschedulable in the L-mode, under these parameters.     

\subsection{H-mode allocation} 
In this second stage of our allocation heuristic, we determine how the pages reclaimed from
the idled L-tasks at the switch from L-mode to H-mode, are to be redistributed to the H-tasks.
Let $\sigma_i^H$ denote the number of cache pages in the last-level cache allocated to a task $\tau_i\in \tau(H)$ in the H-mode. 
Our ILP formulation for the H-mode allocation derives $\sigma_i^H$ for each task $\tau_i \in \tau(H)$ 
in such a way that the overall steady H-mode system utilisation is minimised. 
We define a binary decision variable $UH_{i,j}$ such that:  
\begin{align*}
UH_{i,j} = 
\begin{cases}
 1,  & \text{if $j$ pages are assigned to $\tau_i \in \tau(H)$ in H-mode} \\
 0,  &  \text{otherwise}
\end{cases}
\end{align*}

The objective function in this stage minimises the H-mode utilisation and the ILP formulation is given below:

\begin{align}
&\text{Minimise}  \sum_{\forall \tau_i \in \tau(H)} \sum_{j=0}^{\sigma^T} UH_{i,j} \times U_i^H(j) \label{eq:Hmode_Objective}\\
%& \textbf{Subject to:} \nonumber\\
\textbf{s. t.}~& \sum_{j=0}^{\sigma^T} UH_{i,j} =1, \; \forall \tau_i \in \tau(H) \label{eq:AllocateForEachTaskAtleaseOneSizeInHmode}\\
& \sum_{\forall \tau_i \in \tau(H)} \sum_{j=0}^{\sigma^T}  j \times UH_{i,j}  \leq \sigma^T  \label{eq:CacheSizeConstraintInHMode}\\
& \sum_{j=0}^{\sigma^T} UH_{i,j} \times U_i^H(j) \leq 1, \; \forall \tau_i \in \tau (H)  \label{eq:UtilisationPerTaskLE1InHmode}\\
& \sum_{\forall \tau_i \in \tau(H)} \sum_{j=0}^{\sigma^T} UH_{i,j} \times U_i^H(j) \leq m, \; \forall \tau_i \in \tau (H)  \label{eq:TotalUtilisationLE1InHmode}\\
& \sum_{j=0}^{\sigma^T} j \times UH_{i,j} \geq \sigma_i^L, \; \forall \tau_i \in \tau(H) \label{eq:HModePageGELModePages}
\end{align}
The constraints given by~\eqref{eq:AllocateForEachTaskAtleaseOneSizeInHmode}-\eqref{eq:TotalUtilisationLE1InHmode} 
are similar to those given by~\eqref{eq:AllocateForEachTaskAtleaseOneSizeInLmode}-\eqref{eq:TotalUtilisationLE1InLmode} for the L-mode.
These constraints ensure that every H-task is considered for allocation, 
the sum of allocated pages does not exceed the total number of pages in the cache 
($\sum_{\forall \tau_i \in \tau(H)} \sigma_i^H \leq \sigma^T$), 
each task has utilisation not greater than one ($U_i^H(\sigma_i^H)\leq 1$) and sum of their utilisations 
is less than or equal to the number of cores ($\sum_{\tau_i\in \tau(H)} U_i^H(\sigma_i^H)\leq m$). 
As for the set of constraints given by~\eqref{eq:HModePageGELModePages}, they express the fact that 
$\sigma_i^H \geq \sigma_i^L,~\forall\tau_i \in \tau(H)$.
In other words, in the H-mode, an H-task may be allocated additional pages, reclaimed from the idled L-tasks, but never fewer.
The reason for restricting the solution space in this manner is practical: Unlike cache pages allocated to L-tasks in the L-mode
which are reclaimable immediately after a mode switch (since no L-tasks execute in the H-mode), the instant that some cache page 
could be taken away from an H-task is ill-defined if there is a carry-over job from that task. Even if it is assumed that a page
can be taken away from that H-task, once its carry-over job completes, this would introduce an arbitrarily long (in the general 
case) transition to steady H-mode, in the case of a carry-over job with long outstanding execution time and even longer deadline.
The schedulability analysis would then become extremely complicated, with hardly any gains expected
from such a more general model. 

As a final note,
one might consider optimising $\sigma_i^L$ and $\sigma_i^H$ jointly in a single step but this would be non-trivial due to lack 
of a single meaningful objective function to minimise. 

\begin{table}[!t]
\centering
\begin{tabular}{|c|c|}
\hline
 Parameters                                                & Values \\
 \hline \hline
 Task-set size ($n$)                                       &  $ \lbrace \underline{10}~,13,~15,~20 \rbrace $                 \\\hline
 Inter-arrival time $T_i$                                  &  $10$ to $100$~msec ($1$ msec resol.)                           \\\hline
 Fraction of H-tasks in $\tau$                             &  $\lbrace 20\%,~\underline{40\%},~60\%,~80\% \rbrace$           \\\hline
 Ratio of $C_i^H$ to $C_i^L$                               &  $\lbrace 4,~6,~\underline{8},~10,~12\rbrace $                  \\\hline
 Lower bound $\alpha$ on $C_i^L(0)/C_i^H(0)$               &  $\lbrace \underline{0.1},~0.2,~0.4,~0.8\rbrace $               \\\hline
 Mean ($\lambda$) for x-coordinate (in pages) of the       &  $\lbrace 5,~10,~15,~20,~25,~\underline{30}\rbrace$             \\
 taper point (X, Y) in the prog. lockdown curve            &                                                                 \\\hline
 Cache size                                                &  $\lbrace$\underline{$512$} KB, $1$ MB, $2$ MB, $4$ MB$\rbrace$ \\\hline
 Number of cores ($m$)                                     &  $\lbrace \underline{1},~2,~4,~8\rbrace $                       \\\hline
 Nominal L-mode utilisation ($\frac{1}{m} \sum_{\tau_i \in \tau} U_i^L(0)  $)         &  $\lbrace0.1:0.1:1.5\rbrace$         \\\hline
 Page size                                                 &  $4$ KB                                                         \\\hline
\end{tabular} 
\caption{Overview of Parameters}
\label{tb:parameters}
\end{table} 

\section{Evaluation}
\label{s:evaluation}
We experimentally explore the effectiveness of our proposed allocation heuristics and dynamic redistribution mechanism
in terms of schedulability.

\subsection{Experimental Setup}
We developed a Java tool for our experiments. (Sources found at~\cite{Manberg_code_link}.) 
Its first module generates the synthetic workload (task sets). A second module implements the ILP models 
for the allocation heuristics. Using the generated task-set and platform information as input, it partitions the cache
to the tasks. A third module uses the task set, platform information and cache assignment as an input and 
performs the schedulability analysis and task-to-core allocation. The following parameters control the task set generation:

\begin{itemize}
\item 
We generate the L-mode task utilisations with zero locked pages ($U^L_i(0)$) for a given target task set L-mode 
utilisation ($\sum_{i\in \tau} U^L_i(0)$) using UUnifast-discard~\cite{UUnifast,UUnifast_Discard}, for unbiased distribution 
of task utilisations. 
\item 
Task periods are generated with a log-uniform distribution in the range $10$-$100$ ms. We also assume implicit deadlines,
even if our analysis holds for constrained deadlines.  
\item 
The L-mode progressive lockdown curve of a task $\tau_i$ is derived as follows. $C_i^L(0)$ is obtained as $U^L_i(0) \cdot T_i$. 
Then the L-WCET with full cache ($C_i^L(\sigma^T)$) is randomly generated with uniform distribution over $[\alpha \cdot C_i^L(0), C_i^L(0)]$,
where $\alpha<1$ is a user-defined parameter. Then we add a ``bending point'' with random coordinates ($X$,$Y$). $X$ is sampled from
a Poisson distribution with median $\lambda$ (user-defined parameter) and $Y$ is sampled from a uniform distribution in the range
$[C_i^L(\sigma^T),~Z]$, where $Z$ is the y-coordinate of the point where the $x=X$ axis intersects the line 
$((0, C_i^L(0)),~(\sigma^T, C_i^L(\sigma^T)))$. See Figure~\ref{f:X:Y} for an illustration.
The two linear segments $((0, C_i^L(0)),~(X,Y))$ and $((X,Y),~(\sigma^T, C_i^L(\sigma^T)))$, joined at an angle at (X,Y)
define our L-mode progressive lockdown curve. This generation scheme can output (i)``L-shaped'' curves, where the L-WCET drops 
sharply with a few pages and then stays flat; (ii)~flat L-WCET curves, largely insensitive to the number of pages; and 
(iii)~in-between\footnote{We thank Renato Mancuso, for having shared with us his empirical observations about the shapes
of progressive lockdown curves.}.
\item Based on the target fraction of H-tasks in the task set (user-specified), a number of tasks (rounded-up) will be H-tasks.
For those ones, an H-mode progressive lockdown curve is generated, by up-scaling of the respective L-mode curve. The scalefactor
(multiplier) is user-specified.
\end{itemize}

For each set of input parameters, we generate 100 task sets. We use independent pseudo-random number generators for the utilisations, 
minimum inter-arrival times/deadlines, Poisson distribution and $C_i^L(\sigma^T)$ generation, and reuse their seeds~\cite{Raj_91}.

\begin{figure}[!t]
\centering
 \includegraphics[width=0.8\linewidth]{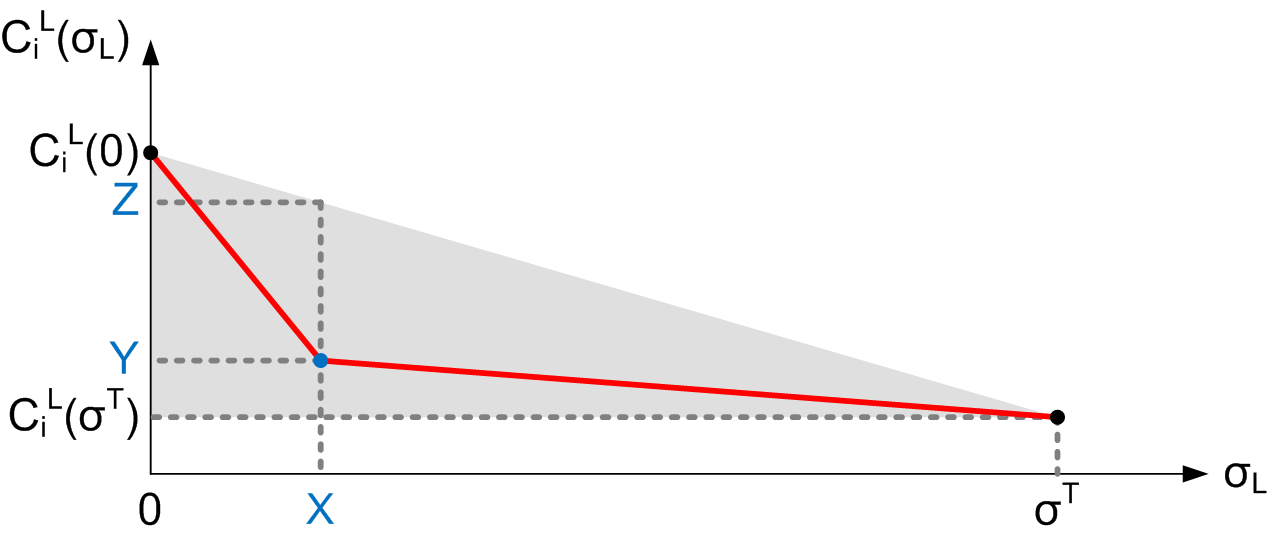}
 \caption{\label{f:X:Y} Illustration of progressive lockdown curve generation.}
\end{figure}

The second module of our tool models the ILP formulations for an input task set on IBM ILOG CPLEX v12.6.3 and interfaces it with the Java tool using 
CONCERT technology. In all experiments, the defaults for different parameters are underlined in Table~\ref{tb:parameters}. 
Note that, since the target utilisation when invoking UUnifast-discard corresponds to $\sum_{i\in \tau} U^L_i(0)$ (i.e., the L-mode
task set utilisation when \textit{none} of the tasks uses any cache), it is possible that some task sets where this nominal utilisation is 
greater than $m$ may be in fact schedulable, since the allocation of cache pages to their tasks may drive down their L-WCETs and 
decrease the L-mode utilisation to below $m$. So, we explore tasks sets with nominal utilisation up to 1.5.
To generate task sets with nominal utilisation greater that $m$, we first generate a task set with target nominal utilisation of 
$m$ using UUnifast-discard, and subsequently multiply with the desired scalar. This preserves the properties 
of UUnifast-discard. 

Finally, note that in our experiments, after the values of $\sigma_i^L$ and $\sigma_i^H$ are determined for the tasks, 
we employ First-Fit bin-packing for task-to-core assignments. The 
new analysis (introduced in Section~\ref{s:schedulability_analysis}) is used as a schedulability test, on each processor, for 
testing the assignments. Note that our new analysis is also used for the derivation of the deadline scalefactors for the H-tasks,
using Ekberg's and Yi's (otherwise, unmodified) approach~\cite{Ekberg_Yi_12}. The bin-packing ordering is
by decreasing criticality and decreasing deadline -- a task ordering that works well with Ekberg and Yi's algorithm, 
as shown in our previous work~\cite{Awan_KPE_17}.

\subsection{Results}
Since presenting plots for each possible combination of parameters would be impractical,  
each experiment varies only one parameter, with the rest conforming to the respective defaults from Table~\ref{tb:parameters}.
Even so, the number of plots would still be too high to accommodate. 
So, instead of providing plots comparing the approaches in terms of scheduling success ratio (i.e., the fraction of task sets 
deemed schedulable under the respective schedulability test), we condense this information by providing plots of 
\textit{weighted schedulability}.\footnote{The plots of (non-weighted) schedulability can still be found in the Appendix of our 
TR~\cite{manberg:TR}.}
This performance metric, adopted from~\cite{Bastoni_etal_2010}, condenses what would have been 
three-dimensional plots into two dimensions. It is a weighted average that gives more weight to task-sets 
with higher utilisation, which are supposedly harder to schedule. Specifically,
using the notation from~\cite{Burns_Davis_14}:

Let $S_y(\tau,p)$ represent the binary result ($0$ or $1$) 
of the schedulability test $y$ for a given task-set $\tau$ with an input parameter $p$. 
Then $W_y(p)$, the weighted schedulability for that schedulability test $y$ 
as a function $p$, is:

\begin{equation}
\label{e:weighted_schedulability_measure}
W_y(p) = \frac{\sum_{\forall\tau} \left(\bar{U}^L(\tau) \cdot S_y(\tau,p)\right)}{\sum_{\forall\tau} \bar{U}^L(\tau)}
\end{equation}  

In~\eqref{e:weighted_schedulability_measure}, (adapted from~\cite{Burns_Davis_14}), 
$\bar{U}^L(\tau) \equals \frac{U^L(\tau)}{m}$ is the system utilisation in L-mode, 
normalised by the number of cores. $m$.

The purpose for our experiments was to quantify the schedulability improvement over a system model without cache reallocation at mode switch.
However, the state-of-the-art scheduling algorithm by Ekberg and Yi, assumed for the latter, is cache-agnostic: whether the L-WCETs and H-WCETs
estimates used are cache-cognisant or not is opaque to the algorithm. Therefore, in order to have a fair comparison,
we needed to specify an efficient cache partitioning heuristic, even for the case of no cache reallocation.

The different curves depicted on our plots are the following:

\begin{itemize}
\item[] \textbf{VT:} 
This ``Validity Test'', for indicative purposes, is a \textit{necessary} condition for a task set to be mixed-criticality 
schedulable at all (i.e., under any possible scheduling arrangement). The actual condition, verifiable with low computational complexity, is:
\begin{align*}
\left(U_i^L(\sigma^T) \leq 1,~\forall \tau_i \in \tau \right)&~~\wedge & \left( U_i^H(\sigma^T) \leq 1,~\forall \tau_i \in \tau(H) \right) &~~\wedge \nonumber\\  
 \Big( \sum_{\tau_i \in \tau} U_i^L(\sigma^T) \leq m \Big)   &~~\wedge & \Big( \sum_{\tau_i \in \tau(H)} U_i^H(\sigma^T) \leq m \Big)      & \nonumber
\end{align*}
\item[] \textbf{ILP}: 
A tighter \textit{necessary} condition, for a task set to be mixed-criticality 
schedulable at all. It is tested via our ILP (i.e., if it succeeds). 
It holds if and only if there exists an assignment of values to the $\sigma_i^L$ and $\sigma_i^H$ variables such that:

\begin{eqnarray}
\left( \sigma_i^L \leq \sigma^T,~\forall \tau_i \in \tau \right) \wedge \left( \sigma_i^L \leq \sigma_i^H \leq \sigma^T ,~\forall \tau_i \in \tau(H)\right) \wedge \nonumber\\
\Big( \sum_{\tau_i \in \tau} \sigma_i^L \leq \sigma^T \Big) \wedge \Big( \sum_{\tau_i \in \tau(H)} \sigma_i^H \leq \sigma^T \Big)                           \wedge \nonumber\\
\left( U_i^L(\sigma_i^L) \leq 1,~\forall \tau_i \in \tau \right) \wedge \left( U_i^H(\sigma_i^H) \leq 1,~\forall \tau_i \in \tau(H) \right)                 \wedge \nonumber\\
\Big( \sum_{\tau_i \in \tau} U_i^L(\sigma_i^L) \leq m \Big) \wedge \Big( \sum_{\tau_i \in \tau(H)} U_i^H(\sigma_i^H) \leq m \Big)~~                                \nonumber
\end{eqnarray}
\item[] \textbf{V-Ekb:} 
Similar to ``ILP'', but with the added constraint that $\sigma_i^L=\sigma_i^H,~\forall \tau_i \in \tau(H)$.
Hence, it is a \textit{necessary} condition for mixed-criticality schedulability for any approach that does not 
redistribute cache pages reclaimed from L-tasks to the H-tasks.
\item[] \textbf{Z-Ekb:}
This is a \textit{sufficient} test for partitioned scheduling using Ekberg and Yi's algorithm~\cite{Ekberg_Yi_12}, 
using the specified bin-packing, if the system is crippled by disabling of the last-level cache. In that case,
$\sigma_i^L=0,~\forall \tau_i \in \tau$ and similarly $\sigma_i^H=0,~\forall \tau_i \in \tau(H)$, meaning that
Ekberg and Yi's original analysis is applied, with $C_i^L=C_i^L(0)$ and $C_i^H=C_i^H(0)$. Intuitively ``Z-Ekb'' is meant as 
a lower-bound for the performance by this approach, once the cache is taken into account.
\item[] \textbf{E-Ekb:}
A \textit{sufficient} test for partitioned scheduling using Ekberg and Yi's algorithm,
using the specified bin-packing, when (i)~the cache is distributed equally to the tasks in the L-mode; i.e.,
$\sigma_i^L= \left\lfloor \frac{\sigma^T}{n}\right\rfloor,~\forall \tau_i \in \tau$ and (ii)~there is no redistribution
of cache pages, i.e., $\sigma_i^L=\sigma_i^H,~\forall \tau_i \in \tau(H)$.
Since Ekberg and Yi's algorithm is cache-agnostic, dividing the cache equally
is a simple, reasonable heuristic.
\item[] \textbf{N-Ekb:} Another \textit{sufficient} test, but which instead uses the output of the ILP, 
aimed at minimising L-mode utilisation ($\sum_{\tau_i \in \tau} U_i^L(\sigma_i^L)$), as
values to the $\sigma_i^L$ variables. Again, there is no cache redistribution at mode change, 
i.e., $\sigma_i^L=\sigma_i^H,~\forall \tau_i \in \tau(H)$.
Re-using the ILP solution for the L-mode both (i)~enables a fair comparison
(by equipping the ``opponent'' with the same good heuristic for the L-mode allocation, and
(ii)~takes the ILP and the L-mode allocation heuristic out of the equation, as much as possible, and isolates
the improvement originating from to the dynamic cache redistribution.
\item[] \textbf{Manberg:} A \textit{sufficient} test for our approach (named after Mancuso and Ekberg), which redistributes
cache pages at mode switch. The $\sigma_i^L$ and $\sigma_i^H$ variables are set to the respective
outputs of the ILP, under the heuristic that picks the $\sigma_i^L$ values that minimise the L-mode utilisation 
and, subject to that selection, the $\sigma_i^H$ values that minimise the H-mode utilisation.
\end{itemize}

Note that VT theoretically dominates ILP, which in turn theoretically dominates all other curves. Additionally, V-Ekb theoretically dominates all *-Ekb curves.

As observed (Fig.~\ref{fig:exp:weighted:1}--\ref{fig:exp:weighted:7}), the curves of
N-Ekb and V-Ekb almost coincide. This means that the heuristic of choosing the $\sigma_i^L$ values that minimise the 
L-mode utilisation is  very efficient, if cache redistribution at mode change is not permitted (bin-packing considerations aside\footnote{What 
we mean is that our experiments cannot possibly account for all possible task assignments,
since bin-packing is an NP-complete problem; we limit ourselves to just one reasonable and popular bin-packing heuristic.}).
Even when no dynamic cache reallocation is performed, the improvement just from using this particular heuristic, 
vs using the still reasonable ``E''-heuristic is respectable (Fig.~\ref{fig:exp:weighted:1}--\ref{fig:exp:weighted:7}).
Even if it only goes up to 10\% higher weighted schedulability (and usually around 2\% to 3\%),
it still means a large increase in the number of provably schedulable task sets. This is because in all plots except the two 
(Fig.~\ref{fig:exp:weighted:4} and~\ref{fig:exp:weighted:5}) in which the difference between N-Ekb and E-Ekb is greatest in absolute terms,
even the V-Ekb necessary test stays below 18\% in weighted schedulability. 
As can be seen, the relative improvement of N-Ekb over E-Ekb is more significant when the execution time is more sensitive to 
the cache resources (Fig.~\ref{fig:exp:weighted:1}) or the scarcer the latter are (Fig.~\ref{fig:exp:weighted:2}).

As for our approach (Manberg), in all experiments it outperforms V-Ekb, 
i.e., what is theoretically possible without cache redistribution (bin-packing considerations aside).
It is often nearer to the ILP curve (a necessary schedulability test, for any algorithm, with cache redistribution permitted)
than to the V-Ekb curve. The absolute improvement in weighted schedulability is up to 3.64\% but, in relative terms it
can be up to 30.59\%. This means many more tasks sets (especially among higher-utilisation ones) found schedulable.
Table~\ref{tb:improvement} summarises the improvement. Note that the more sensitive to cache allocations the WCETs are,
the greater the relative gains by our approach (Fig.~\ref{fig:exp:weighted:1}). 
Similarly when the $C_i^H$/$C_i^L$ ratio is higher, i.e., when the mode switch is harder 
to accommodate (Fig.~\ref{fig:exp:weighted:4}). We believe that these findings validate our approach.

(Note that two near-zero negative values in Table~\ref{tb:improvement} reflect the fact that V-Ekb
is a necessary test, whose success does \textit{not} imply that bin-packing will be feasible!)

In our experiments, the ILP run-time was a few seconds (upto four seconds for the feasible solutions), 
but the deadline-scaling (which repeatedly invokes the schedulability test) took 23 hours for 6000 task sets.

\begin{table}
\centering
\begin{tabular}{|c|c|c|}
\hline
 ~                                                         &  \multicolumn{2}{c|}{Improvement (Manberg vs V-Ekb)}     \\\hline
 Experiment / parameter varied                             &         absolute         &           relative            \\
 \hline \hline
 Task-set size ($n$)                                       &    1.36\%--2.36\%        &    13.98\%--30.59\%           \\\hline
 Fraction of H-tasks in $\tau$                             &    0.13\%--2.34\%        &     6.18\%--16.05\%           \\\hline
 Ratio of $C_i^H$ to $C_i^L$                               &    1.41\%--3.64\%        &    10.24\%--27.13\%           \\\hline
 L. bound  $\alpha$ on $C_i^L(0)/C_i^H(0)$                 &    0.25\%--1.65\%        &     3.36\%--15.62\%           \\\hline
 $\lambda$                                                 &    1.60\%--2.32\%        &    12.95\%--20.25\%           \\\hline
 Cache size                                                &    1.64\%--2.23\%        &    12.67\%--15.81\%           \\\hline
 Number of cores ($m$)                                     &   -0.00\%--1.21\%        &    -0.67\%--17.45\%           \\\hline
 \end{tabular} 
\caption{Improvement in weighted schedulability}
\label{tb:improvement}
\end{table}

\begin{figure*}[!t]
\begin{minipage}{1\linewidth}
\centering
\begin{minipage}[b]{0.47\linewidth}
\centering
\includegraphics[width=0.95\columnwidth]{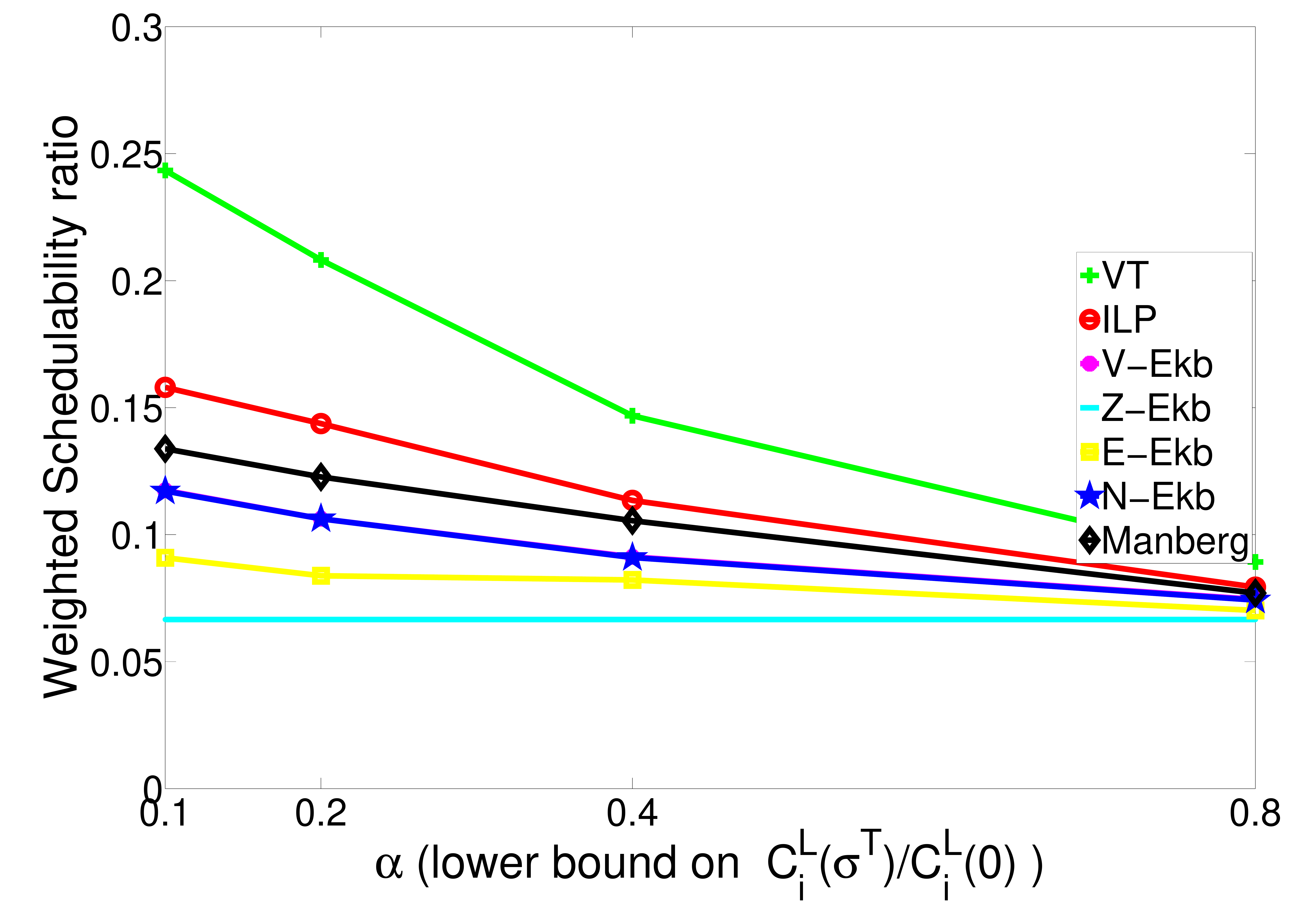}
  \caption{\label{fig:exp:weighted:1}}
\end{minipage}
\begin{minipage}[b]{0.47\linewidth}
\centering
\includegraphics[width=0.95\columnwidth]{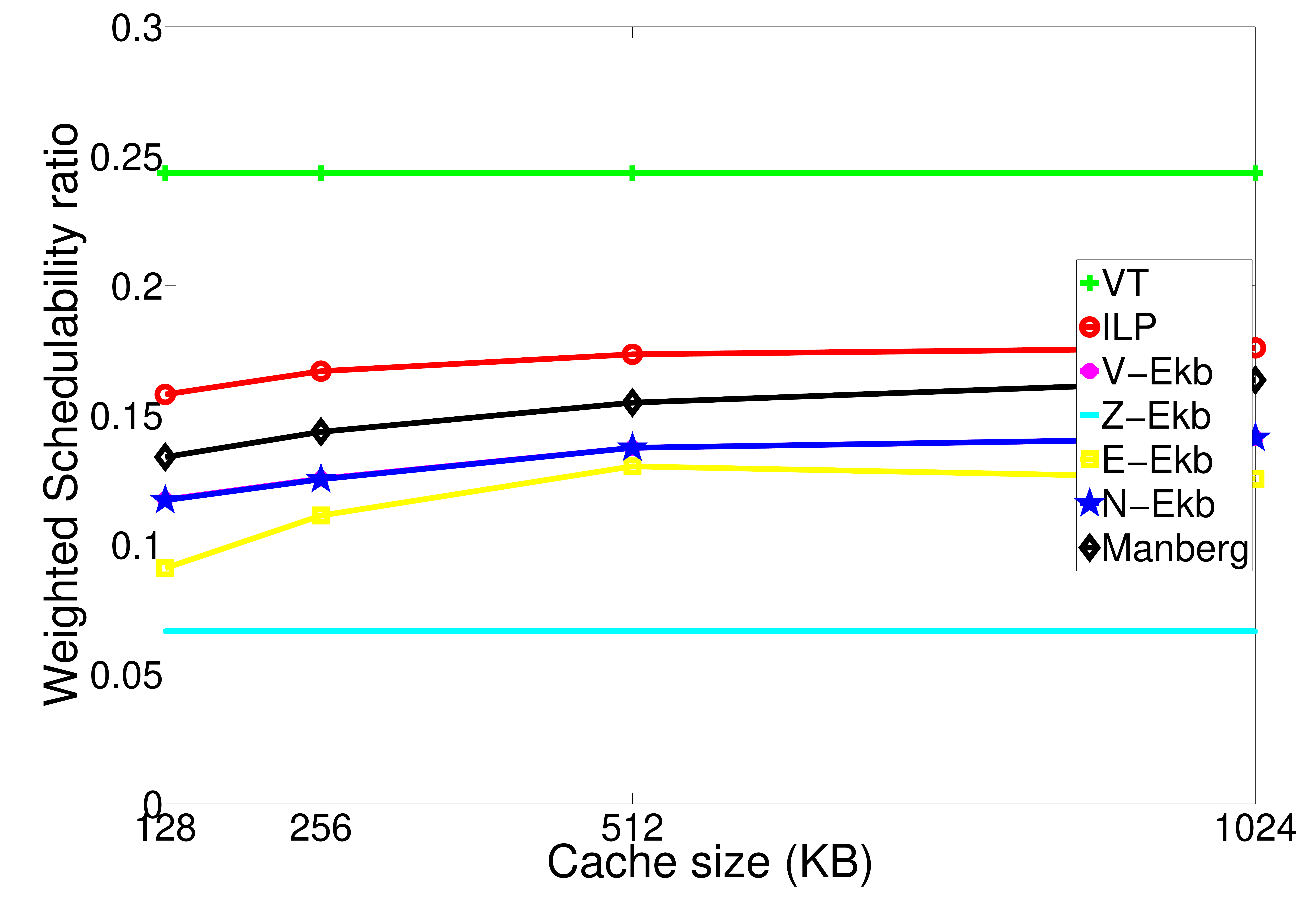}
  \caption{\label{fig:exp:weighted:2}}
\end{minipage}
\end{minipage}
\begin{minipage}{1\linewidth}
\centering
\begin{minipage}[b]{0.47\linewidth}
\centering
\includegraphics[width=0.95\columnwidth]{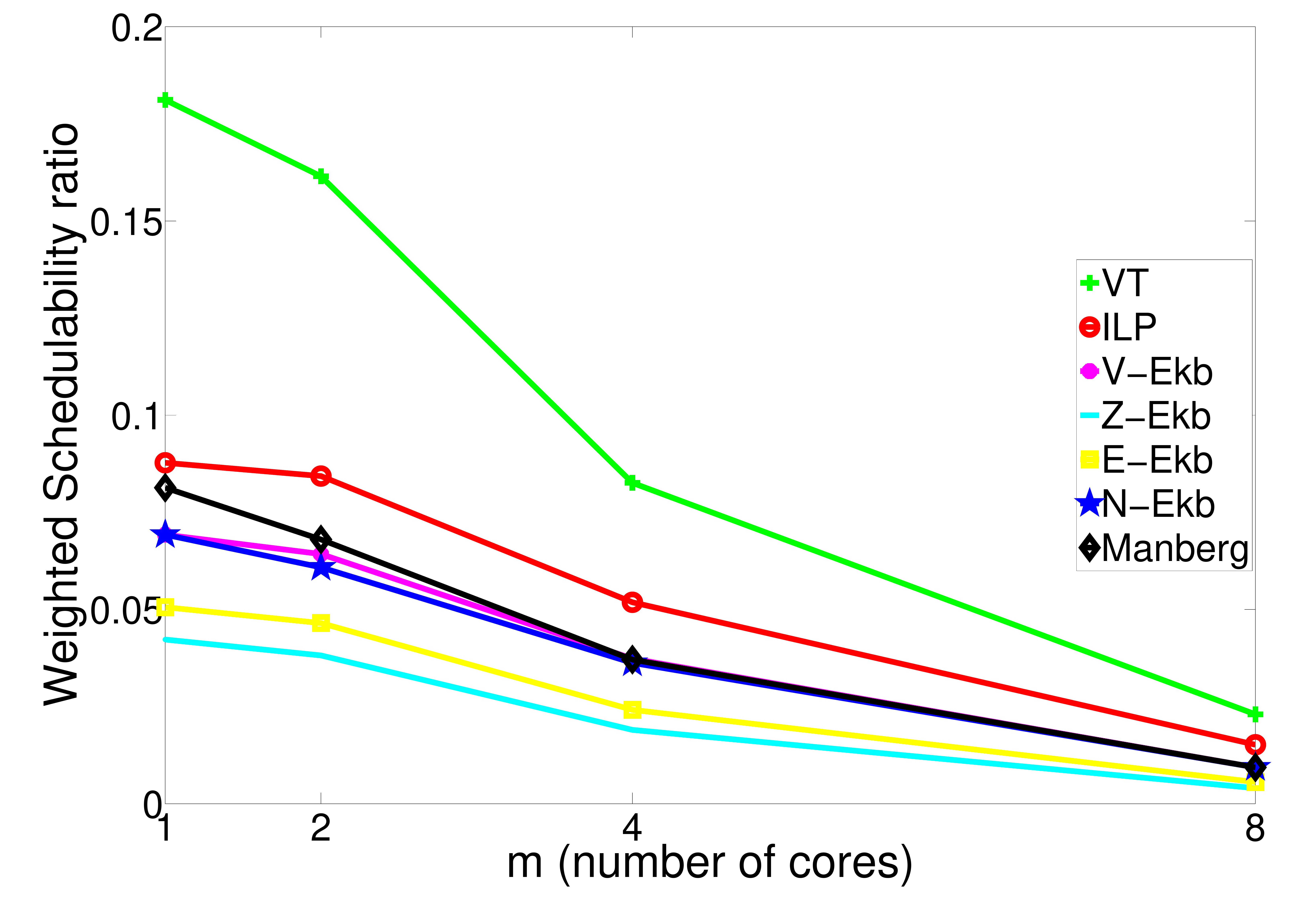}
  \caption{\label{fig:exp:weighted:3}}
\end{minipage}
\begin{minipage}[b]{0.47\linewidth}
\centering
\includegraphics[width=0.95\columnwidth]{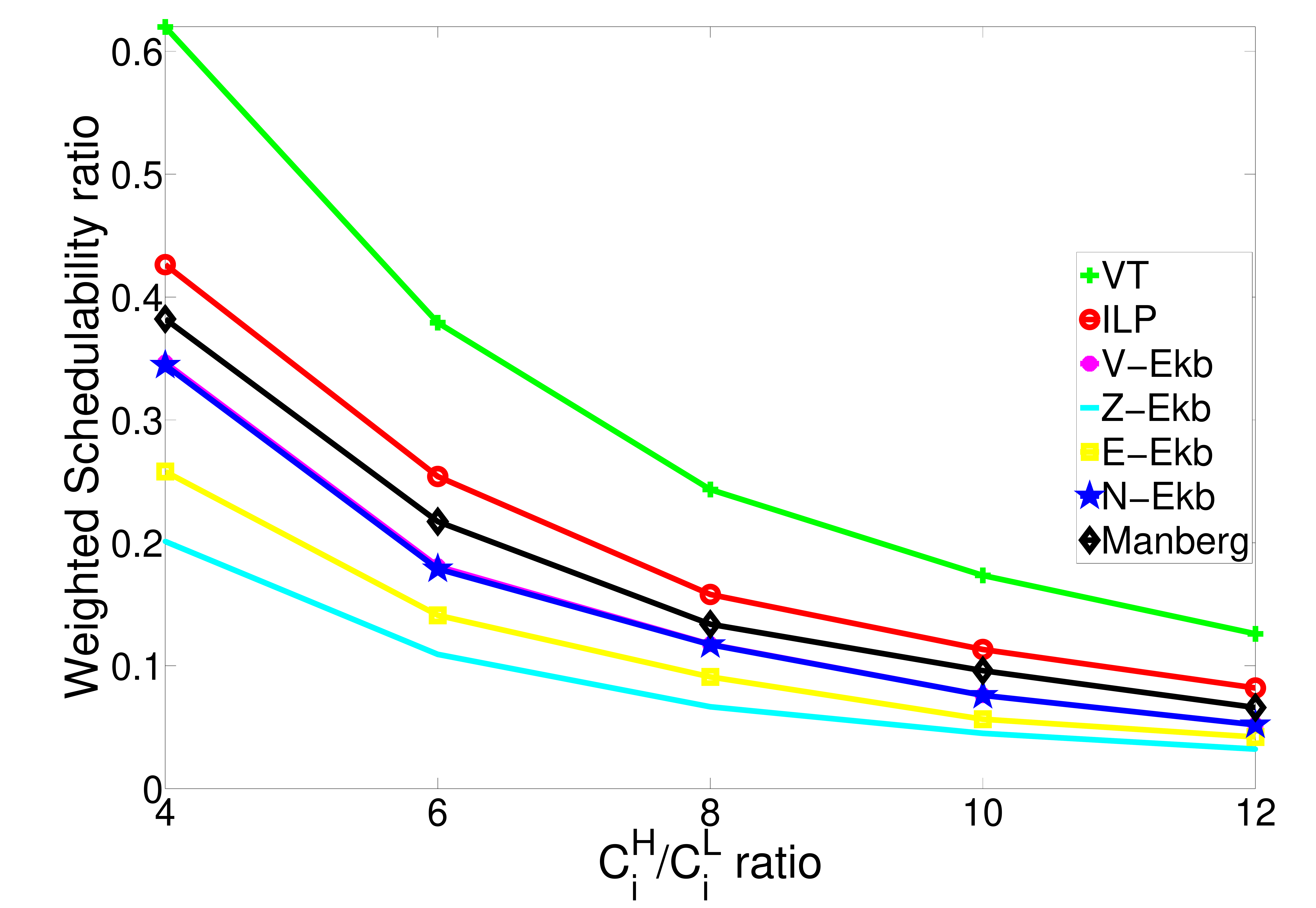}
  \caption{\label{fig:exp:weighted:4}}
\end{minipage}
\end{minipage}
\begin{minipage}{1\linewidth}
\centering
\begin{minipage}[b]{0.47\linewidth}
\centering
\includegraphics[width=0.95\columnwidth]{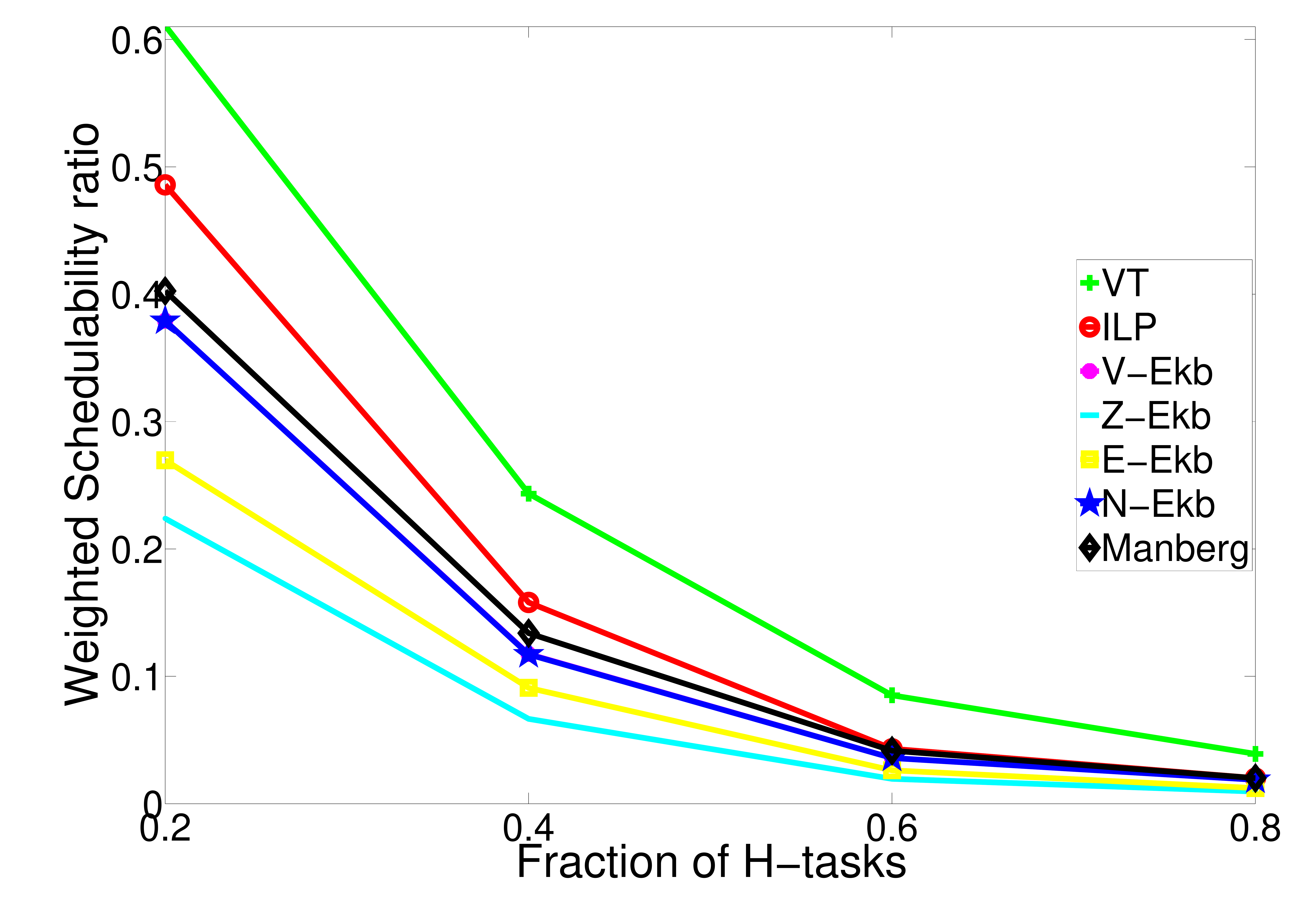}
  \caption{\label{fig:exp:weighted:5}}
\end{minipage}
\begin{minipage}[b]{0.47\linewidth}
\centering
\includegraphics[width=0.95\columnwidth]{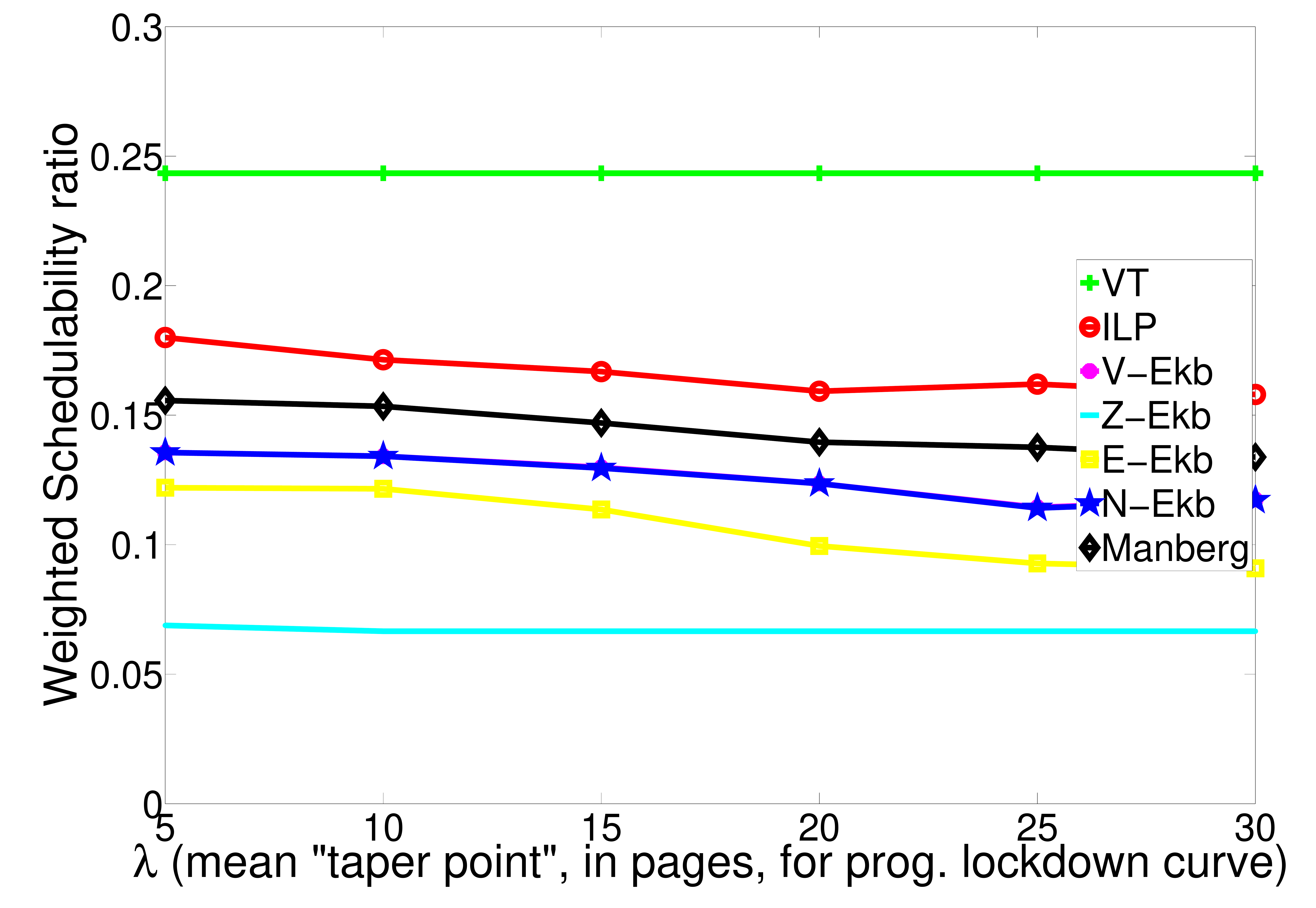}
  \caption{\label{fig:exp:weighted:6}}
\end{minipage}
\end{minipage}
\end{figure*}

\begin{figure}
\centering
\begin{minipage}[b]{0.47\linewidth}
\centering
 \includegraphics[width=0.95\columnwidth]{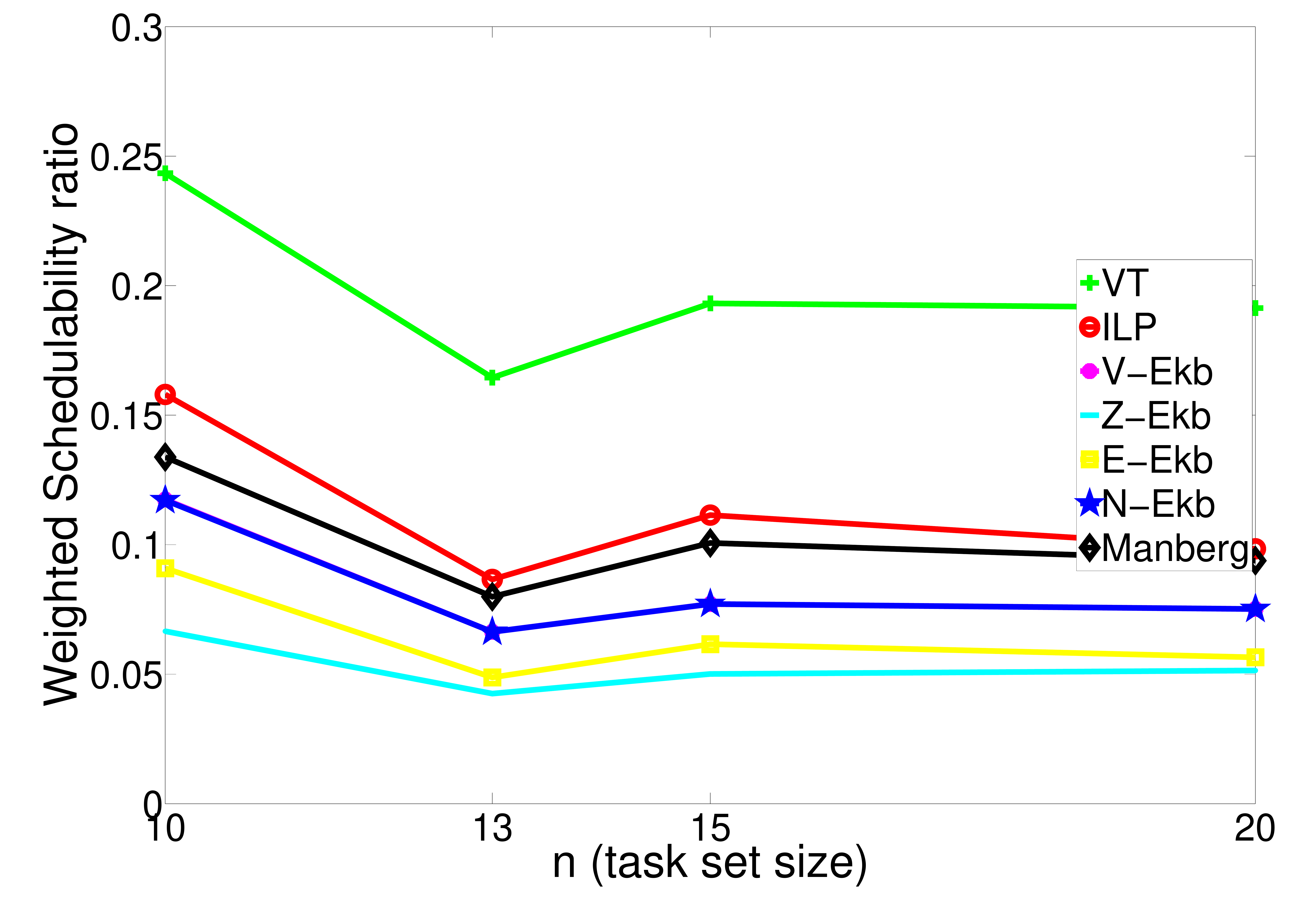}
 \caption{\label{fig:exp:weighted:7}}
\end{minipage}
\end{figure}
 
\section{Conclusions\label{s:conclusions}}
In this work, we proposed the redistribution of resources from low-criticality tasks to high-criticality tasks,
at mode change, for better scheduling performance. Focusing on one particular resource, the last-level cache, 
we formulated analysis and showed the potential gains. This validates the notion that more detailed models
of the platform and the allocation of its resources, can be used to improve both the performance and 
the confidence in the analysis of mixed-criticality systems.
In the future, we plan to consider additional system resources. 
We also intend to explore efficient non-ILP-based cache allocation heuristics.

%%
%% Bibliography
%%

%% Either use bibtex (recommended), 

%\bibliography{paper_TechnicalReport}

%% .. or use the thebibliography environment explicitely

\appendix
We now present the figures for the schedulability ratio of all the parameters discussed above in Section~\ref{s:evaluation}. 
%%% first set 
\begin{figure*}[b]
\begin{minipage}{1\linewidth}
\centering
\begin{minipage}[b]{0.49\linewidth}
\centering
\includegraphics[width=0.99\columnwidth]{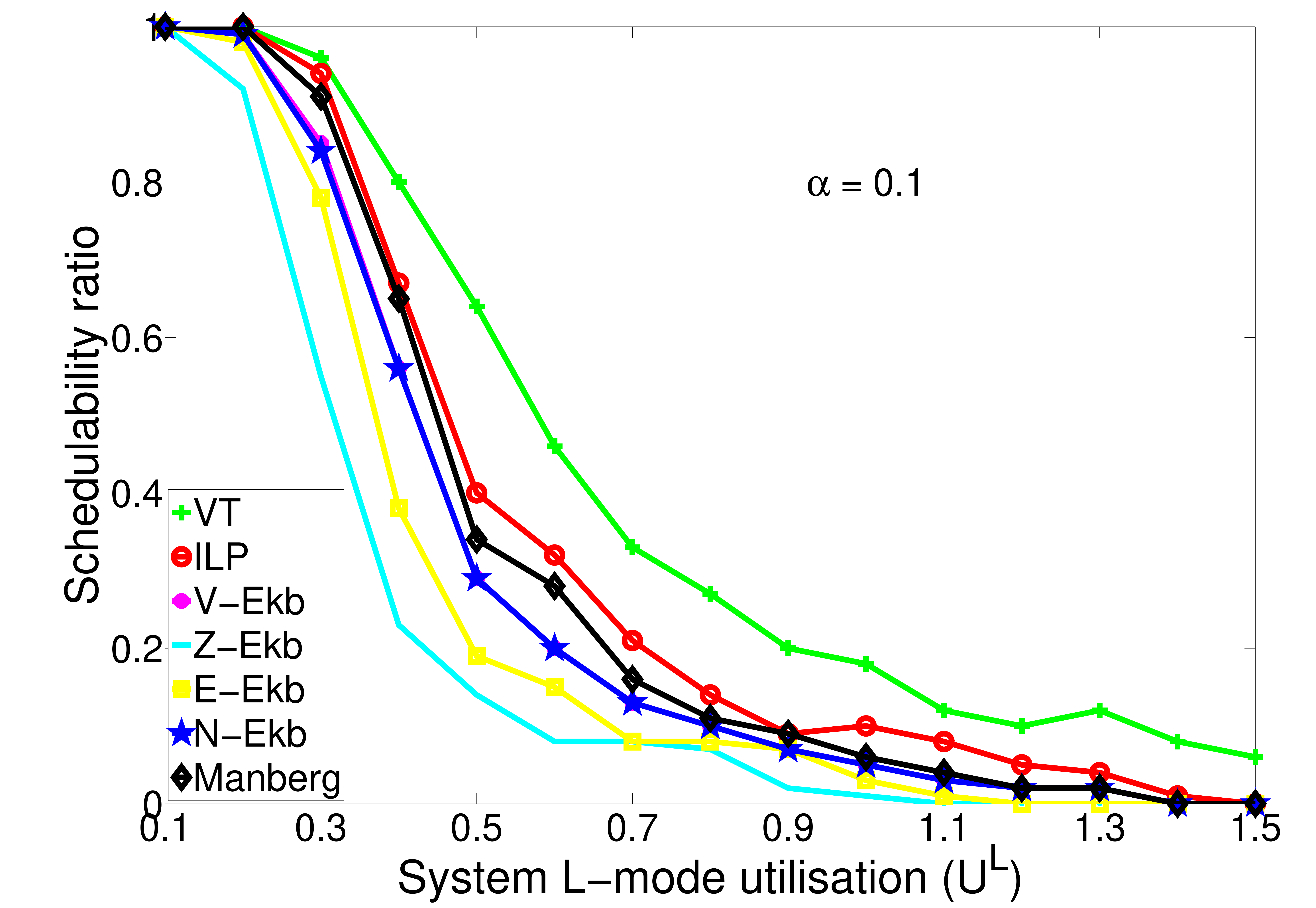}
  %\vspace{-3ex}
  \caption{\label{fig:exp:SR:8}}
  \vspace{3ex}
\end{minipage}
\hspace{-3ex}
\begin{minipage}[b]{0.49\linewidth}
\centering
\includegraphics[width=0.99\columnwidth]{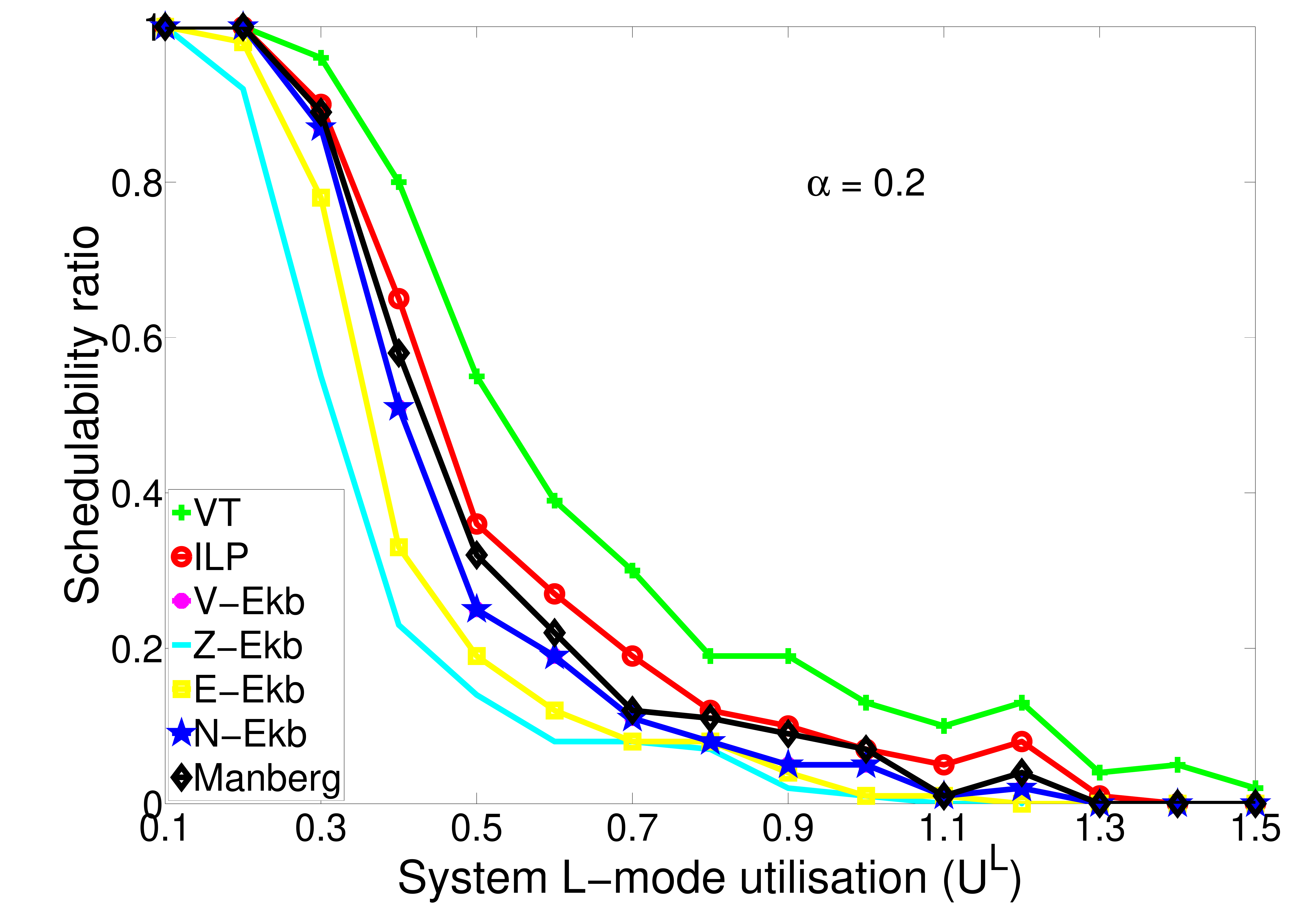}
  %\vspace{-3ex}
  \caption{\label{fig:exp:SR:9}}
  \vspace{3ex}
\end{minipage}
\end{minipage}
\end{figure*}

\begin{figure*}
\begin{minipage}{1\linewidth}
\centering
\begin{minipage}[b]{0.49\linewidth}
\centering
\includegraphics[width=0.99\columnwidth]{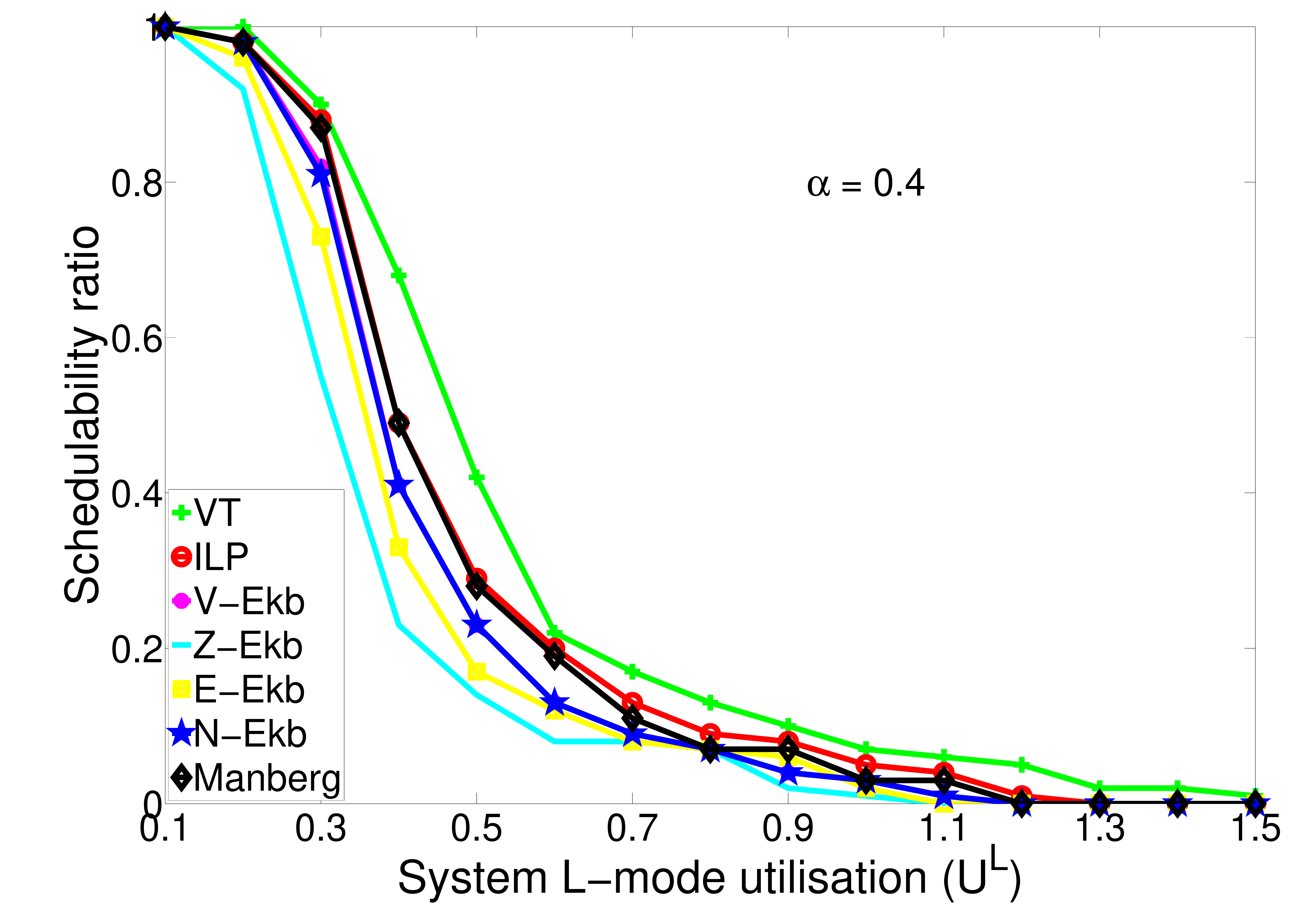}
  %\vspace{-3ex}
  \caption{\label{fig:exp:SR:10}}
\end{minipage}
\hspace{-3ex}
\begin{minipage}[b]{0.49\linewidth}
\centering
\includegraphics[width=0.99\columnwidth]{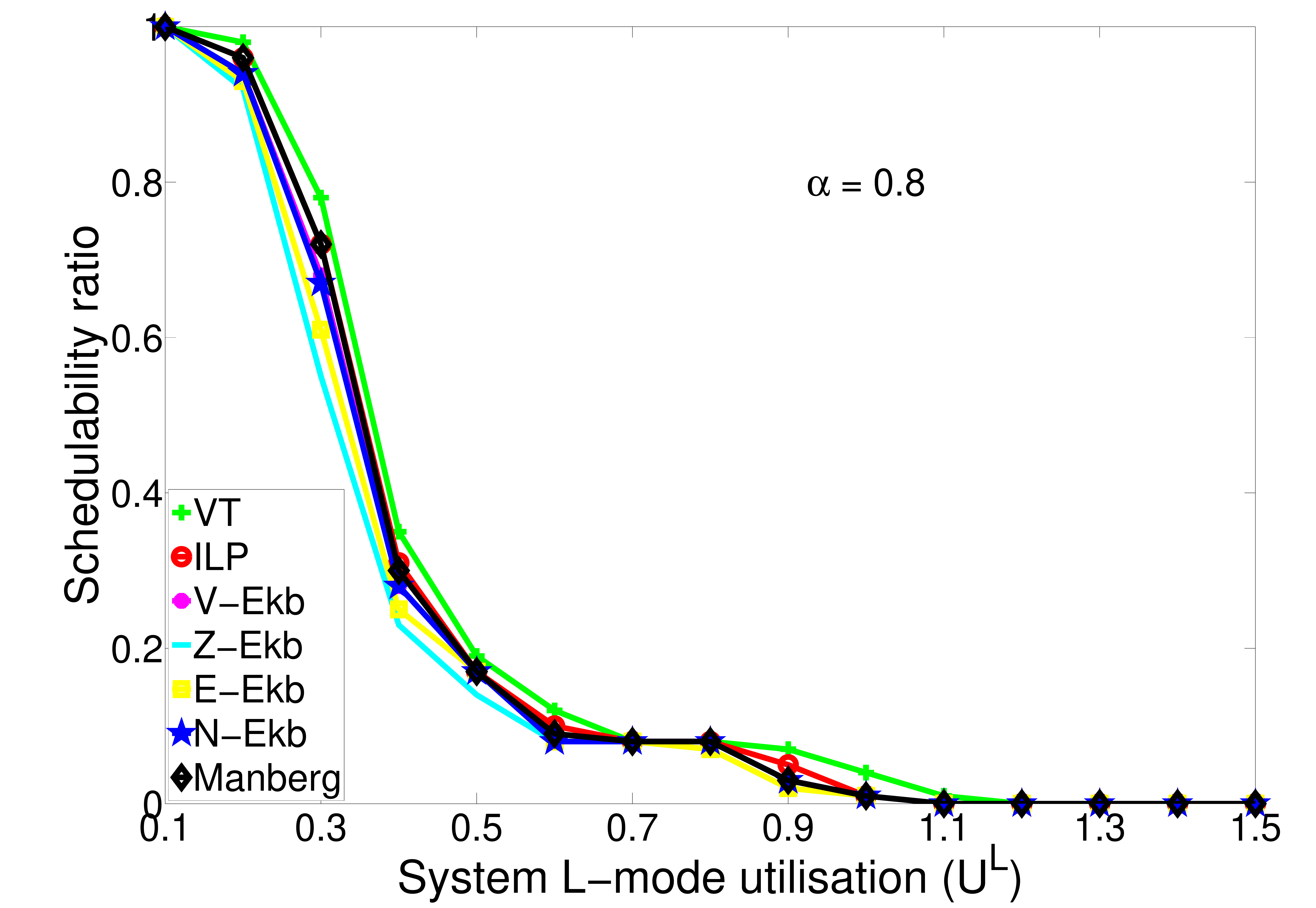}
  %\vspace{-3ex}
  \caption{\label{fig:exp:SR:11}}
\end{minipage}
\end{minipage}
\end{figure*}

\begin{figure*}
\begin{minipage}{1\linewidth}
\centering
\begin{minipage}[b]{0.49\linewidth}
\centering
\includegraphics[width=0.99\columnwidth]{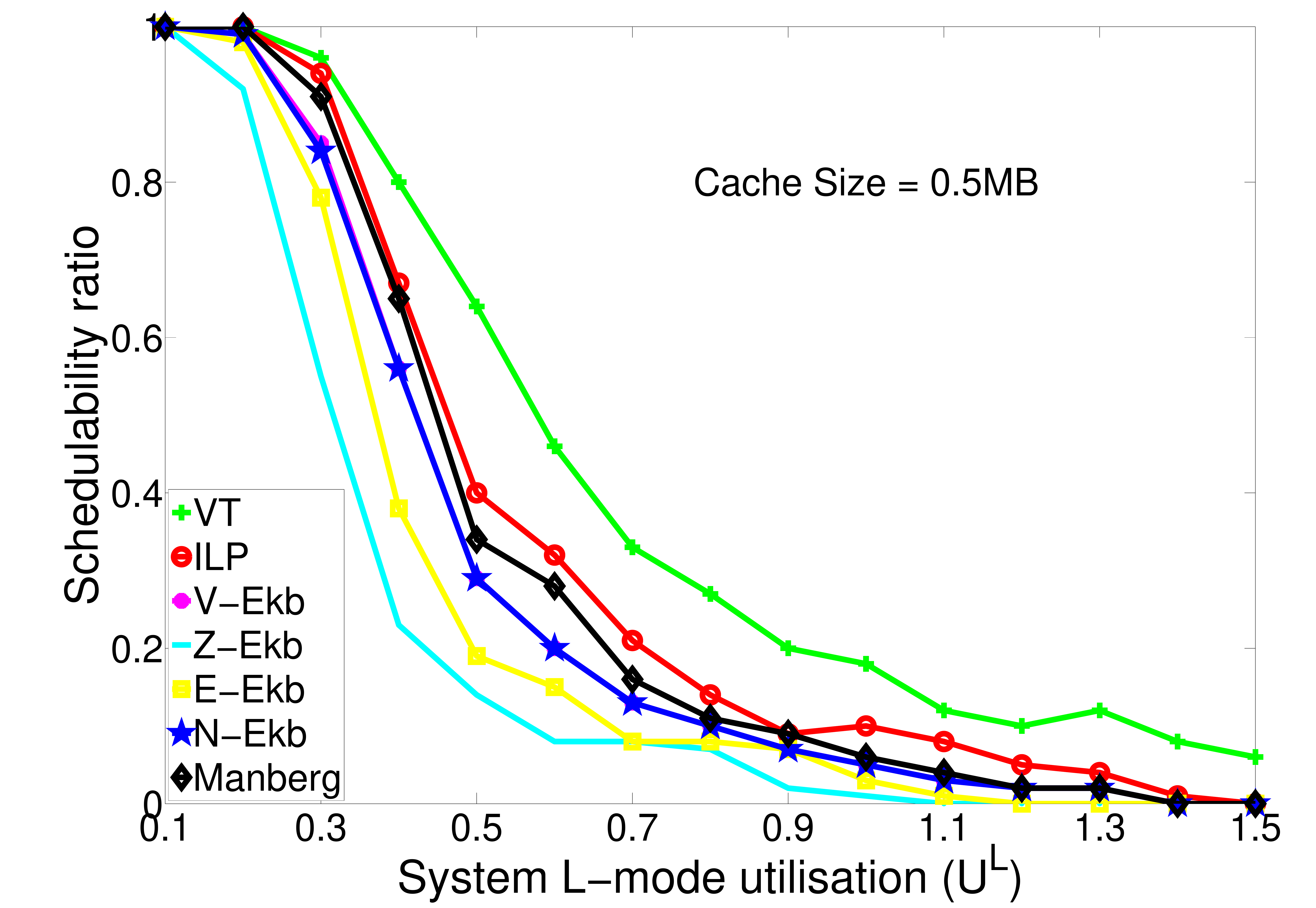}
  %\vspace{-3ex}
  \caption{\label{fig:exp:SR:12}}
\end{minipage}
\hspace{-3ex}
\begin{minipage}[b]{0.49\linewidth}
\centering
\includegraphics[width=0.99\columnwidth]{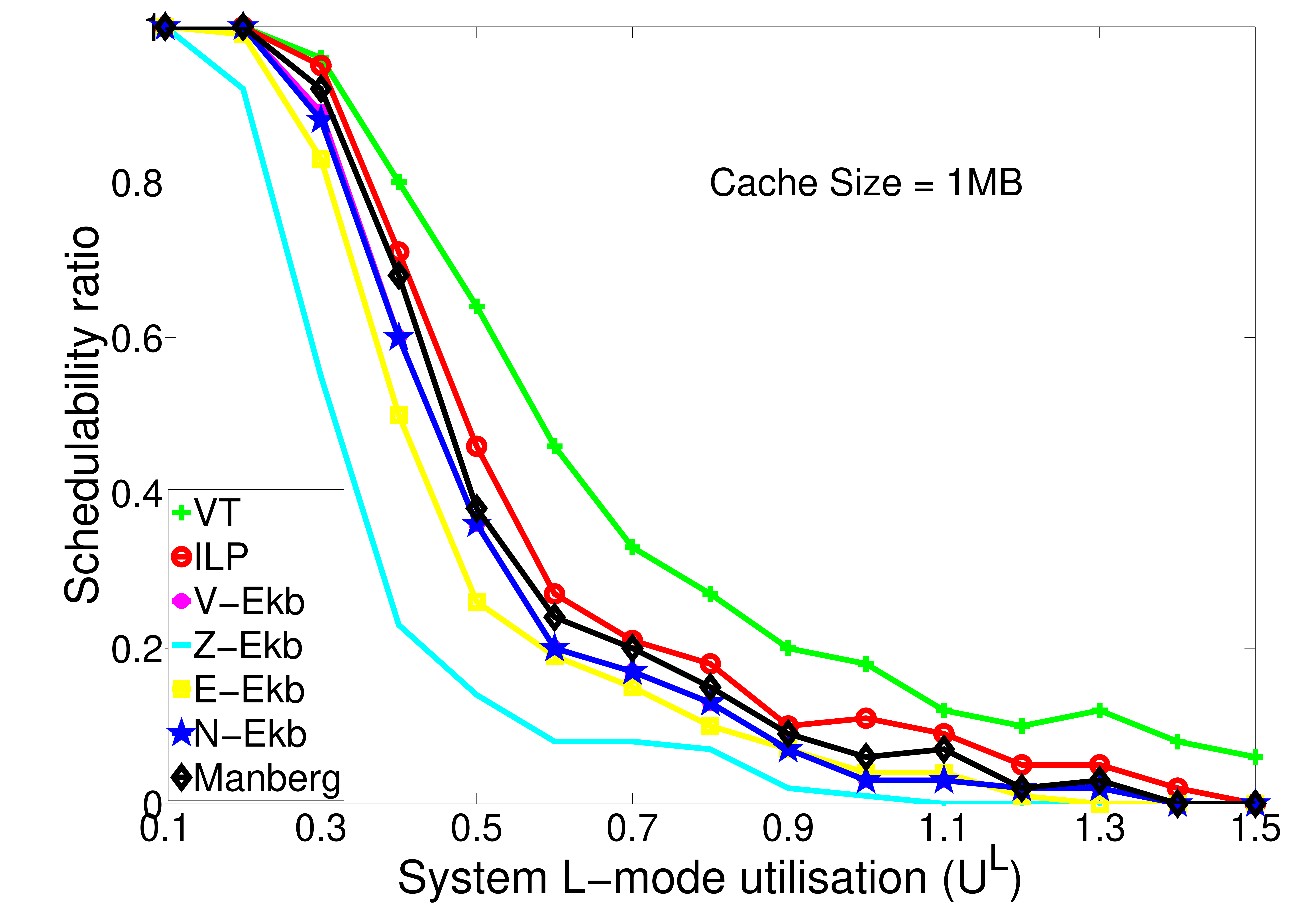}
  %\vspace{-3ex}
  \caption{\label{fig:exp:SR:13}}
\end{minipage}
\end{minipage}
\end{figure*}

\begin{figure*}
\begin{minipage}{1\linewidth}
\centering
\begin{minipage}[b]{0.49\linewidth}
\centering
\includegraphics[width=0.99\columnwidth]{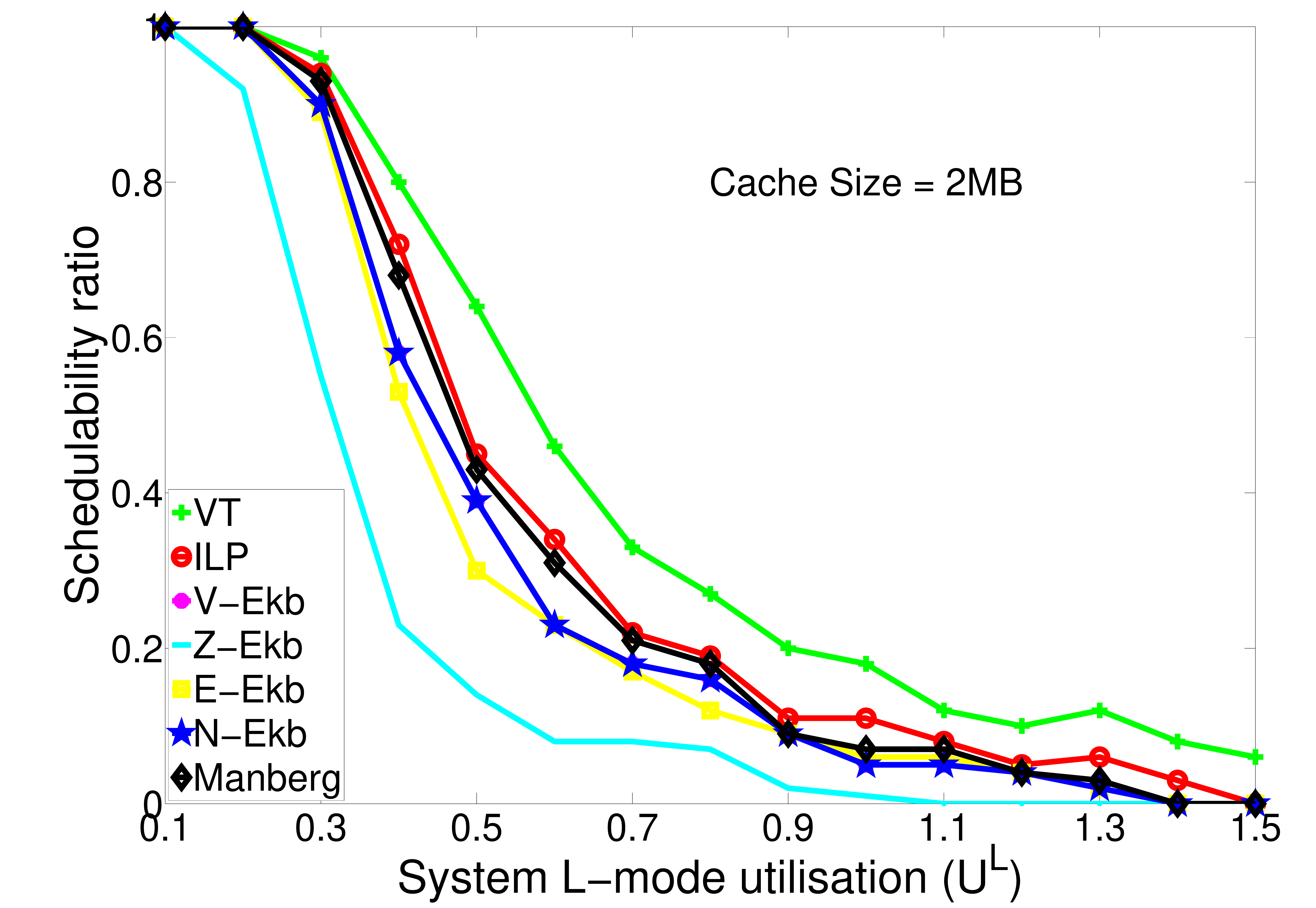}
  %\vspace{-3ex}
  \caption{\label{fig:exp:SR:14}}
  \vspace{3ex}
\end{minipage}
\hspace{-3ex}
\begin{minipage}[b]{0.49\linewidth}
\centering
\includegraphics[width=0.99\columnwidth]{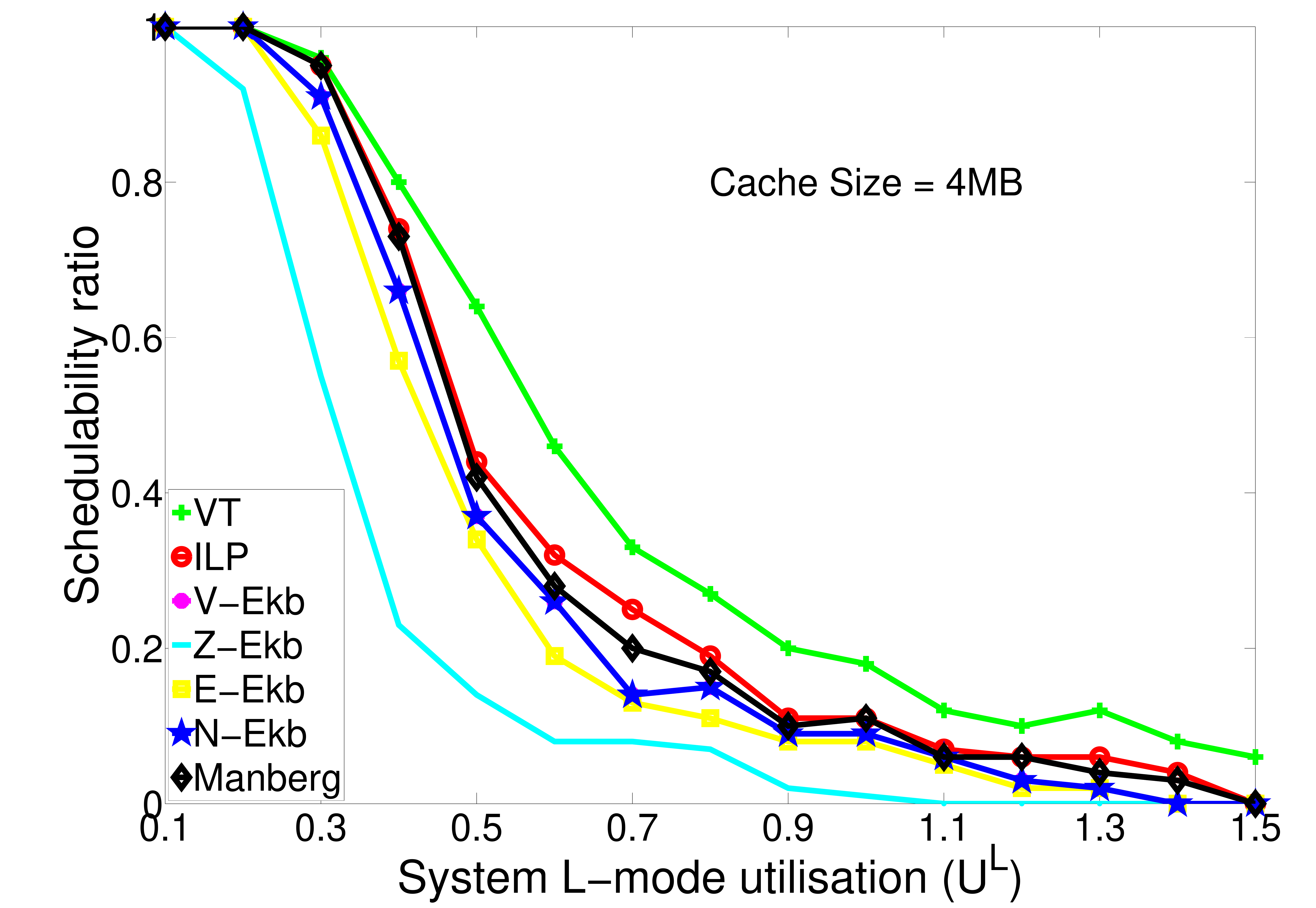}
  %\vspace{-3ex}
  \caption{\label{fig:exp:SR:15}}
  \vspace{3ex}
\end{minipage}
\end{minipage}
%\vspace{-1ex}
\end{figure*}

%%% second set 
\begin{figure*}[!t]
\begin{minipage}{1\linewidth}
\centering
\begin{minipage}[b]{0.49\linewidth}
\centering
\includegraphics[width=0.99\columnwidth]{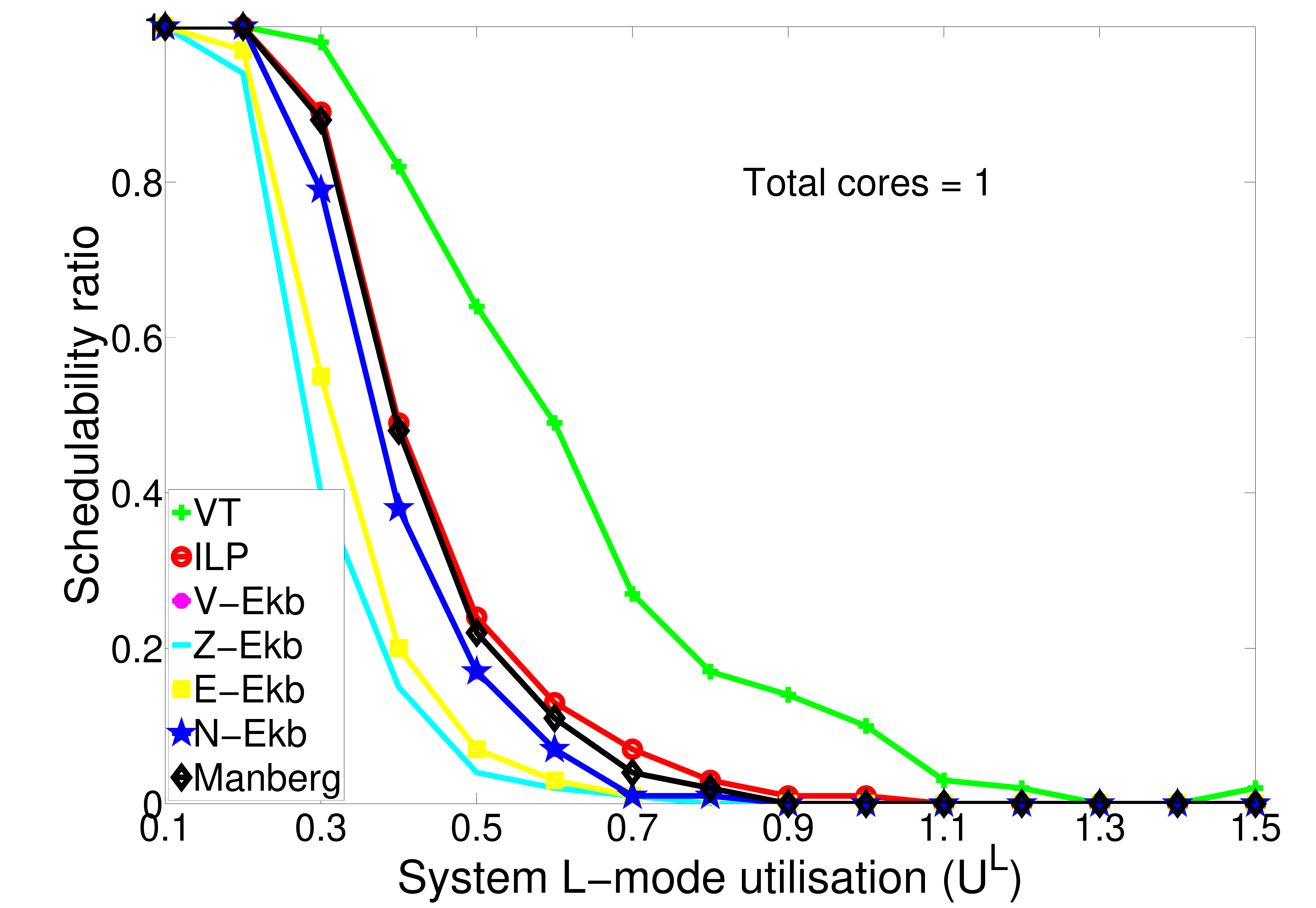}
  %\vspace{-3ex}
  \caption{\label{fig:exp:SR:16}}
\end{minipage}
\hspace{-3ex}
\begin{minipage}[b]{0.49\linewidth}
\centering
\includegraphics[width=0.99\columnwidth]{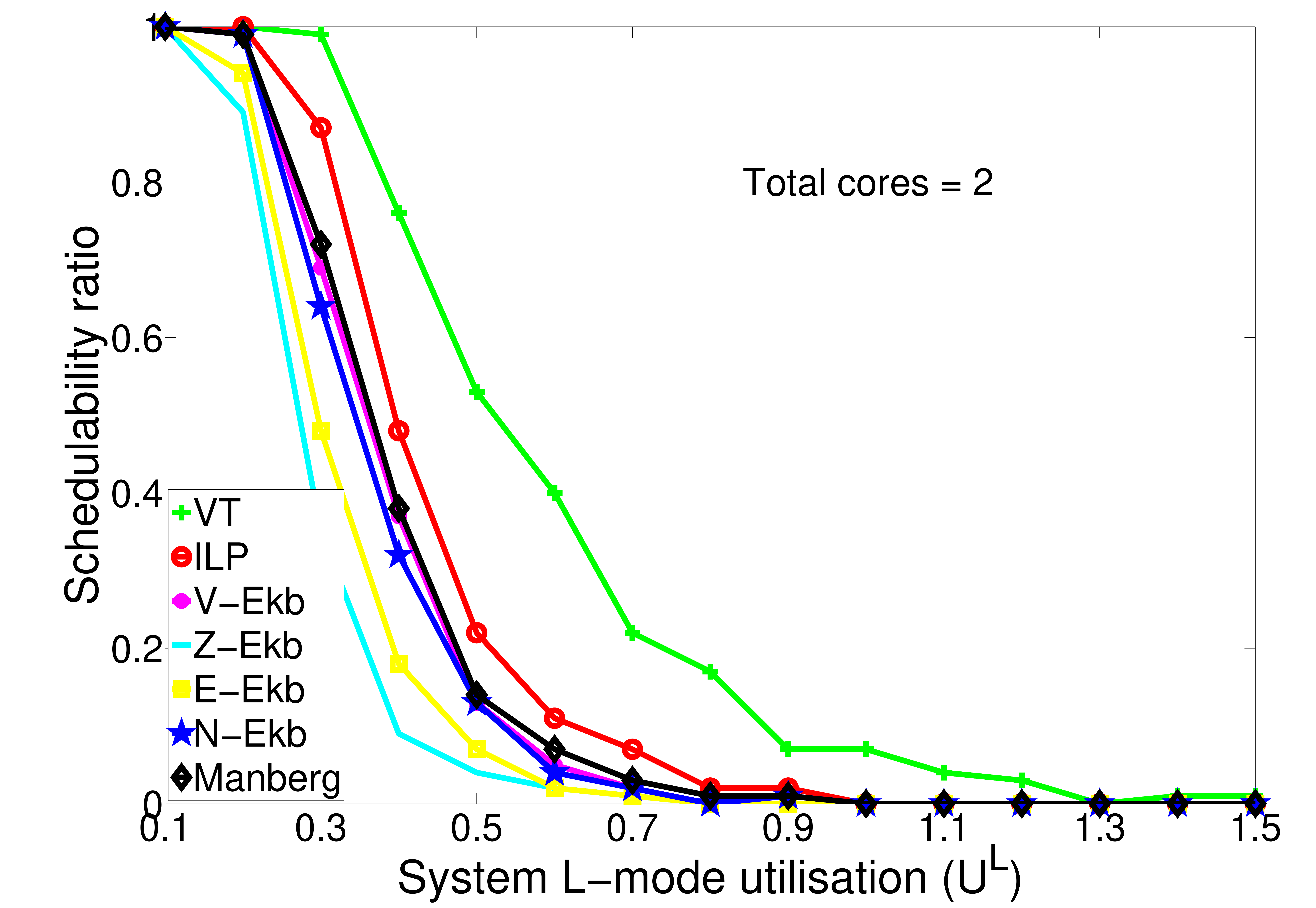}
  %\vspace{-3ex}
  \caption{\label{fig:exp:SR:17}}
\end{minipage}
\end{minipage}
\end{figure*}

\begin{figure*}
\begin{minipage}{1\linewidth}
\centering
\begin{minipage}[b]{0.49\linewidth}
\centering
\includegraphics[width=0.99\columnwidth]{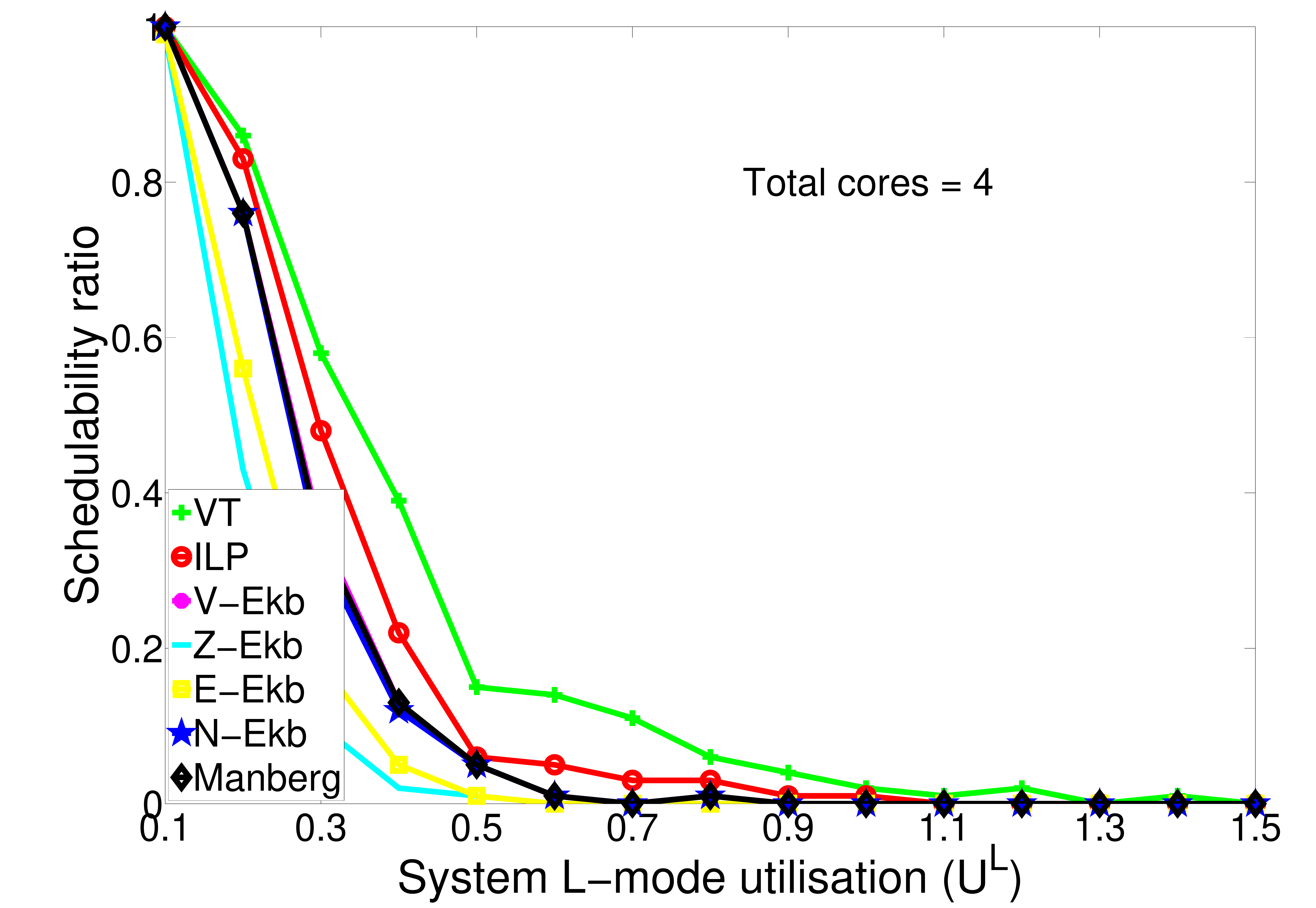}
  %\vspace{-3ex}
  \caption{\label{fig:exp:SR:18}}
\end{minipage}
\hspace{-3ex}
\begin{minipage}[b]{0.49\linewidth}
\centering
\includegraphics[width=0.99\columnwidth]{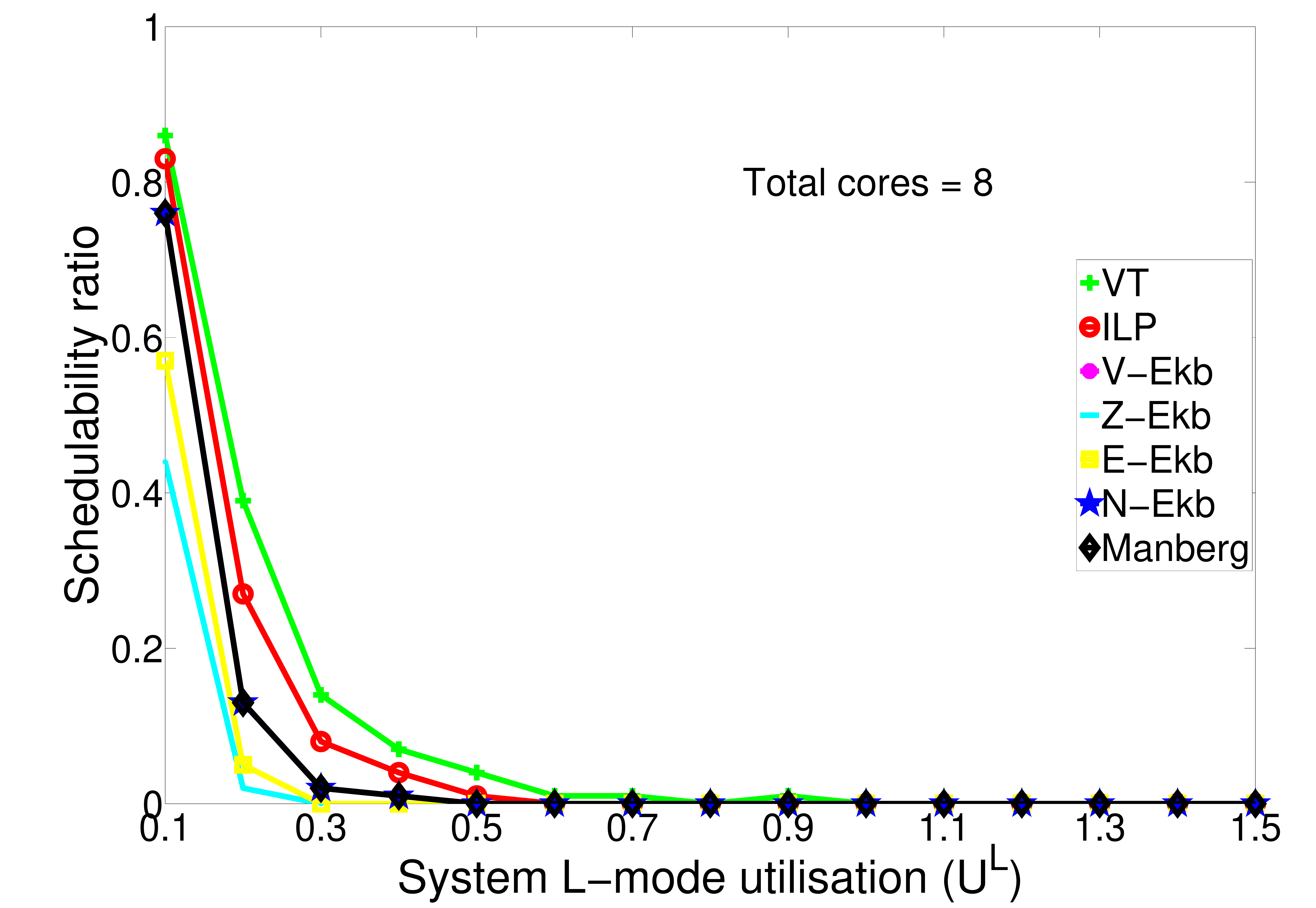}
  %\vspace{-3ex}
  \caption{\label{fig:exp:SR:19}}
\end{minipage}
\end{minipage}
\end{figure*}

\begin{figure*}
\begin{minipage}{1\linewidth}
\centering
\begin{minipage}[b]{0.49\linewidth}
\centering
\includegraphics[width=0.99\columnwidth]{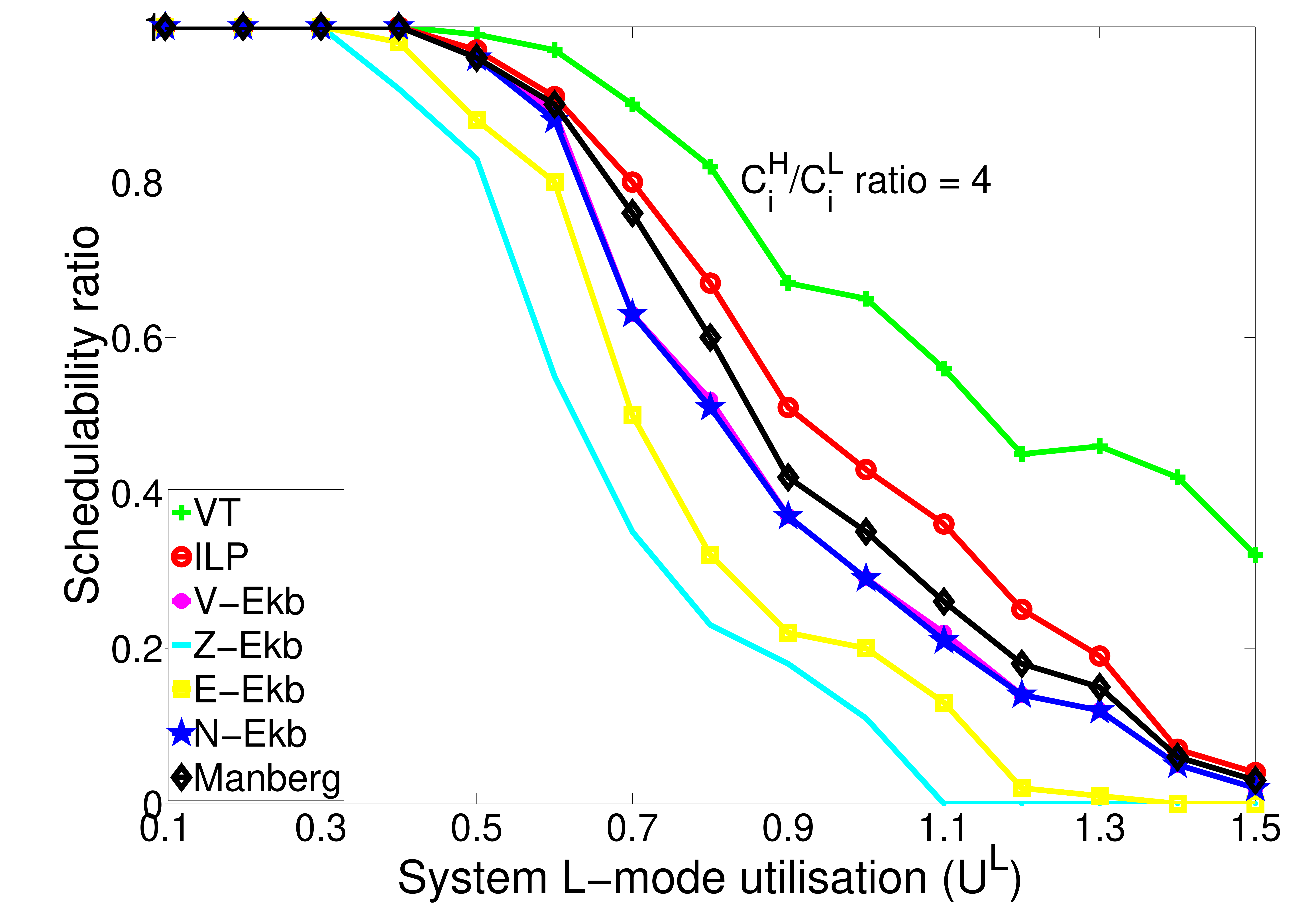}
  %\vspace{-3ex}
  \caption{\label{fig:exp:SR:20}}
  \vspace{3ex}
\end{minipage}
\hspace{-3ex}
\begin{minipage}[b]{0.49\linewidth}
\centering
\includegraphics[width=0.99\columnwidth]{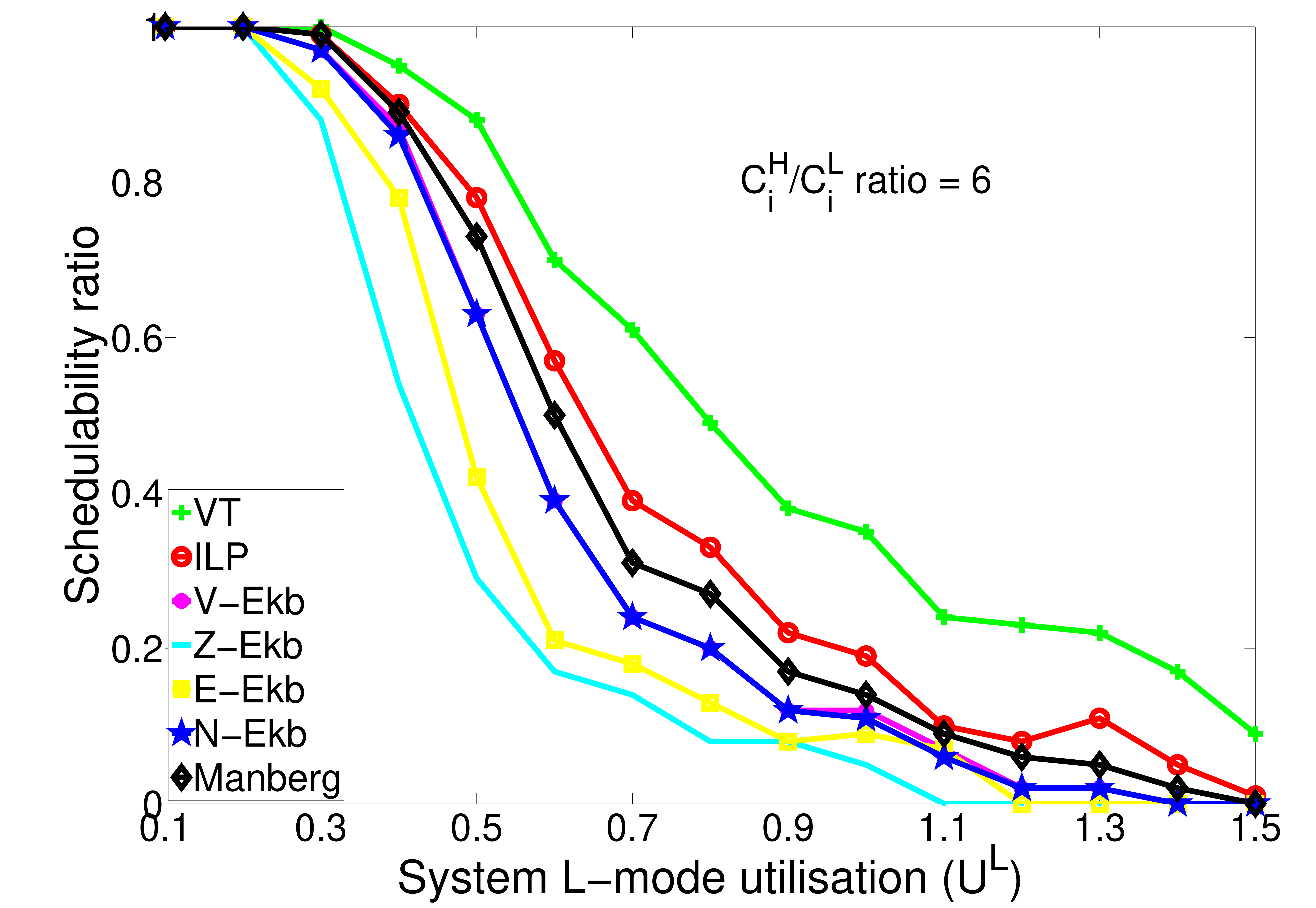}
  %\vspace{-3ex}
  \caption{\label{fig:exp:SR:21}}
  \vspace{3ex}
\end{minipage}
\end{minipage}
\end{figure*}

\begin{figure*}
\begin{minipage}{1\linewidth}
\centering
\begin{minipage}[b]{0.49\linewidth}
\centering
\includegraphics[width=0.99\columnwidth]{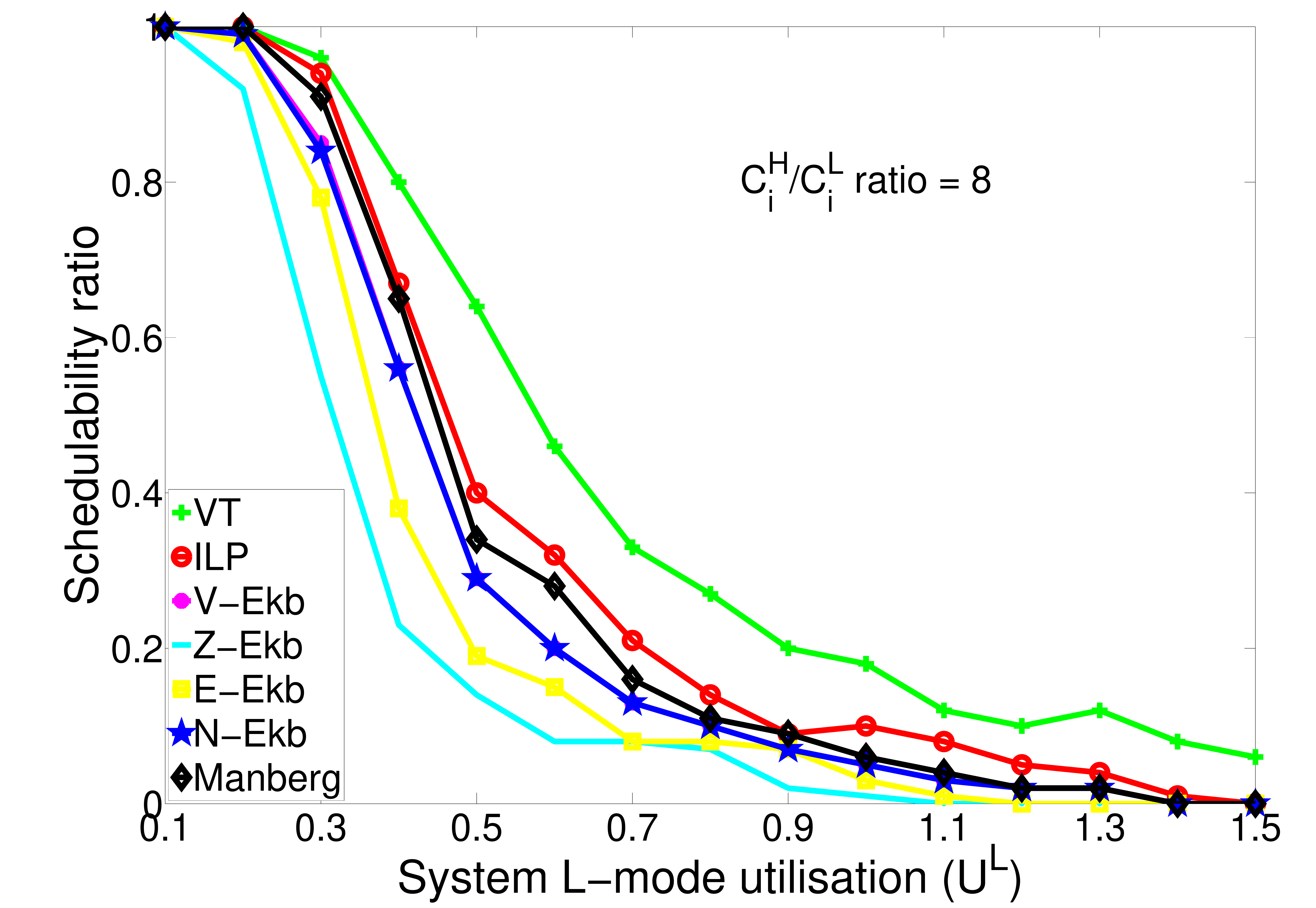}
  %\vspace{-3ex}
  \caption{\label{fig:exp:SR:22}}
\end{minipage}
\hspace{-3ex}
\begin{minipage}[b]{0.49\linewidth}
\centering
\includegraphics[width=0.99\columnwidth]{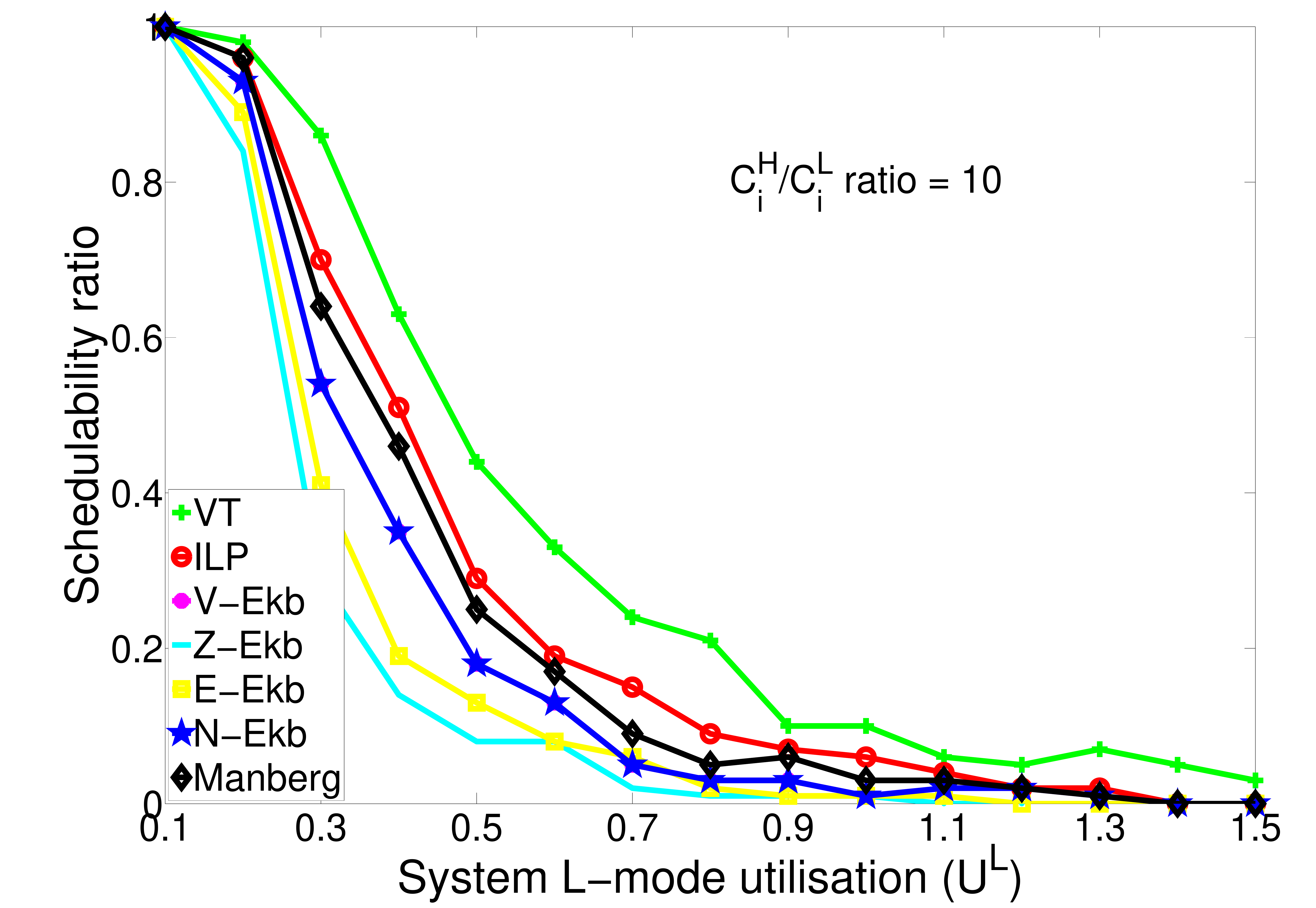}
  %\vspace{-3ex}
  \caption{\label{fig:exp:SR:23}}
\end{minipage}
\end{minipage}
%\vspace{-1ex}
\end{figure*}

%%% third set 
\begin{figure*}[!t]
\begin{minipage}{1\linewidth}
\centering
\begin{minipage}[b]{0.49\linewidth}
\centering
\includegraphics[width=0.99\columnwidth]{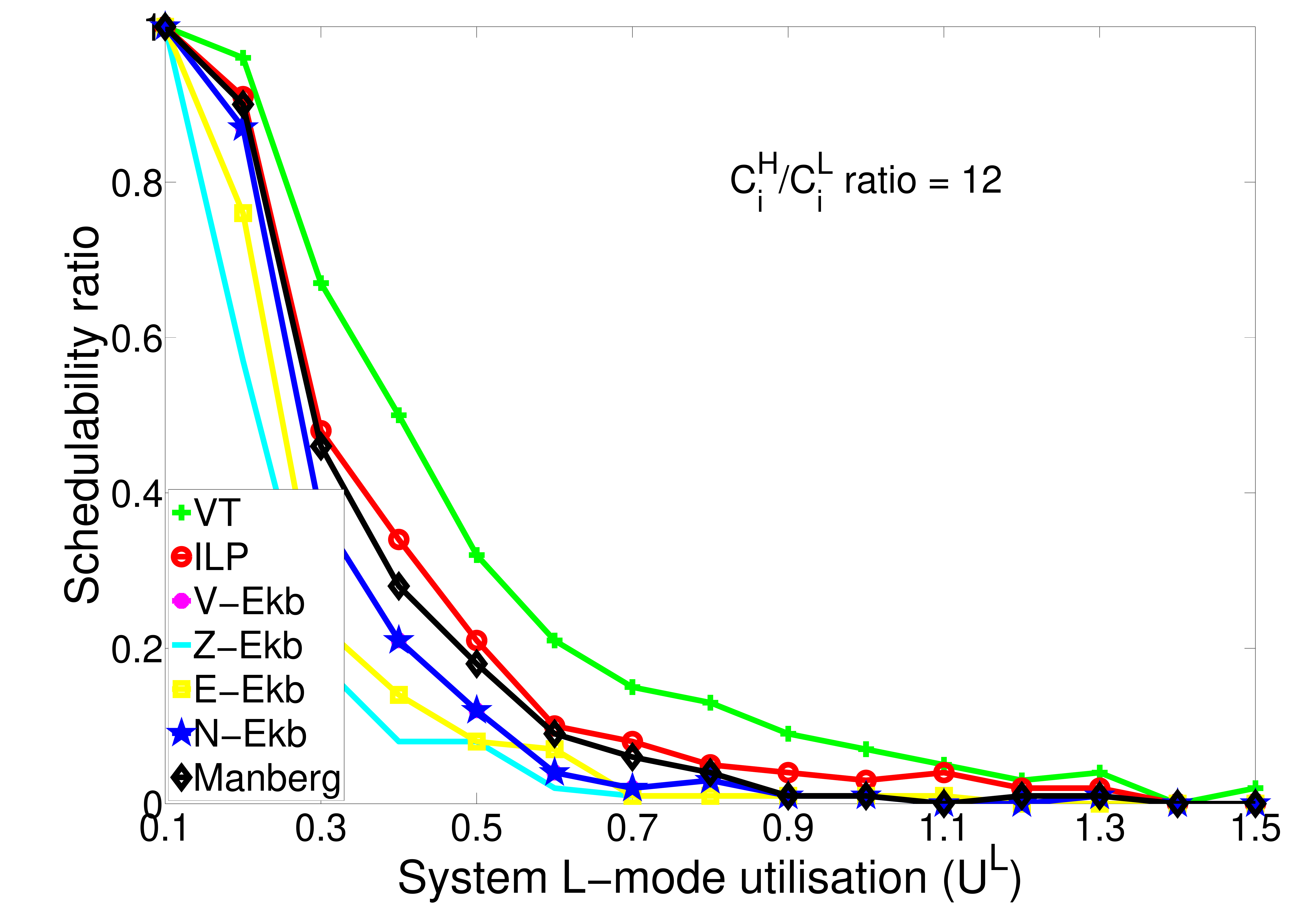}
  %\vspace{-3ex}
  \caption{\label{fig:exp:SR:24}}
\end{minipage}
\hspace{-3ex}
\begin{minipage}[b]{0.49\linewidth}
\centering
\includegraphics[width=0.99\columnwidth]{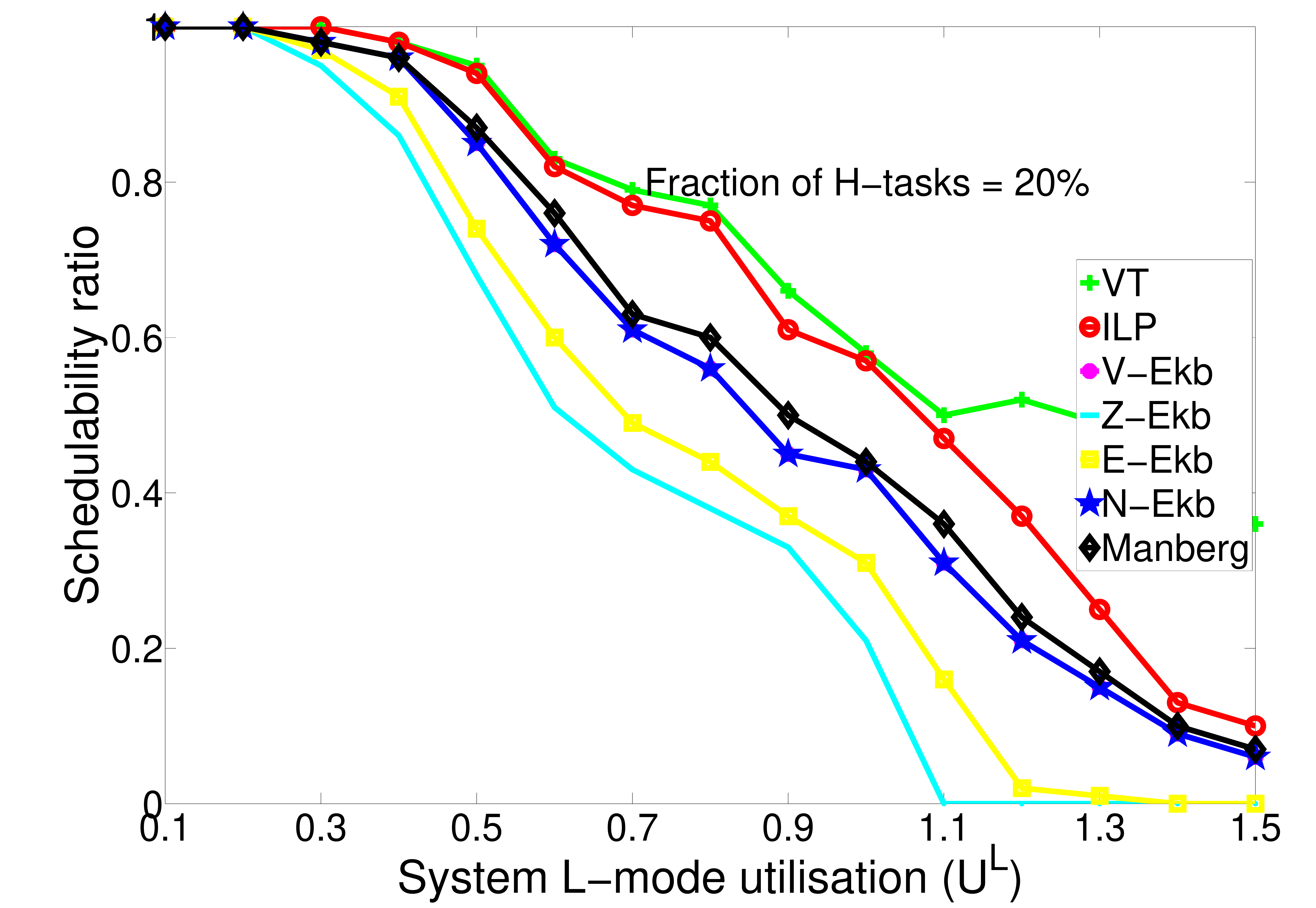}
  %\vspace{-3ex}
  \caption{\label{fig:exp:SR:25}}
\end{minipage}
\end{minipage}
\end{figure*}

\begin{figure*}
\begin{minipage}{1\linewidth}
\centering
\begin{minipage}[b]{0.49\linewidth}
\centering
\includegraphics[width=0.99\columnwidth]{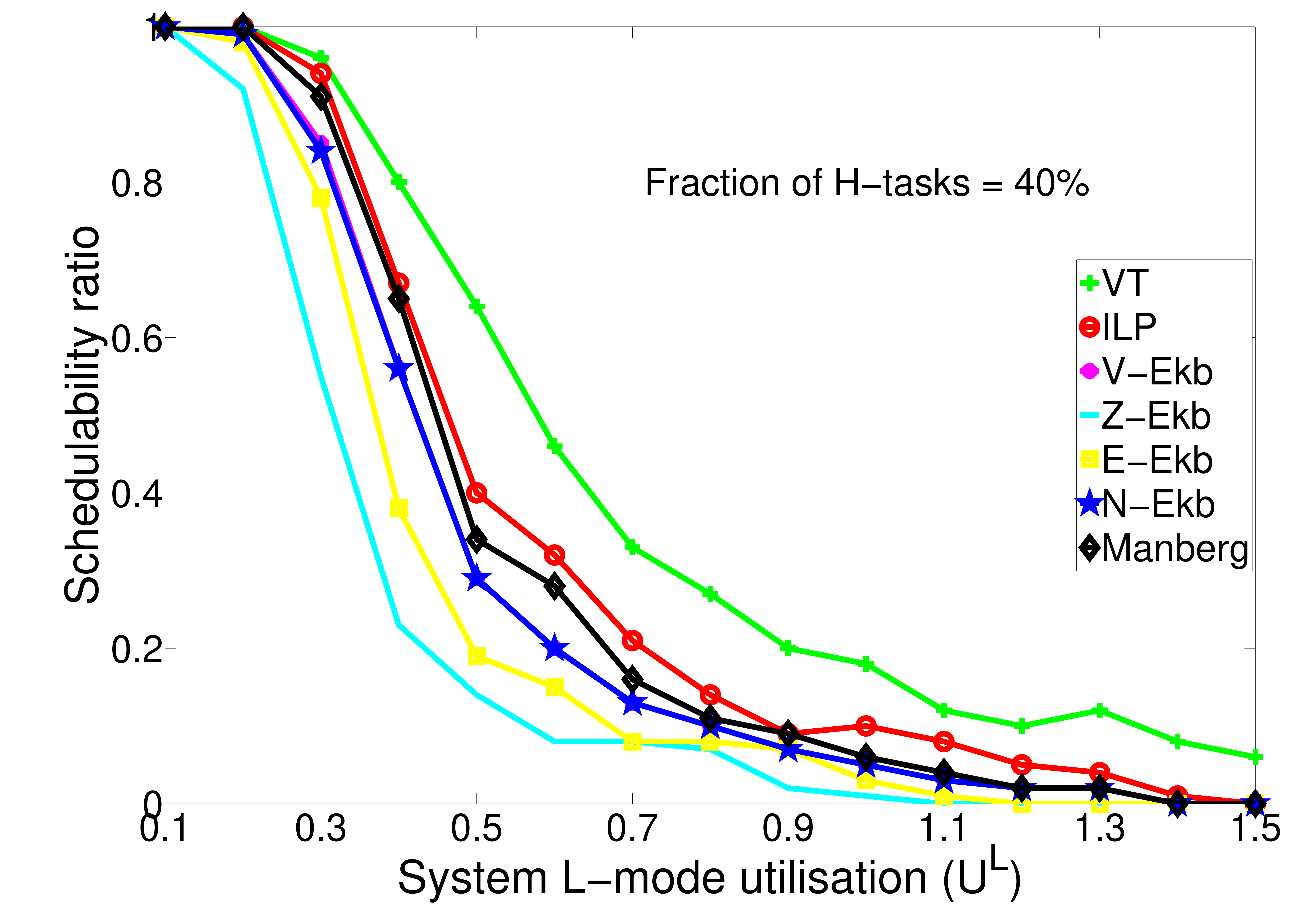}
  %\vspace{-3ex}
  \caption{\label{fig:exp:SR:26}}
  \vspace{3ex}
\end{minipage}
\hspace{-3ex}
\begin{minipage}[b]{0.49\linewidth}
\centering
\includegraphics[width=0.99\columnwidth]{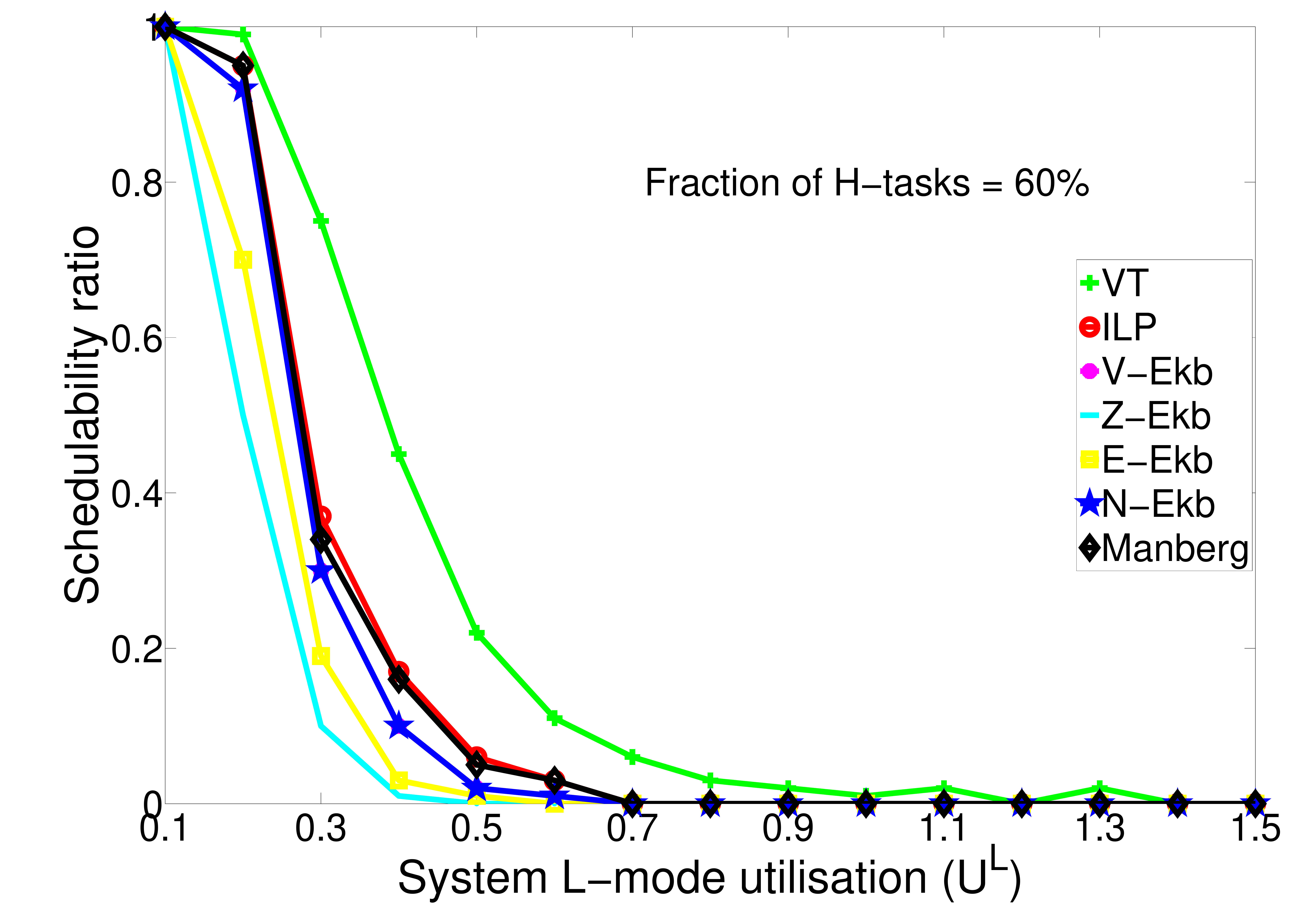}
  %\vspace{-3ex}
  \caption{\label{fig:exp:SR:27}}
  \vspace{3ex}
\end{minipage}
\end{minipage}
\end{figure*}

\begin{figure*}
\begin{minipage}{1\linewidth}
\centering
\begin{minipage}[b]{0.49\linewidth}
\centering
\includegraphics[width=0.99\columnwidth]{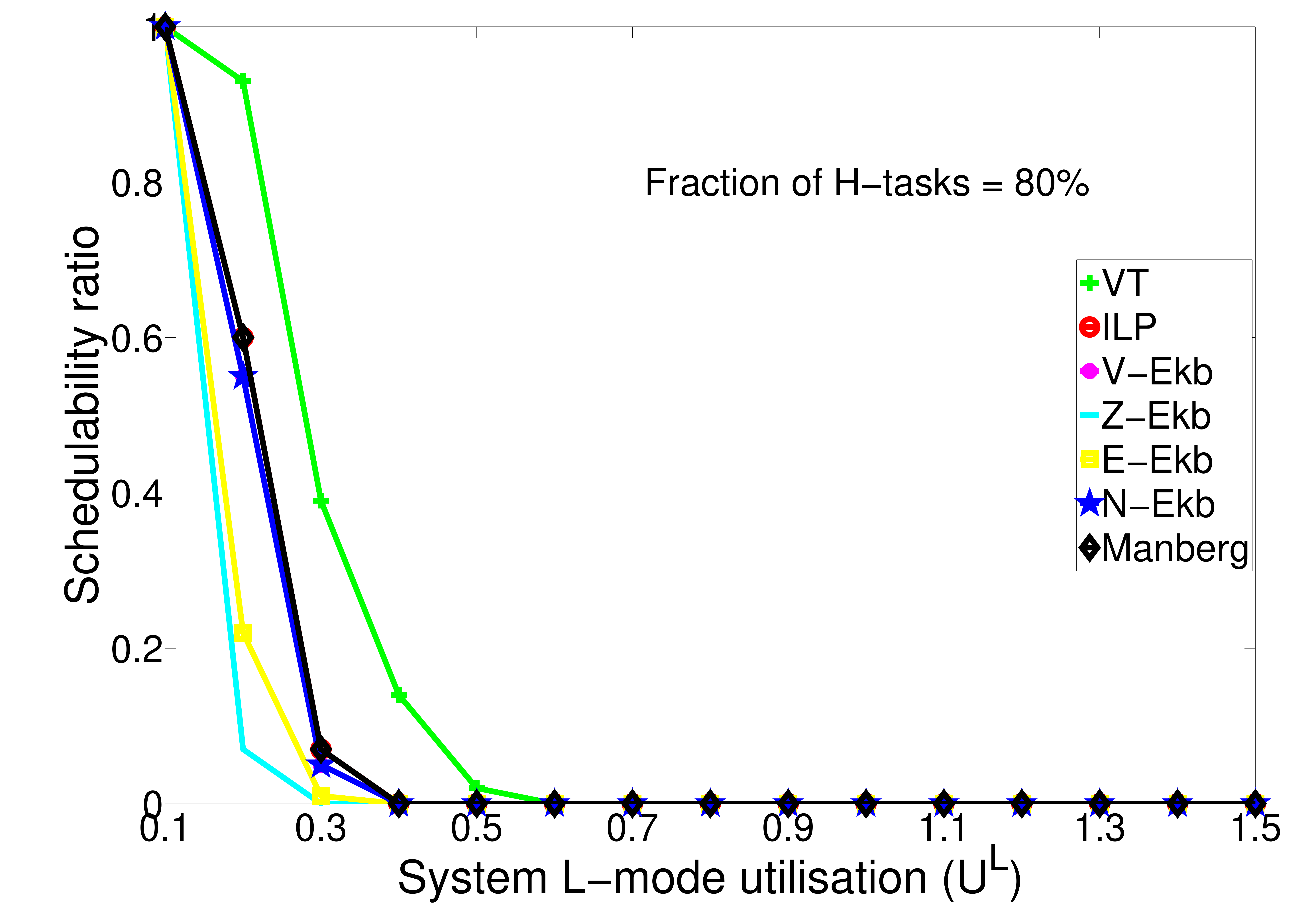}
  %\vspace{-3ex}
  \caption{\label{fig:exp:SR:28}}
\end{minipage}
\hspace{-3ex}
\begin{minipage}[b]{0.49\linewidth}
\centering
\includegraphics[width=0.99\columnwidth]{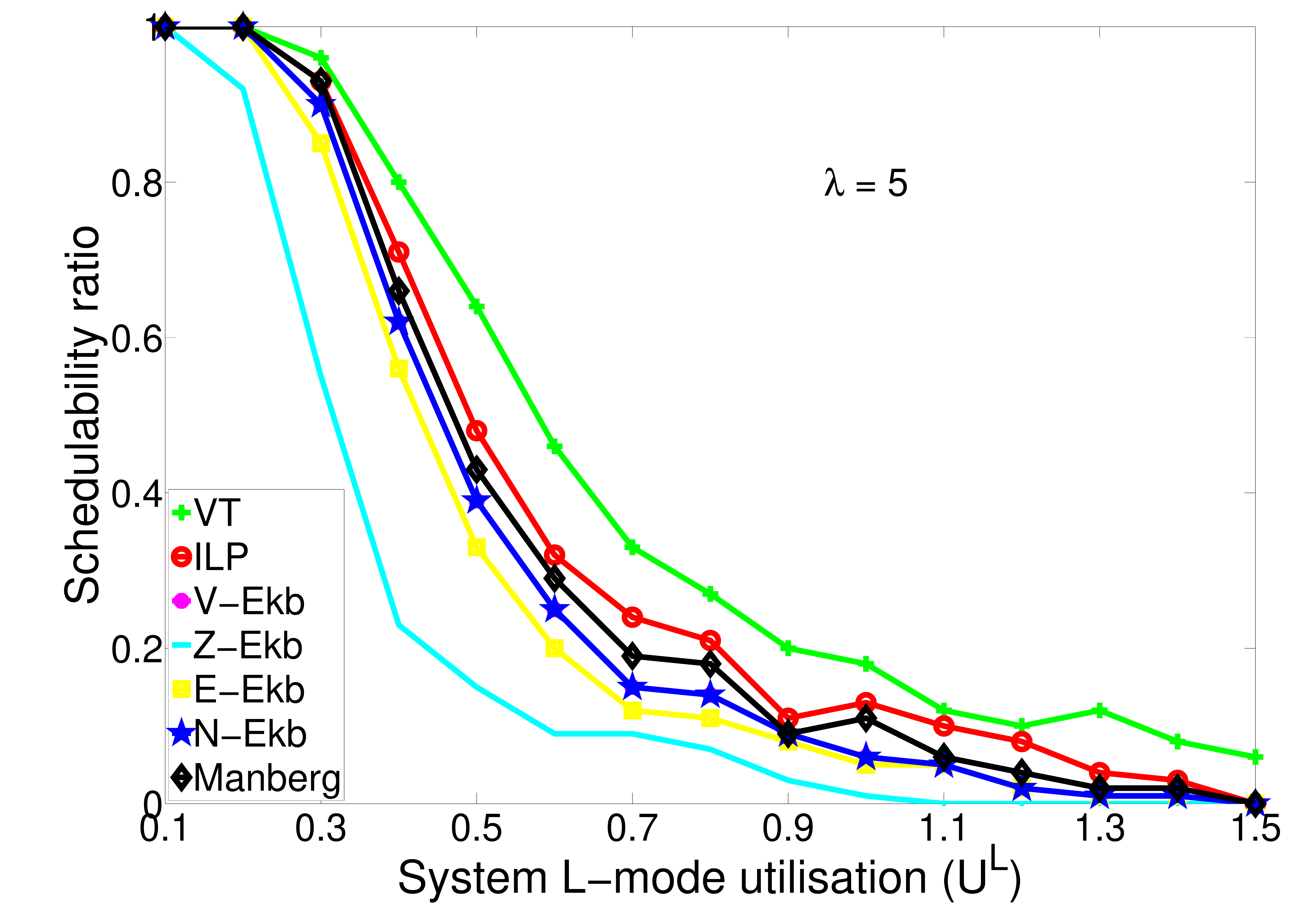}
  %\vspace{-3ex}
  \caption{\label{fig:exp:SR:29}}
\end{minipage}
\end{minipage}
\end{figure*}

\begin{figure*}
\begin{minipage}{1\linewidth}
\centering
\begin{minipage}[b]{0.49\linewidth}
\centering
\includegraphics[width=0.99\columnwidth]{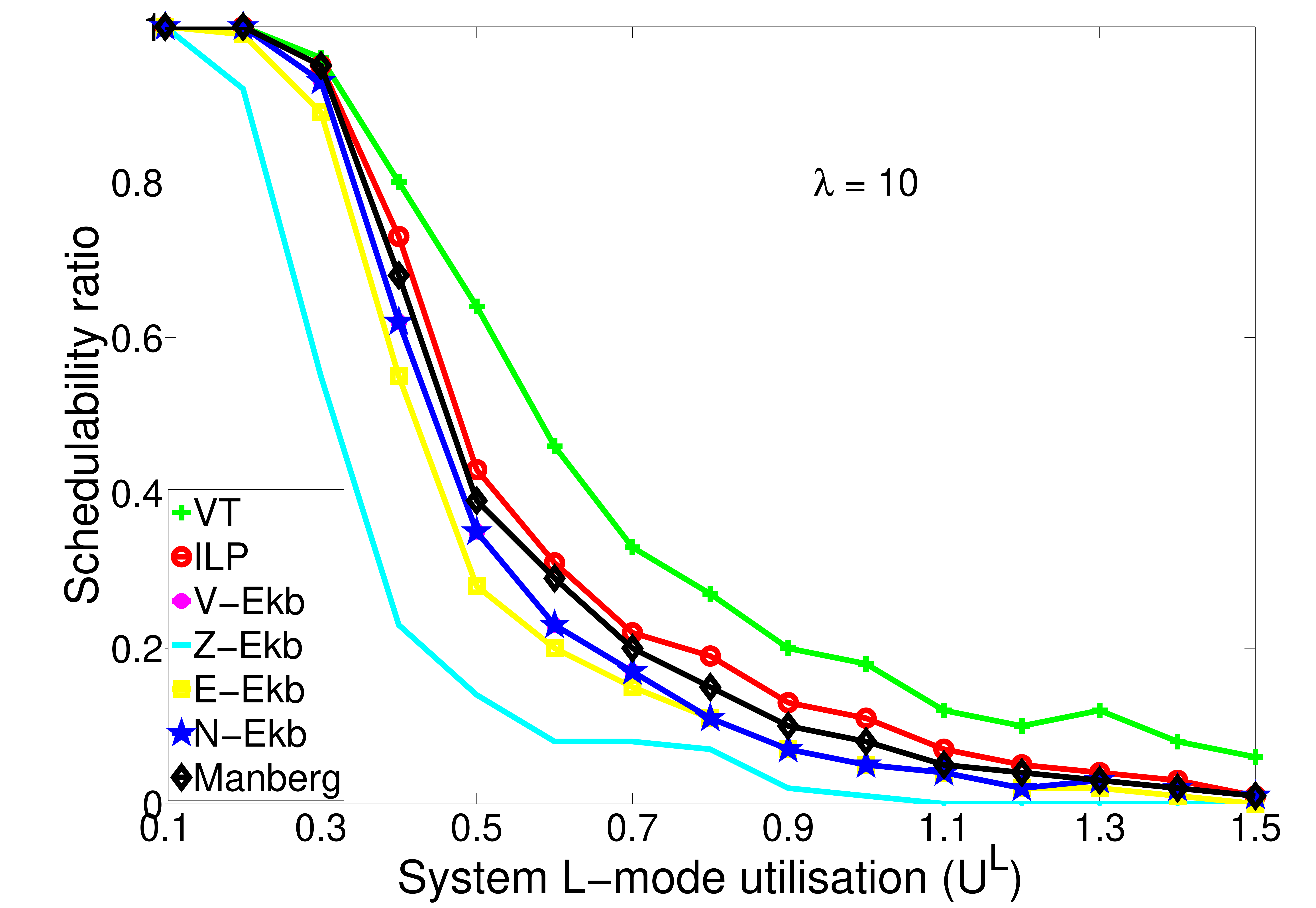}
  %\vspace{-3ex}
  \caption{\label{fig:exp:SR:30}}
\end{minipage}
\hspace{-3ex}
\begin{minipage}[b]{0.49\linewidth}
\centering
\includegraphics[width=0.99\columnwidth]{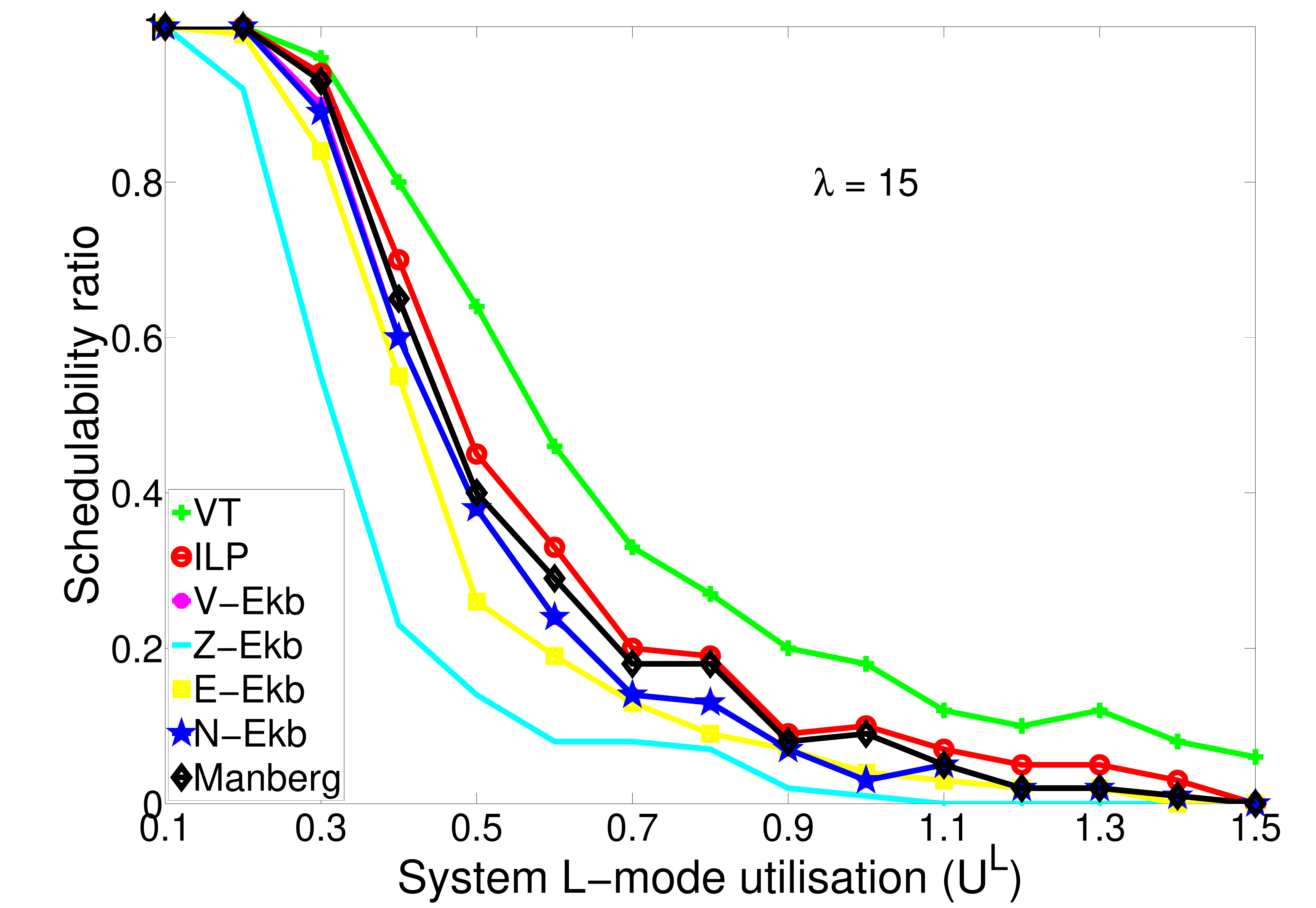}
  %\vspace{-3ex}
  \caption{\label{fig:exp:SR:31}}
\end{minipage}
\end{minipage}
%\vspace{-1ex}
\end{figure*}

%%% fourth set 
\begin{figure*}[!t]
\begin{minipage}{1\linewidth}
\centering
\begin{minipage}[b]{0.49\linewidth}
\centering
\includegraphics[width=0.99\columnwidth]{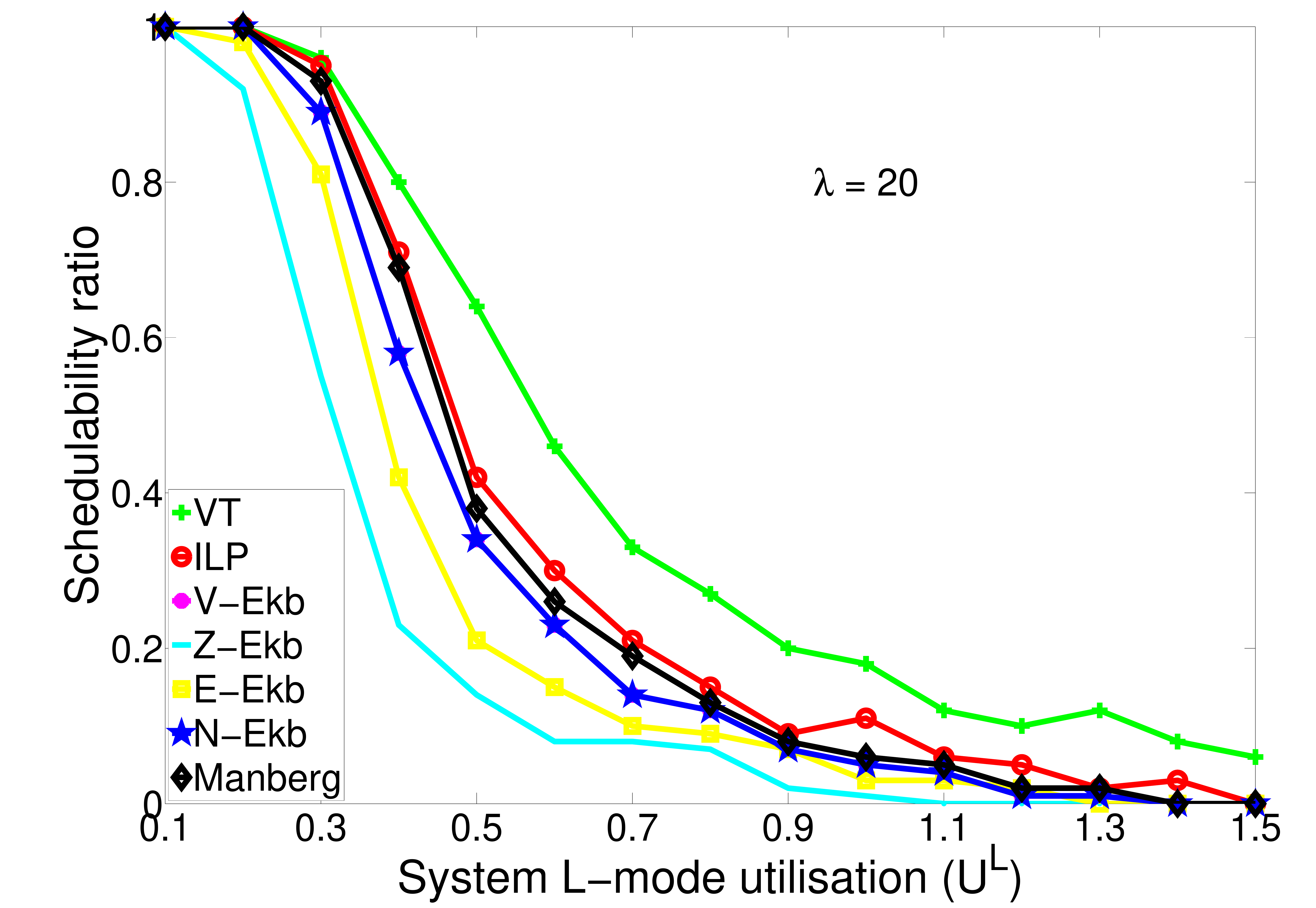}
  %\vspace{-3ex}
  \caption{\label{fig:exp:SR:32}}
  \vspace{3ex}
\end{minipage}
\hspace{-3ex}
\begin{minipage}[b]{0.49\linewidth}
\centering
\includegraphics[width=0.99\columnwidth]{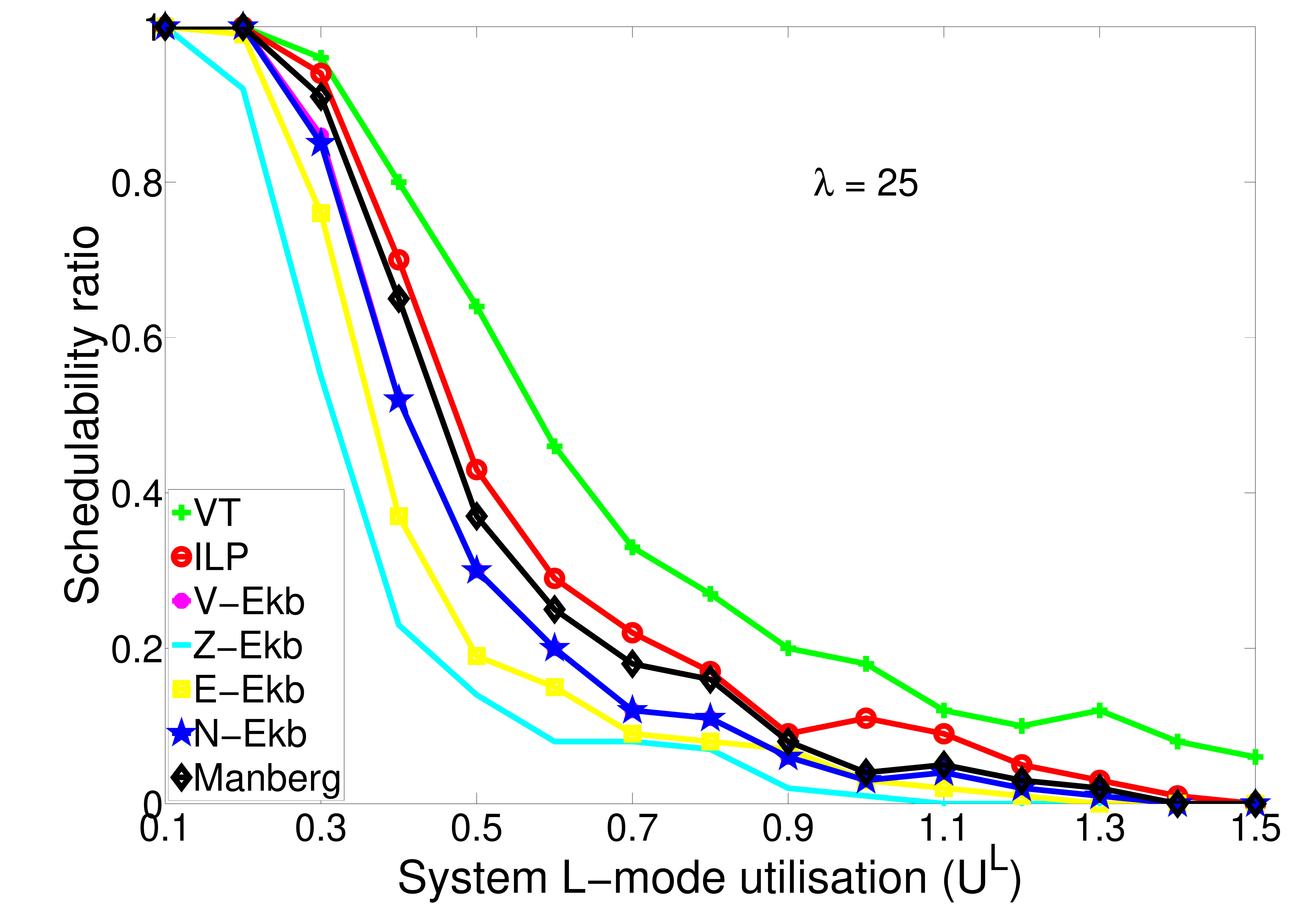}
  %\vspace{-3ex}
  \caption{\label{fig:exp:SR:33}}
  \vspace{3ex}
\end{minipage}
\end{minipage}
\end{figure*}

\begin{figure*}
\begin{minipage}{1\linewidth}
\centering
\begin{minipage}[b]{0.49\linewidth}
\centering
\includegraphics[width=0.99\columnwidth]{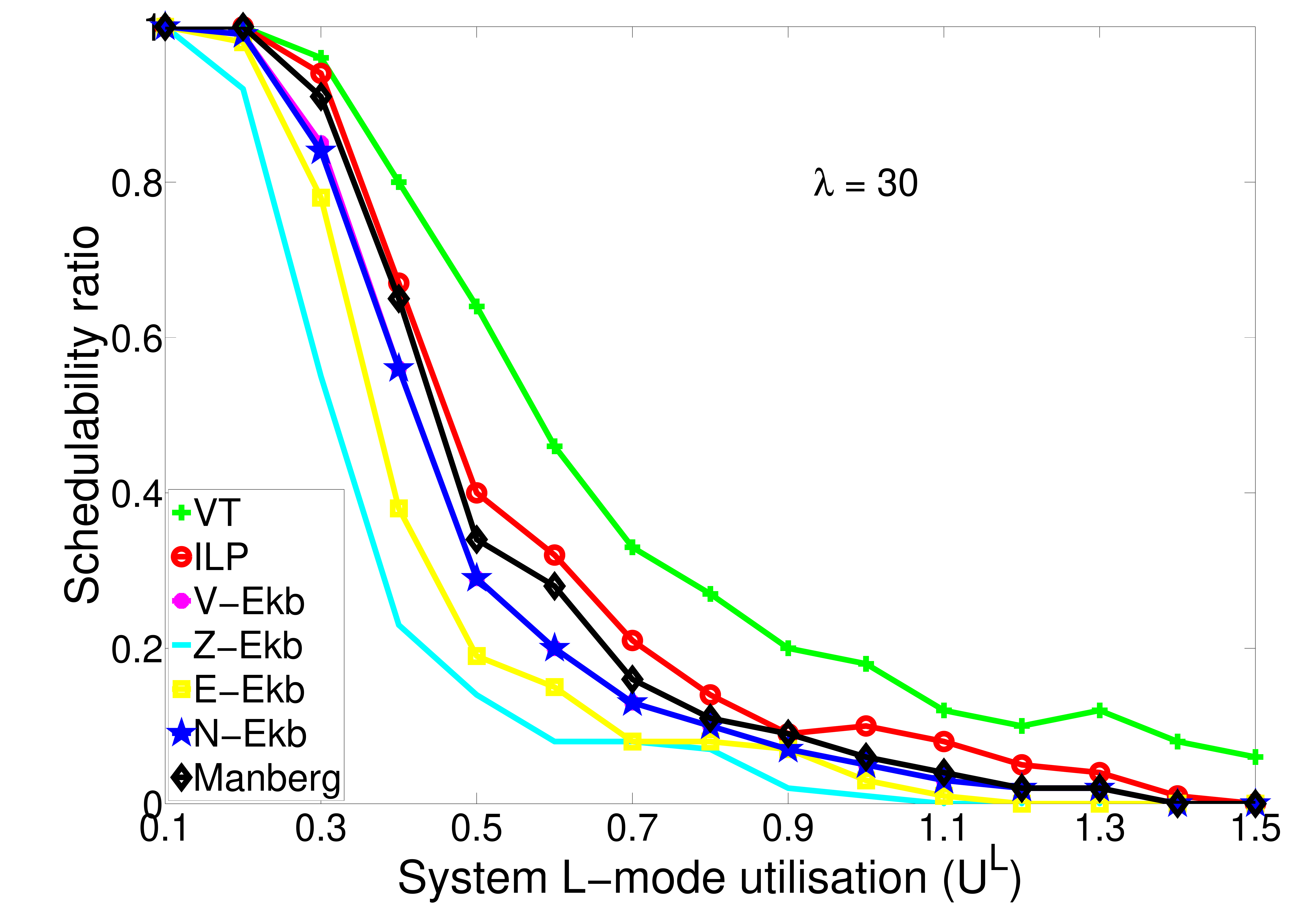}
  %\vspace{-3ex}
  \caption{\label{fig:exp:SR:34}}
\end{minipage}
\hspace{-3ex}
\begin{minipage}[b]{0.49\linewidth}
\centering
\includegraphics[width=0.99\columnwidth]{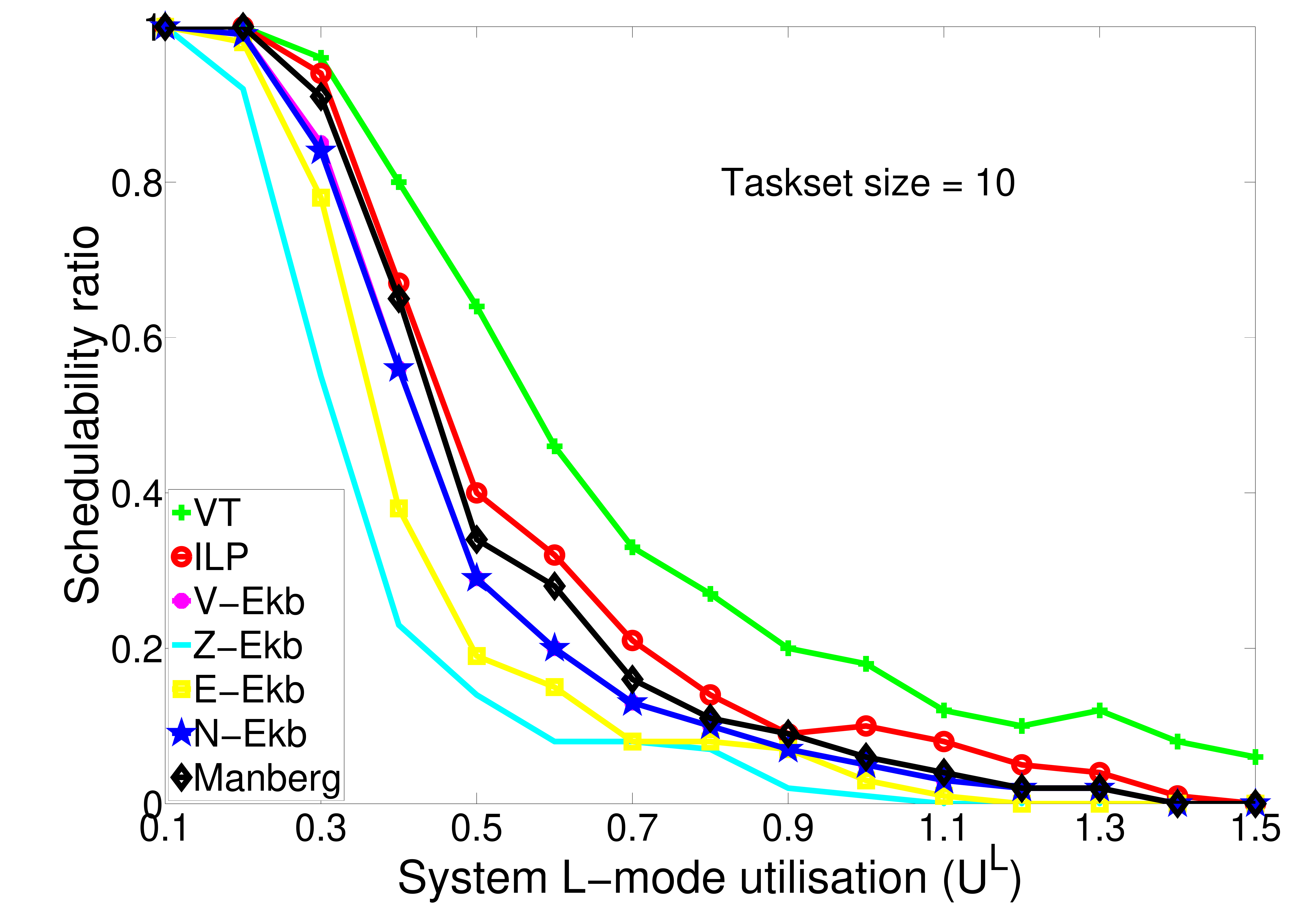}
  %\vspace{-3ex}
  \caption{\label{fig:exp:SR:35}}
\end{minipage}
\end{minipage}
\end{figure*}

\begin{figure*}
\begin{minipage}{1\linewidth}
\centering
\begin{minipage}[b]{0.49\linewidth}
\centering
\includegraphics[width=0.99\columnwidth]{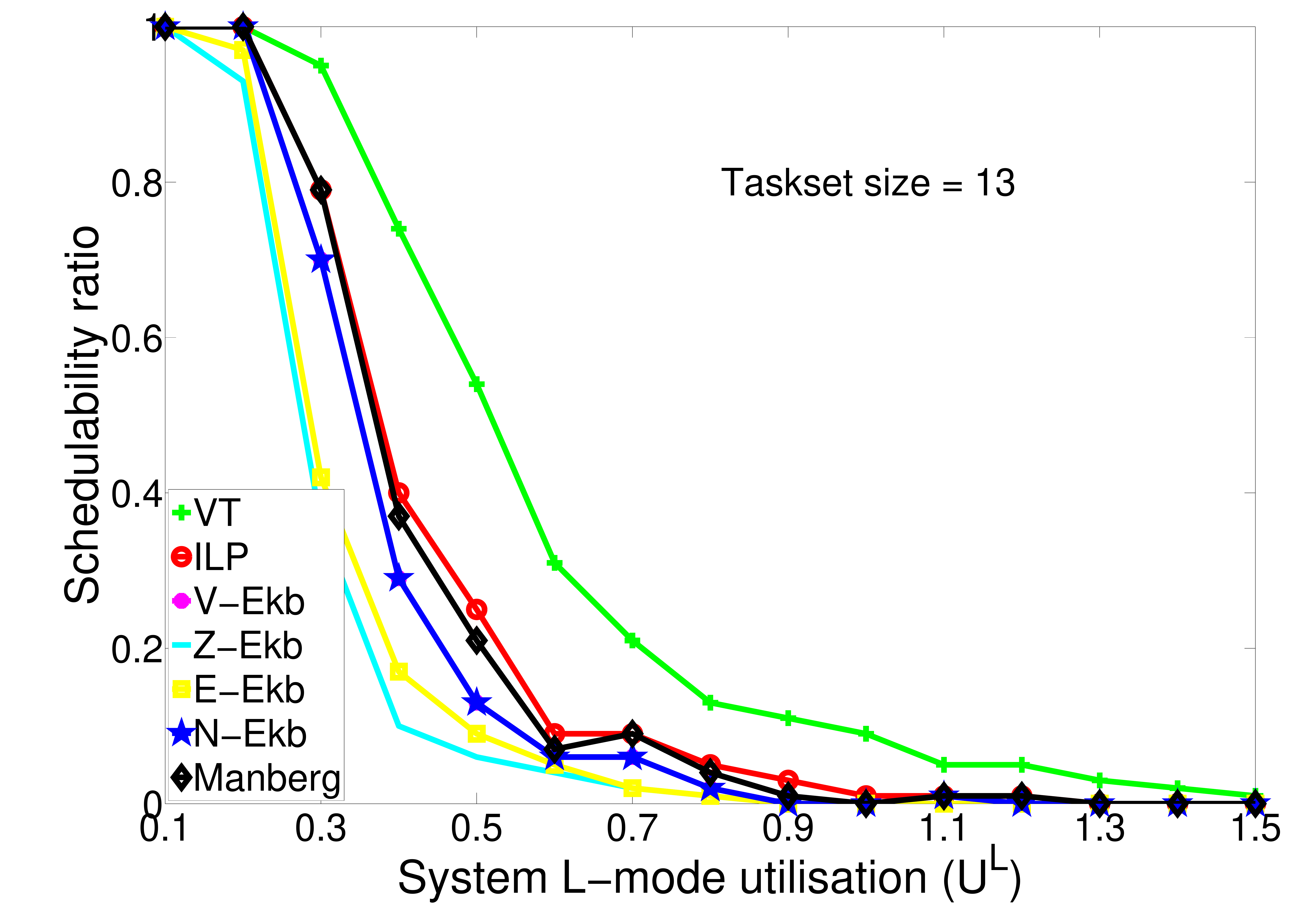}
  %\vspace{-3ex}
  \caption{\label{fig:exp:SR:36}}
\end{minipage}
\begin{minipage}[b]{0.49\linewidth}
\centering
\includegraphics[width=0.99\columnwidth]{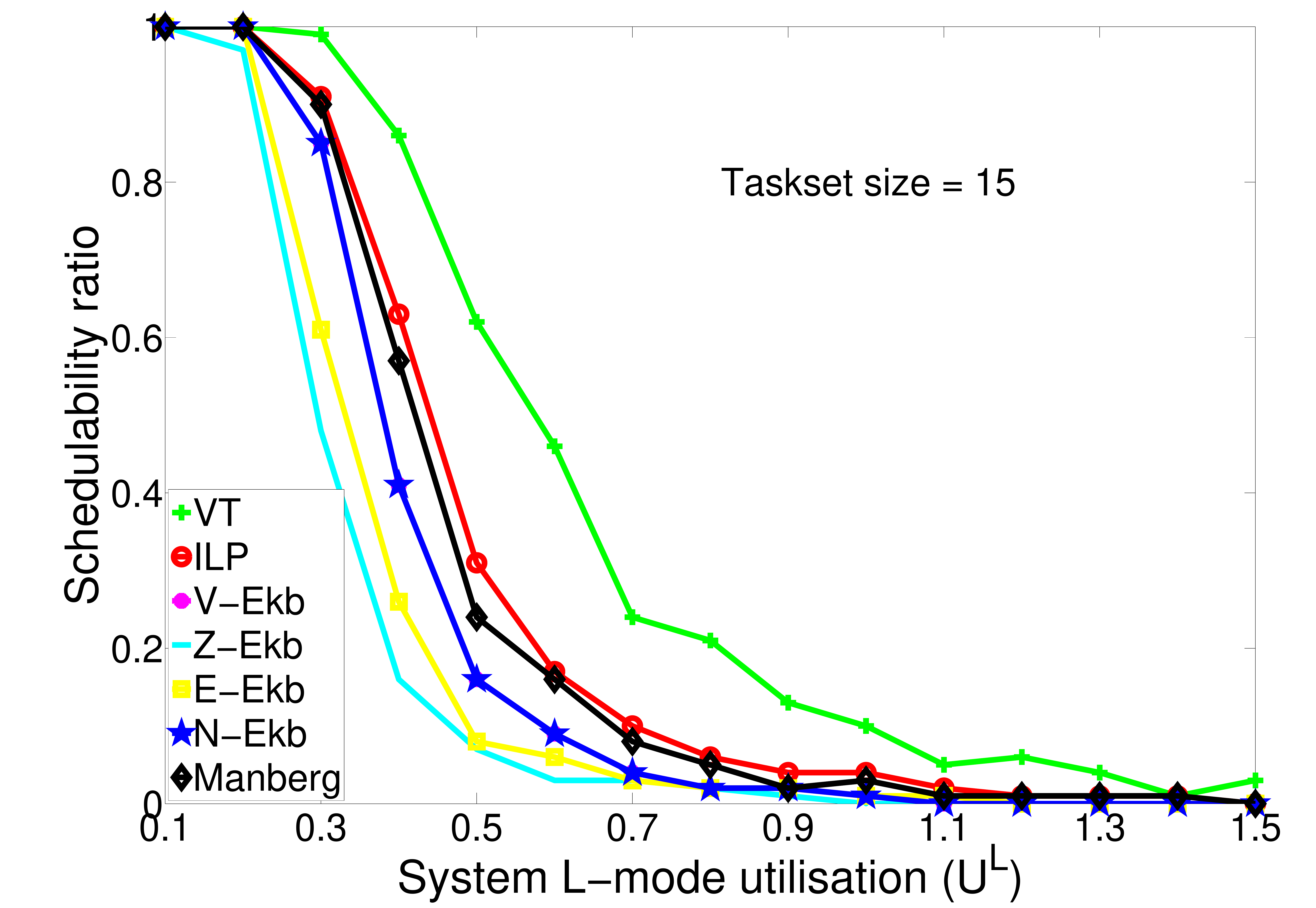}
  %\vspace{-3ex}
  \caption{\label{fig:exp:SR:37}}
\end{minipage}
\end{minipage}
\end{figure*}

\begin{figure*}
\begin{minipage}{1\linewidth}
\centering
\begin{minipage}[b]{0.49\linewidth}
\centering
 \includegraphics[width=0.99\columnwidth]{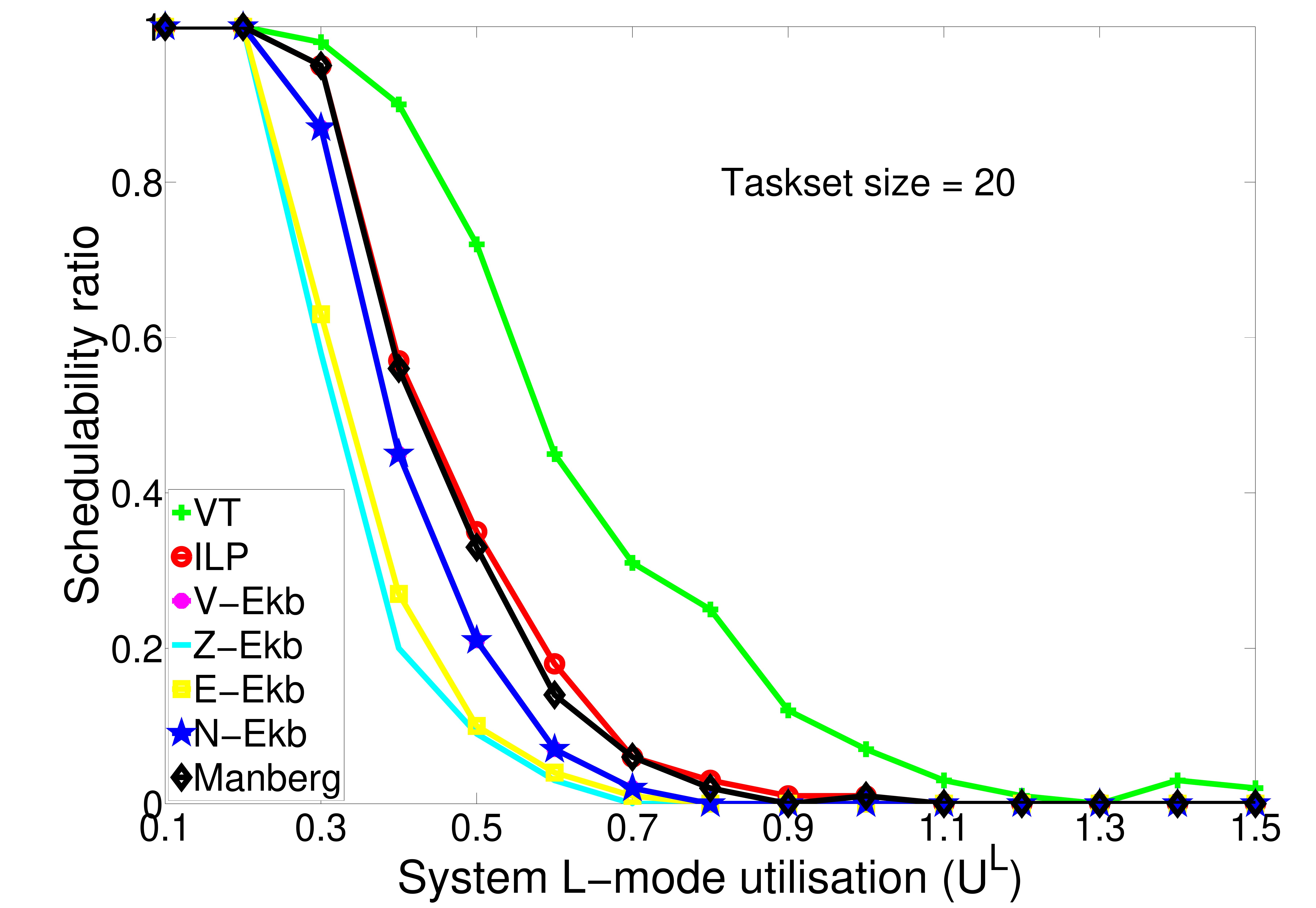}
  \caption{\label{fig:exp:SR:38}}
\end{minipage}
\end{minipage}
\end{figure*}

\end{document}